\documentclass{lmcs}
\pdfoutput=1

\usepackage{lastpage}
\renewcommand{\dOi}{14(4:11)2018}
\lmcsheading{}{1--\pageref{LastPage}}{}{}%
{Feb.~28,~2017}{Oct.~30,~2018}{}


\usepackage{amsfonts}
\usepackage{amsmath}
\usepackage{amssymb}
\usepackage{bussproofs}
\usepackage{hyperref}
\usepackage{stmaryrd}
\usepackage{tikz-cd}
\usepackage{enumitem}
\usepackage{array}
\usepackage{multicol}

\newcommand{\gcompE}[2]{#1 \mathop{@_{\mathcal{E}}}\nolimits #2}
\newcommand{\gimp}{\mathop{- \!\!\!\!\! - \!\!\! \blacktriangleright}}
\newcommand{\limp}{\mathop{\blacktriangleright \!\!\!\! - \!\!\!\!\! -}}
\newcommand{\satisfaction}[1]{ \models_{#1} }
\newcommand{\model}{ \mathcal{M} }
\newcommand{\Atom}[1]{\mbox{\rm #1}}
\newcommand{\subLabel}[1]{ \mathcal{S}(#1) }
\newcommand{\domain}[1]{ \mathcal{D}(#1) }
\newcommand{\alphabet}[1]{ \mathcal{A}(#1) }
\newcommand{\closure}[1]{ \overline{#1} }
\newcommand{\omegasubgraph}{\Omega}
\newcommand{\RLabel}[1]{\RightLabel{\begin{scriptsize}#1\end{scriptsize}}}
\newcommand{\labelledFormula}[3]{ \mathbb{#1} #2 : #3 }
\newcommand{\labelledFormulaIndex}[4]{ \mathbb{#1}_{#4} #2_{#4} : #3_{#4}}  
\newcommand{\CSS}[2]{ \langle #1, #2 \rangle }
\newcommand{\CSSfiniteInclusion}{ \subseteq_f }
\newcommand{\CSSinclusion}{ \subseteq }
\newcommand{\concatList}{ \oplus }
\newcommand{\realization}[1]{ \lfloor #1 \rfloor }
\newcommand{\valuation}{ \mathcal{V} }
\newcommand{\logicfont}[1]{#1}
\newcommand{\wand}{\mathbin{-\mkern-6mu*}}

\newcommand*{\myalign}[2]{\multicolumn{1}{#1}{#2}}

\begin{document}

\title{Intuitionistic Layered Graph Logic:\\ Semantics and Proof Theory}

\author{Simon Docherty}
\address{University College London, UK}
\email{simon.docherty.14@ucl.ac.uk} 

\author{David Pym}
\address{\vskip-7pt}
\email{d.pym@ucl.ac.uk} 

\begin{abstract}
	\noindent Models of complex systems are widely used in the physical and social 
	sciences, and the concept of layering, typically building upon graph-theoretic 
	structure, is a common feature. We describe an intuitionistic substructural 
	logic called ILGL that gives an account of layering. The logic is a bunched system, combining 
	the usual intuitionistic connectives, together with a non-commutative, non-associative 
	conjunction (used to capture layering) and its associated implications. We give soundness 
	and completeness theorems for a labelled tableaux system with respect to a Kripke   
	semantics on graphs. We then give an equivalent relational semantics, 
	itself proven equivalent to an algebraic semantics via a representation theorem. 
	We utilise this result in two ways. First, we prove decidability of the logic by showing the 
	finite embeddability property holds for the algebraic semantics. 
	Second, we prove a Stone-type duality theorem for the logic. By introducing the notions 
	of ILGL hyperdoctrine and indexed layered frame we are able to extend this result to a 
	predicate version of the logic and prove soundness and completeness theorems for an extension 	
	of the layered graph semantics . We indicate the utility of predicate ILGL with a resource-labelled 
	bigraph model.
\end{abstract}

\maketitle
\section{Introduction} \label{sec:introduction} 

Complex systems can be defined as the field of science that studies, on the one hand, 
how it is that the behaviour of a system, be it natural or synthetic, derives from the 
behaviours of its constituent parts and, on the other, how the system interacts with its 
environment. A commonly employed and highly effective concept that helps to manage 
the difficulty in conceptualizing and reasoning about complex systems is that of 
\emph{layering}: the system is considered to consist of a collection of interconnected 
layers each of which has a distinct, identifiable role in the system's operations. Layers 
can be informational or physical and both kinds may be present in a specific system.

Graphs provide a suitably abstract setting for a wide variety of modelling purposes, 
and layered graphs already form a component of many existing systems modelling approaches. For example, 
both social networks \cite{brodka} and transportation systems \cite{PhysRevLett.96.138701}, 
have been modelled by a form of layered graph in which multiple layers are given by
relations over a single set of nodes. A key feature of the TCP/IP conceptual model of 
communications on the Internet \cite{clark} is its separation into layers. This form of layering 
is not immediately represented in terms of graphs. However, the form of its information flows  
may be captured quite naturally using layered graphs \cite{CMP14}. Elsewhere layered graph 
models have been deployed to solve problems related to telecommunications networks \cite{Gouveia} 
and to aid the design of P2P systems for businesses \cite{Wang}. A bigraph \cite{milner} is a form of layered graph that superimposes a spatial \emph{place graph} of locations and a \emph{link graph} designating communication structure on a single set of nodes. Such graphs provide models of distributed systems and have been used to generalize process models like petri nets and the $\pi$-calculus. Similar ideas have also been used to give layered models of biological systems \cite{maus}. More generally, multi-layered networks have become ubiquitous in a range of complex system modelling approaches (see \cite{KAB14} for a survey).

In this paper, we give a formal definition of layered graph and provide a bunched logic for 
reasoning about layering. 
Bunched logics freely combine systems of different structural strengths. 
For example, the logic of bunched implications (BI) \cite{OP99,HP2003,POY,GMP05} 
combines intuitionistic propositional logic with multiplicative intuitionistic linear logic 
(MILL) \cite{Girard}. 

This kind of combination can be most clearly understood proof-theoretically. Consider the sequent calculus, 
a proof system that directly represents consequence. In a sequent $\Gamma \vdash \varphi$ in intuitionistic logic, 
the context is simply a finite sequence of formulae connected by an operation $\, , \,$ which admits the
structural properties of Exchange (E), Contraction (C), and Weakening (W).  In the bunched system BI, the context $\Gamma$ 
in a sequent $\Gamma \vdash \varphi$ is constructed as a finite tree, using two operations $\, ; \,$ and $\, ,\,$ with 
formulae at the leaves and $\, ; \,$ and $\, ,\,$ at the internal vertices. In this set-up, the semi-colon $\, ; \,$ admits 
Exchange, Contraction, and Weakening, but the comma $\, , \,$ admits only Exchange. Both are associative. 
As a result, we have in BI the  `deep' structural rules for $; \,$ 
\[
{\small \frac{\Gamma(\phi ; \psi) \vdash \chi}{\Gamma(\psi ; \phi) \vdash \chi} \quad E
	\qquad \frac{\Gamma(\psi ; \psi) \vdash \varphi}{\Gamma(\psi) \vdash \varphi} \quad C 
	\qquad \frac{\Gamma(\psi) \vdash \varphi}{\Gamma(\psi ; \chi) \vdash \varphi} \quad W}
\]
but do not have the corresponding $C$ and $W$ rules for the comma $\, ,\, $.

Corresponding to the two operations are additive and multiplicative conjunctions, $\wedge$ and $\ast$, each of 
which is accompanied by a corresponding implication, $\rightarrow$ and $\wand$, respectively. So we have rules 
such as
\[
{\small \frac{\Gamma (\psi_1 , \psi_2) \vdash \varphi}{\Gamma (\psi_1 \ast \psi_2) \vdash \varphi} \quad \mbox{\rm $\ast$ L} 
	\qquad \frac{\Gamma (\psi_1 ; \psi_2) \vdash \varphi}{\Gamma (\psi_1 \wedge \psi_2) \vdash \varphi} \quad \mbox{\rm $\wedge$ L}}
\]
and 
\[
{\small \frac{\Gamma , \psi \vdash \varphi}{\Gamma \vdash \psi \wand \varphi} \quad \mbox{\rm $\wand$ R} 
	\qquad \frac{\Gamma ; \psi \vdash \varphi}{\Gamma \vdash \psi \rightarrow \varphi} \quad \mbox{\rm $\rightarrow$ R}}
\]
and their right and left counterparts. 

The logic BI has a classical counterpart, called Boolean BI (BBI), in which the additives are classical, That is, BBI freely 
combines classical propositional logic and MILL. In \cite{CMP14,CMP15}, the logic LGL, or `layered graph logic', was introduced. 
Like BBI, LGL employs classical additives, but, unlike BBI, it employs a multiplicative conjunction that is neither 
commutative --- that is, does not admit Weakening --- nor associative. 

The key difference with LGL in the present work is that the additive component of the bunched logic we employ is \emph{intuitionistic}. Our logic is thus related to BI in precisely the same manner that LGL is related to BBI.

There are a number of reasons to investigate such a logic. Propositional intuitionistic logic was famously given a truth-functional semantics on ordered sets of possible worlds by Kripke \cite{Kripke}. Propositions are intuitionistically true if their classical truth value \emph{persists} with respect to the introduction of new facts. This is formally captured by the notion of \emph{persistent valuation}. A valuation $\mathcal{V}: \mathrm{Prop} \rightarrow \mathcal{P}(X)$ is persistent iff for all $x, y \in X$ and all $\mathrm{p} \in \mathrm{Prop}$: $x \in \mathcal{V}(p)$ and $x \preccurlyeq y$ implies $y \in \mathcal{V}(\Atom{p})$. For persistent $\mathcal{V}$, this property extends to the satisfaction relation of Kripke's semantics. Thus $\varphi$ is true at a world $x$ iff $\varphi$ is true at all worlds $y$ such that $x \preccurlyeq y$.

Directed graphs carry a natural order: the subgraph relation. An intuitionistic logic on graphs is therefore suitable for reasoning about properties that persist with respect to this order. One example relevant to systems modelling is the existence of a path satisfying desirable properties. This principle can be generalized to a chosen order on the subgraphs of a model. 

There are also technical motivations that lie behind the introduction of an intuitionstic variant of the logic, most prominently the question of completeness. In \cite{CMP14} an algebraic completeness result was given for a class of structures that included the layered graph models. Some of the methods in this paper also yield a  completeness result for a class of relational structures that closely resemble the layered graph models. However a completeness result specifically for the class of layered graph models eludes proof. This is not unexpected: completeness proofs typically capture a general classes 
 of models, of which the intended class of models is a subclass. In the case of \emph{intuitionistic} LGL, however, we are able to give a labelled tableaux system that outputs intuitionistic layered graph countermodels to invalid formulae, yielding completeness for the intended class of models. Moreover, we are able to combine this result with those for algebraic and relational structures to prove strong metatheoretic properties like the decidability of validity on layered graph models.

There are precursors in the literature for such a logic. The first, a spatial logic for querying graphs \cite{CGG02} also includes multiplicative connectives for reasoning about the decomposition of graphs into disjoint subgraphs. In this treatment directionality is not considered, with the consequence that the spatial decomposition cannot capture a notion of layering. As a result the multiplicative conjunction that captures decomposition is both commutative and associative, in contrast to our logic. It also differs in that the additive component of the logic is classical.
A closer relative is the logic BiLog, designed specifically for reasoning about bigraphs \cite{bilog}. In BiLog there are two sets of multiplicatives reflecting the two different ways bigraphs can be composed: the side-by-side composition of place graphs and the composition of link graphs at designated interfaces. As both compositions are order-sensitive, the associated connectives are non-commutative, however both are associative. Once again the additive connectives are classical. In both cases, very little metatheoretic work has been done beyond specification of semantics. In contrast we provide a comprehensive proof theoretic treatment of intuitionistic LGL and give an affirmative result on decidability. This is bolstered further by mathematically substantial results on topological duality for the logic and a sound and complete extension to predicate logic. Although we defer to another occasion a detailed investigation of the relationship between bigraphs and our logic, we indicate how bigraph-like structures provide an instance of our semantics.

The route map for this paper is as follows. In Section~\ref{sec:ILGL}, we introduce a formal definition of layered graph. This motivates the introduction of ordered scaffolds, the central semantic structure for the logic. This is followed by a specification of the syntax and semantics of Intuitionistic Layered Graph Logic (ILGL). We give two proof systems for the logic: a simple Hilbert-type system extending that of propositional intuitionistic logic and a labelled tableaux system. 

In Section~\ref{sec:metatheory}, we 
establish the basic metatheory of the logic: namely, the soundness and completeness of the layered graph semantics with respect to the labelled tableaux system.  This is facilitated by design choices in the labelling algebra that allow us to extract ordered scaffold countermodels for formulae lacking a tableaux proof. 

In Section~\ref{sec:decidability}, we turn towards decidability of the logic. Our proof proceeds in two stages. First, we define an algebraic semantics based on ILGL's Hilbert system and a relational semantics based on ILGL's layered graph semantics. Using the labelled tableaux system we are able to show that the relational semantics and layered graph semantics are equivalent. Further, by a representation theorem for the algebraic semantics, we are able to show equivalence with the relational semantics. With this established, we prove the finite model property holds for the algebraic semantics, and thus decidability holds for the logic. In Section~\ref{sec:duality} we extend the representation theorem to a Stone-type topological duality.

Finally, in Section~\ref{sec:predicate} we define predicate ILGL and prove soundness and completeness with respect to a Hilbert system for a semantics on indexed relational structures. As an example of how this might be used, we define an extension of the layered graph semantics for a specific signature of predicate ILGL, and show this semantics is an instantiation of the one proven complete. We finish the section with an extension of the duality theorem of Section~\ref{sec:duality}. In Section~\ref{sec:conclusions} we discuss directions for further work. 

This paper is an expanded version of \cite{DP16} and contains several new results. In particular the content of Sections~\ref{sec:decidability} - \ref{sec:predicate} is new. This includes decidability of the logic, the representation \& duality theorems for ILGL's algebraic semantics and the extension to predicate ILGL with its associated soundness, completeness and duality theorems. We also give full proofs of tableaux completeness that were previously only sketched.

While layered graphs are a key component of models of complex systems, other structure 
is also important. For example, in modelling the structure and dynamics of distributed 
systems --- in order, for example, to study security properties --- it is necessary to capture the architecture of system locations, their associated system resources, and the 
processes that describe how the system delivers its services. Tools for reasoning about the 
structure and behaviour of complex systems therefore must handle location. But they must also handle resource and process. Thus logics for layered graphs represent just a first step in establishing a logical account of complex systems modelling. A second step would be to 
reformulate the Hennessy-Milner-van~Bentham-style logics of state for location-resource-processes \cite{Pym09a,CMP2012,AP16} to incorporate layering. 

\section{Intuitionistic layered graph logic} \label{sec:ILGL} 

\subsection{The Layered Graph Construction}  \label{sec:graph-semantics} 

We begin with 
a formal, graph-theoretic account of the notion of layering that, we claim, captures the concept as used in 
complex systems. In this definition, two layers in a directed graph are connected by a specified set of edges, each 
element of which starts in the upper layer and ends in the lower layer. 

Given a directed graph, $\mathcal{G}$, we refer to its \emph{vertex set} by 
$V(\mathcal{G})$. Its edge set is given by a subset $E(\mathcal{G}) \subseteq V(\mathcal{G}) \times V(\mathcal{G})$, while its set of subgraphs is denoted 
$S\!g(\mathcal{G})$. Here, H is a subgraph of $\mathcal{G}$ iff $V(H) \subseteq V(\mathcal{G})$ and $E(H) \subseteq V(\mathcal{G})$. We thus overload set theoretic inclusion to also refer to the subgraph relation: $H \subseteq \mathcal{G} \text{ iff } H \in S\!g(\mathcal{G})$. For a 
distinguished edge set $\mathcal{E} \subseteq E(\mathcal{G})$, the reachability relation 
$\leadsto_{\mathcal{E}}$ on subgraphs of $\mathcal{G}$ is defined $H \leadsto_{\mathcal{E}} K$ iff  
there exist $u \in V(H)$ and $v \in V(K)$ such that $(u,v) \in \mathcal{E}$. 

This generates a partial composition $@_{\mathcal{E}}$ on subgraphs. Let $\downarrow$ denote definedness and $\uparrow$ denote undefinedness. For subgraphs $H$ and $K$, $\gcompE{H}{K}\downarrow \text{ iff } 
V(H) \cap V(K) = \emptyset, H \leadsto_{\mathcal{E}} K \text{ and } K \not \leadsto_{\mathcal{E}} H$  with output given by the graph union of the two subgraphs 
and the $\mathcal{E}$-edges between them. Formally, if $\gcompE{H}{K} \downarrow$, then $\gcompE{H}{K}$ is defined by  $V(\gcompE{H}{K}) = V(H) \cup V(K)$ and 
$E(\gcompE{H}{K}) = E(H) \cup E(K) \cup \{ (u, v) \mid u \in V(H), v \in V(K) \text{ and } (u, v) \in \mathcal{E} \}$.

For a graph $H$, we say it is \emph{layered} (with respect to $\mathcal{E}$) if there exist $K_0$, $K_1$ such that $\gcompE{K_0}{K_1}\downarrow$ and 
$H = \gcompE{K_0}{K_1}$ (see Figure~\ref{fig:layered-graph}). Layering is evidently neither commutative 
nor (because of definedness) associative:  for a full exposition of these properties see \cite{CMP14}.

\begin{figure}[!tbp]
  \centering
  \begin{minipage}[b]{0.51\textwidth}
    \includegraphics[scale=0.3]{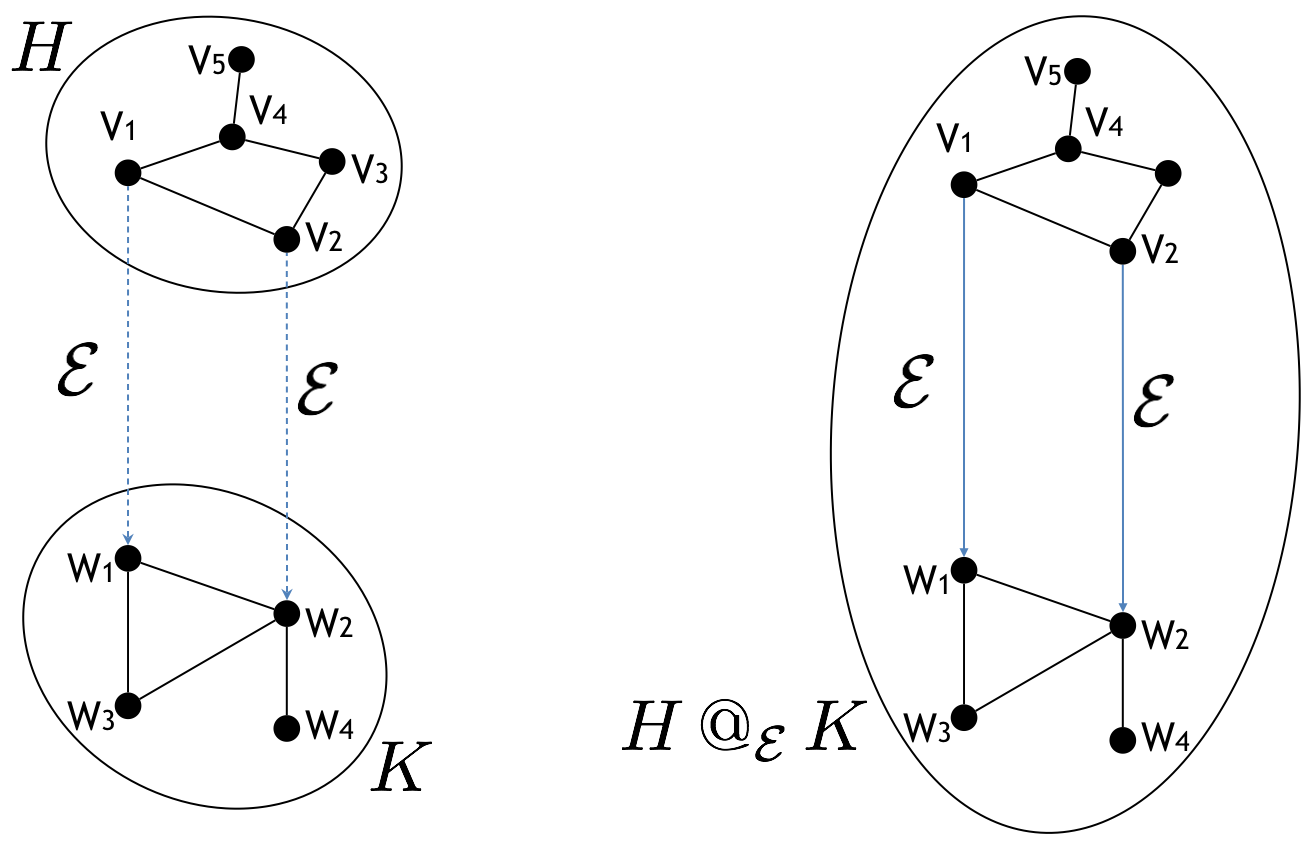}
    \caption{{\bf The graph composition $\gcompE{H}{K}$}}
    \label{fig:layered-graph}
  \end{minipage} \qquad
  \begin{minipage}[b]{0.4\textwidth}
    \includegraphics[scale=0.3]{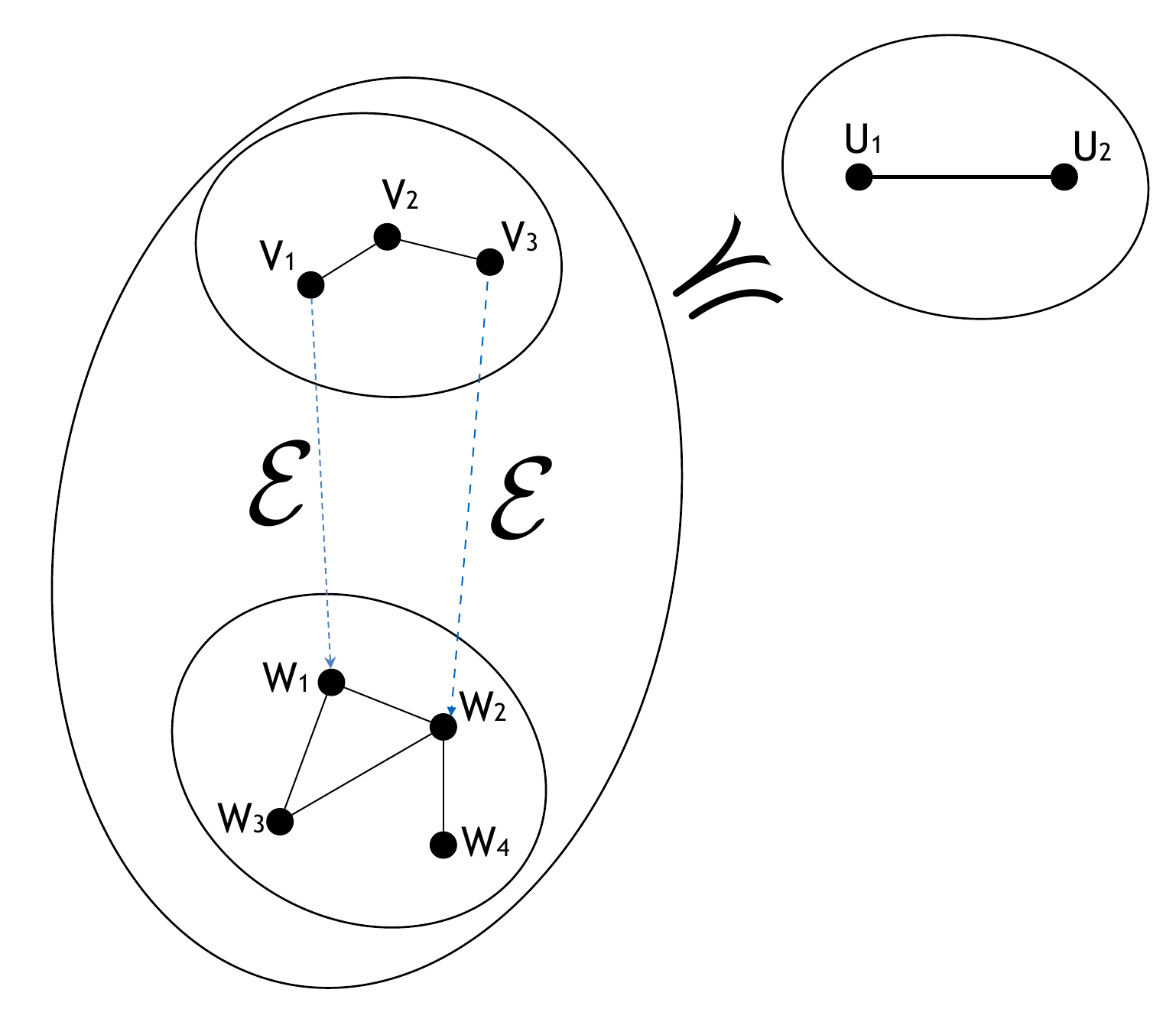}
    \caption{{\bf An ordered scaffold}}
    \label{fig:scaffold}
  \end{minipage}
\end{figure} 

Within a given graph, $\mathcal{G}$, we can identify a specific form of layered structure, 
called an \emph{ordered scaffold}, 
on which we interpret intuitionistic layered graph logic. To set this up, we begin with the definition of \emph{admissible subgraph set}, a subset $X \subseteq S\!g(\mathcal{G})$ such that, for all $H, K \in S\!g(\mathcal{G})$, if $\gcompE{H}{K} \!\!\downarrow$, 
then $H, K \in X$ iff $\gcompE{H}{K} \in X$. 

\begin{defi}[Ordered Scaffold] An  ordered scaffold is a structure $\mathcal{X} = (\mathcal{G}, \mathcal{E}, X, \preccurlyeq)$ such that $\mathcal{G}$ is a graph, $\mathcal{E} \subseteq E(\mathcal{G})$, $X$ an 
admissible subgraph set and $\preccurlyeq$ a preorder on $X$. Layers are present if $\gcompE{H}{K} \!\!\downarrow$ 
for at least one pair $H, K \in X$. 
\end{defi}

Figure~\ref{fig:scaffold} shows a simple example of an ordered scaffold. Note that the scaffold is preordered and we choose a subset of the subgraph set. This is a more general definition of scaffold than that taken in \cite{CMP14,CMP15}, where the 
structure was less tightly defined. 

There are several reasons for these choices. Properties of graphs that are 
inherited by their subgraphs are naturally captured in an intuitionistic logic. This idea is generalized by the 
preorder the ordered scaffold carries, structure that allows us to extend the Kripke interpretation of propositional intuitionistic logic to obtain our semantics for ILGL. We do not specify which preorder the scaffold carries as there are a number of natural choices. It also may be desirable to define more exotic orders for specific modelling situations.

\begin{figure}[!tbp]
  \centering
  \begin{minipage}[b]{0.51\textwidth}
    \includegraphics[scale=0.34]{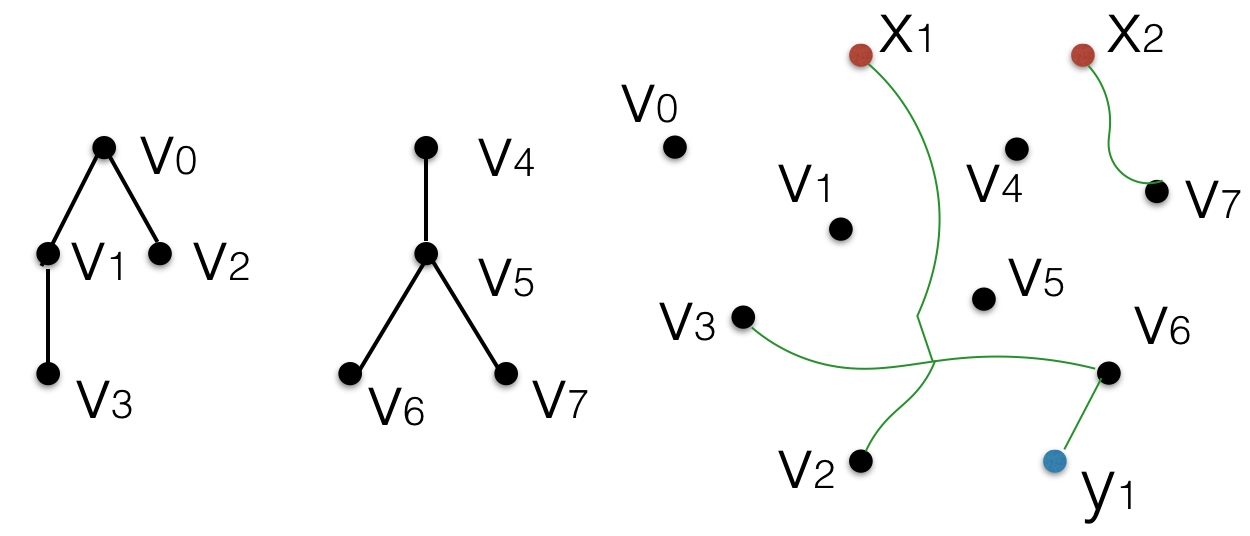}
    \caption{{\bf Place and link graphs}}
    \label{fig:placegraph}
  \end{minipage} \qquad
  \begin{minipage}[b]{0.4\textwidth}
    \includegraphics[scale=0.37]{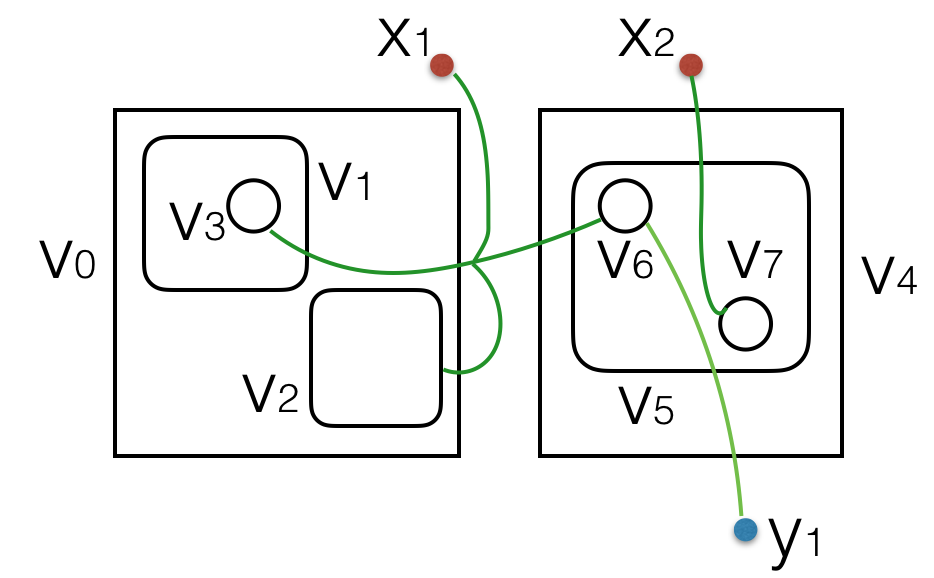}
    \caption{{\bf Bigraph}}
    \label{fig:bigraph}
 \end{minipage}
\end{figure} 

\begin{exa}[Bigraphs] \label{ex:bigraph}
A bigraph \cite{milner} is comprised of a set of nodes on which a \emph{place graph} and a \emph{link graph} are defined. The place graph has the structure of a disjoint union of trees (a forest), whilst the link graph is a hypergraph on which one edge can connect many nodes. Intuitively, the place graph denotes spatial relationships, whilst the link graph denotes the communication structure of the system. 

The link graph has additional structure: finite sets of labelled vertices $\{ x_1, \ldots, x_n \}$, $\{y_1, \ldots, y_m \}$  denoting \emph{inner names} and \emph{outer names} respectively. These act as interfaces to enable the composition of bigraphs: if the outer names of a bigraph match the inner names of another, their link graphs may be connected at these vertices. This compositional quality makes bigraphs ideal structures for modelling distributed systems. Bigraphical Reactive Systems (BRS) provide a dynamics for such models by defining transitions that reconfigure spatial relations and connectivity. Such systems generalize a wealth of process calculi, including $\pi$-calculi and the CCS.

Figure \ref{fig:bigraph} shows a bigraph and Figure \ref{fig:placegraph} its consituent parts. The structure of the place graph is visually realised in the bigraph by the containment of its nodes. We now show how a system of composed bigraphs can be encoded as an ordered scaffold. Given we work with directed graphs, we model a form of \emph{directed} bigraph \cite{GM07}. 

We begin with a single bigraph. First, consider the link graph. We can replace each hyperedge with a vertex attached to which we add an edge for each connection of the hyperedge. This obtains a directed graph with the same path information. Now note that a forest can straightforwardly be seen as a partial order on its vertices. This generates a partial order $\preccurlyeq$ on the set of subgraphs $\{ \{ v \} \mid v \text{ a vertex of the place graph} \}$. We extend $\preccurlyeq$ to the link graph $H$ by specifying that $H \preccurlyeq H$. 

Now we consider a system of composed bigraphs. Given bigraphs $(H, \preccurlyeq)$, $(K, \preccurlyeq')$ where $H$ has the same outer names as $K$'s inner names, we can connect the outer name vertices of $H$ to the inner name vertices of $K$ with new edges. We collect all such edges as $\mathcal{E}$. Thus the composition $\gcompE{H}{K}$ denotes the composition of the link graphs $(H, \preccurlyeq)$ and $(K, \preccurlyeq')$, and we can take the disjoint union of the partial orders to obtain a bigraph $(\gcompE{H}{K}, \preccurlyeq \sqcup \preccurlyeq')$. In this way we obtain an ordered scaffold with the admissible subgraph set given by the closure under composition of the set $\{ \{ v \} \mid v \text{ a vertex of a place graph } \}$ together with each link graph $H$, and order generated by the union of the partial orders defined by the place graphs of the system. \qed
\end{exa}

We choose an admissible subgraph set in order to reason more specifically about the layering structure of interest in the model, and to avoid degenerate cases of layers. For example, two disjoint subgraphs $H$ and $K$ may designate distinct, non-interacting regions in a systems model. However their disjoint union would be interpreted as layered over another subgraph $L$ if $H \leadsto_{\mathcal{E}} L$ but $K \not\leadsto_{\mathcal{E}} L$. The solution is to specify $H \cup K \not\in X$ for the ordered scaffold modelling the system.

There are further technical considerations behind this choice. When we restrict to interpreting ILGL on the full subgraph set, it is impossible to perform any 
composition of models without the states proliferating wildly. A similar issue arises during the 
construction of countermodels from the tableaux system of Section~\ref{sec:metatheory}, a procedure that 
breaks down when we are forced to take the full subgraph set as the set of states. 

\subsection{The logic ILGL}

Having established the layered graph construction and a semantic structure of interest, we now set up the logic ILGL. Let $\mathrm{Prop}$ be a set of atomic propositions, ranged over by $\Atom{p}$. The set $\mathrm{Form}$ of all propositional formulae is generated by the following grammar: \[
\varphi ::= \mathrm{p} \mid \top \mid \bot \mid \varphi \wedge \varphi \mid \varphi \vee \varphi \mid 
		\varphi \rightarrow \varphi \mid 
		\varphi \blacktriangleright \varphi \mid \varphi \gimp \varphi \mid \varphi \limp \varphi
\]

A Hilbert-type proof system for ILGL, $\mathrm{ILGL}_\mathrm{H}$, is given in 
Figure~\ref{fig:hilbert_rules_ILGL}. Intuitively, the judgement $\varphi \vdash \psi$ can be read as $\varphi$ implies $\psi$; or, if $\varphi$ is provable than so too is $\psi$. This intuition is made precise in Section \ref{sec:decidability}. where it is shown that $\vdash$ can be interpreted as the order of a Heyting algebra; it is straightforward to show that $a \leq b$ iff $a \rightarrow b = \top$ in a Heyting algebra.  We say $\varphi \vdash \psi$ is provable if it can be derived in the system; similarly, we say the formula $\varphi$ is provable if $\top \vdash \varphi$ is provable.

The additive fragment, corresponding to 
intuitionistic propositional logic, is standard (e.g., \cite{BdeJ2006}). The presentation of the multiplicative 
fragment is similar to that for BI's multiplicatives 
\cite{POY}, but for the non-commutative and non-associative (following from the absence of a 
multiplicative counterpart to $\wedge_2$) conjunction, $\blacktriangleright$, together with 
its associated left and right implications $\limp$ and $\gimp$ (cf. \cite{Lambek1961,Lambek1993}). 
The multiplicative conjunction is key to our logic, as it captures layering.

\begin{figure}
\centering
\hspace{-.2cm}
{\small 
\begin{tabular}{c}
\hline \\
	
     \AxiomC{}
\RightLabel{$(\mathrm{Ax})$}
\UnaryInfC{$\varphi \vdash \varphi$} 
      \DisplayProof 
      \quad
     \AxiomC{$\varphi \vdash \psi$}
\AxiomC{$\psi \vdash \chi$}
\RightLabel{$(\mathrm{Cut})$}
\BinaryInfC{$\varphi \vdash \chi$}
      \DisplayProof 
      \quad 
    \AxiomC{}
\RightLabel{$(\top)$}
\UnaryInfC{$\varphi \vdash \top$} 
      \DisplayProof 
      \quad 
     \AxiomC{}
\RightLabel{$(\bot)$}
\UnaryInfC{$\bot \vdash \varphi$} 
      \DisplayProof 
      
      \\ \\
      
            \AxiomC{$\varphi \vdash \psi$}
\AxiomC{$\varphi \vdash \chi$}
\RightLabel{$(\wedge_1)$}
\BinaryInfC{$\varphi \vdash \psi \wedge \chi$}  
      \DisplayProof 
      \quad
\AxiomC{}
\RightLabel{$(\wedge_2)$}
\UnaryInfC{$\varphi_1 \wedge \varphi_2 \vdash \varphi_i$}
      \DisplayProof 

	\\ \\

    \AxiomC{}
\RightLabel{$(\vee_1)$}
\UnaryInfC{$\varphi_i \vdash \varphi_1 \vee \varphi_2$} 
      \DisplayProof 
      \quad
      \AxiomC{$\varphi \vdash \chi$}
\AxiomC{$\psi \vdash \chi$}
\RightLabel{$(\vee_2)$}
\BinaryInfC{$\varphi \vee \psi \vdash \chi$}  
      \DisplayProof 
      
      \\ \\

\AxiomC{$\varphi \vdash \psi \rightarrow \chi$}
\AxiomC{$\nu \vdash \psi$}
\RightLabel{$(\rightarrow_1)$}
\BinaryInfC{$\varphi \wedge \nu \vdash \chi$}
      \DisplayProof 
      \quad
 \AxiomC{$\varphi \wedge \psi \vdash \chi$}
\RightLabel{$(\rightarrow_2)$}
\UnaryInfC{$\varphi \vdash \psi \rightarrow \chi$}  
      \DisplayProof 
      \qquad 
\AxiomC{$\varphi \vdash \psi$}
\AxiomC{$\chi\vdash \nu$}
\RightLabel{$(\blacktriangleright)$}
\BinaryInfC{$\varphi \blacktriangleright \chi \vdash \psi \blacktriangleright \nu$}
      \DisplayProof 
     
      \\ \\
      
     \AxiomC{$\varphi \vdash \psi \gimp \chi$} \hspace{-2mm}
\AxiomC{$\nu \vdash \psi$}
\RightLabel{$(\gimp_1)$}
\BinaryInfC{$\varphi \blacktriangleright \nu \vdash \chi$}
      \DisplayProof 
     \,  
    \AxiomC{$\varphi \blacktriangleright \psi \vdash \chi$}
\RightLabel{$(\gimp_2)$}
\UnaryInfC{$\varphi \vdash \psi \gimp \chi$}  
      \DisplayProof 
       
	\\ \\

         \AxiomC{$\varphi \vdash \psi \limp \chi$} \hspace{-2mm}
\AxiomC{$\nu \vdash \psi$}
\RightLabel{$(\limp_1)$}
\BinaryInfC{$\nu \blacktriangleright \varphi \vdash \chi$}
      \DisplayProof 
     \,
\AxiomC{$\varphi \blacktriangleright \psi \vdash \chi$}
\RightLabel{$(\limp_2)$}
\UnaryInfC{$\psi \vdash \varphi \limp \chi$} 
      \DisplayProof 
      
      \\ \\
\hline 
\end{tabular}}
\caption{ {\bf Rules of the Hilbert system, $\mathrm{ILGL}_\mathrm{H}$} }
\label{fig:hilbert_rules_ILGL}
\end{figure}


\begin{defi}[Layered graph model] \label{def:layered-graph-model}
A \emph{layered graph model} $\mathcal{M}=(\mathcal{X}, \mathcal{V})$ of $\mathrm{ILGL}$ is a pair 
where $\mathcal{X}$ is an ordered scaffold and $\mathcal{V}: \mathrm{Prop} \rightarrow \mathcal{P}(X)$ 
is a \emph{persistent} valuation; that is, 
$H \preccurlyeq K \text{ and } H \in \mathcal{V}(\mathrm{p}) \text{ implies } K \in \mathcal{V}(\mathrm{p})$. \hfill \qed
\end{defi}

Given a layered graph model $\mathcal{M} = (\mathcal{X}, \mathcal{V})$, we generate the satisfaction 
relation $\satisfaction{\model} \subseteq X \times \mathrm{Form}$ 
as shown in Figure \ref{fig:graphsat}. 
\begin{figure}
\centering
\hrule
\vspace{1mm}
\setlength\tabcolsep{1.5pt}
\setlength\extrarowheight{2pt}
\begin{tabular}{c c c c l r c c c c r r c}
$H$ & $\vDash_{\model}$ & $\mathrm{p}$ & iff & \myalign{l}{$H \in \mathcal{V}(\mathrm{p})$} & & & &
$H$ & $\vDash_{\model}$ & $\top$ & 
 $H  \not\vDash_{\model} \bot$  \\
$H$ & $\vDash_{\model}$ & $\varphi \land \psi$ & iff & \myalign{l}{$H \vDash_{\model} \varphi$ and $H \vDash_{\model} \psi$} \\
$H$ & $\vDash_{\model}$ & $\varphi \lor \psi$ & iff & $H \vDash_{\model} \varphi$ or  $H \vDash_{\model} \psi$
\\
$H$ & $\vDash_{\model}$ & $\varphi \rightarrow \psi$ & iff & \myalign{l}{for all $K \succcurlyeq H$, $K \vDash_{\model}\varphi$ implies $K$  $\vDash_{\model}$ $\psi$} \\
$H$ & $\vDash_{\model}$ & $\varphi \blacktriangleright \psi$ & iff & \multicolumn{8}{l}{there exists $\gcompE{K_0}{K_1}\downarrow$ s.t. $H \succcurlyeq \gcompE{K_0}{K_1}$, $K_0 \vDash_{\model} \varphi$ and  $K_1 \vDash_{\model} \psi$} \\
$H$ & $\vDash_{\model}$ & $\varphi \gimp \psi$ & iff & \multicolumn{8}{l}{for all $K$ and $L \succcurlyeq H$ s.t. $\gcompE{L}{K} \downarrow$:  $L \vDash_{\model} \varphi$ implies $\gcompE{L}{K} \vDash_{\model} \psi$} \\
$H$ & $\vDash_{\model}$ & $\varphi \limp \psi$ & iff & \multicolumn{8}{l}{for all $K$ and $L \succcurlyeq H$ s.t. $\gcompE{K}{L} \downarrow$: $L \vDash_{\model} \varphi$ implies $\gcompE{K}{L} \vDash_{\model} \psi$}
\end{tabular}
\caption{Satisfaction for layered graph models of ILGL.}
\vspace{1mm}
\hrule
\label{fig:graphsat}
\end{figure}

\begin{defi}[Validity] \label{def:validity}
The judgement $\varphi \vDash_{\model} \psi$ asserts that for every subgraph $H$ of the layered graph model $\model$, whenever $H \vDash_{\mathcal{M}} \varphi$, it follows that $H \vDash_{\mathcal{M}} \psi$. $\varphi \vDash \psi$ asserts that $\varphi \vDash_{\model} \psi$ holds for all layered graph models $\model$. $\varphi$ is \emph{valid in a model $\model$} iff $\top \vDash_{\model} \varphi$; $\varphi$ is \emph{valid} iff $\top \vDash \varphi$. \hfill \qed
%
\end{defi}

A straightforward inductive argument shows that persistence extends to all formulae for this semantics.

\begin{lem}[Persistence] \label{lem:persistence}
For all $\varphi \in \textrm{Form}$, $H \preccurlyeq K$ and $H \satisfaction{\model} \varphi$ 
implies $K \satisfaction{\model} \varphi$. \qed
\end{lem}

%
%
Persistence yields soundness of the Hilbert system for ILGL with respect to the layered graph semantics; we will return to the completeness of the Hilbert system in Section 4.

\begin{thm}[Soundness of (I)LGL]
$\varphi \vdash \psi$ is provable implies $\varphi \vDash \psi$. Similarly, $\varphi$ provable implies $\varphi$ is valid.
\end{thm}

\begin{proof}
Soundness follows by showing that validity of premisses leads to validity of the conclusion for each rule. We demonstrate with the rule $(\limp_1)$ \[ \quad \cfrac{\xi\vdash\varphi\limp\psi\quad\eta\vdash\varphi}{\eta \blacktriangleright\xi\vdash\psi}\] on layered graph models. Suppose $\xi \vDash \varphi \limp \psi$ and $\eta \vDash \varphi$. Let $H$ be an arbitrary subgraph in a layered graph model such that $H \vDash \eta \blacktriangleright \xi$. Then $H \succcurlyeq \gcompE{K_0}{K_1}$ with $K_0 \vDash \eta$ and $K_1 \vDash \xi$. By assumption $K_0 \vDash \varphi$ and $K_1 \vDash \varphi \limp \psi$ so $H \succcurlyeq \gcompE{K_0}{K_1}$ entails $x \vDash \psi$ by persistence, as required.
\end{proof}

Note that, unlike in BI's resource monoid semantics, we require the restriction `for all $K$, $H \preccurlyeq K \ldots $' in the semantic clauses 
for the multiplicative implications. Without this we cannot prove persistence, and with it completeness, because we cannot apply the inductive hypothesis in those cases. The reason for this is that we put no restriction on the interaction between 
$\preccurlyeq$ and $@$ in the definition of preordered scaffold. This is unlike the analogous case for BI, where 
the monoidal composition is required to be bifunctorial with respect to the ordering. One might resolve this issue 
with the following addendum to the definition of preordered scaffold: if $H \preccurlyeq K$ and $\gcompE{K}{L} \downarrow$, 
then $\gcompE{H}{L} \downarrow$ and $\gcompE{H}{L} \preccurlyeq \gcompE{K}{L}$.

Two natural examples of subgraph preorderings show that this would be undesirable. First, consider 
the layering preorder. Let $\preccurlyeq$ be the reflexive, transitive closure of the relation $R(H , K)$  iff  $\gcompE{K}{H}\downarrow$, 
restricted to the admissible subgraph set $X$. Figure~\ref{fig:reachability_ord} shows a subgraph $H$ with $L \preccurlyeq H$ 
and $\gcompE{H}{K}\downarrow$ but $\gcompE{L}{K} \uparrow$. Second, consider 
the subgraph order. In Figure~\ref{fig:subgraph_ord}, 
we have $L \subseteq H$ and $\gcompE{H}{K}\downarrow$ but $\gcompE{L}{K}\uparrow$. It is, however, the case 
that, with this ordering, if $H \subseteq K, \gcompE{K}{L} \downarrow$ and $\gcompE{H}{L}\downarrow$, 
then $\gcompE{H}{L} \subseteq \gcompE{K}{L}$. 

\begin{figure}[!tbp]
  
  \begin{minipage}[b]{0.5\textwidth}
\centering
    \includegraphics[scale=0.3]{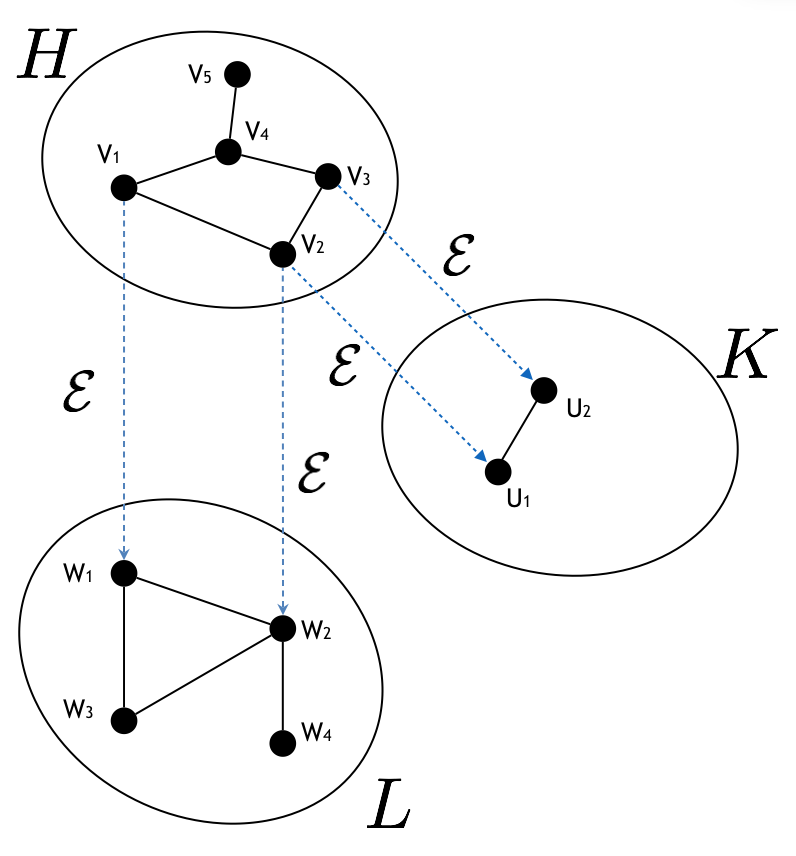}
    \caption{{\bf $\mathcal{E}$-reachability preorder}}
    \label{fig:reachability_ord}
  \end{minipage}
  \begin{minipage}[b]{0.4\textwidth}
	\centering
    \includegraphics[scale=0.3]{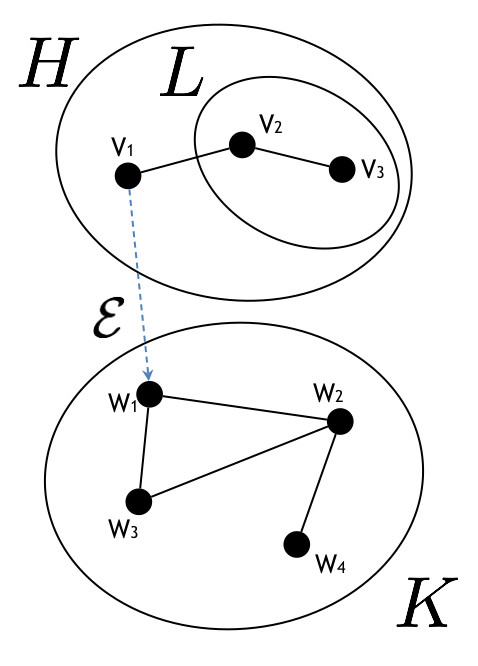}
    \caption{{\bf Subgraph order}}
    \label{fig:subgraph_ord}
  \end{minipage}
\end{figure} 

%
%

\subsection{Labelled tableaux}  \label{subsec:tableaux}
It's clear that obtaining completeness for the class of layered graph models with respect to the Hilbert system $\mathrm{ILGL}_{\mathrm{H}}$ will not be a straightforward task. For one, there does not appear to be any sensible way to 
augment equivalence classes of ILGL formulae with graph theoretic structure in such a way that $\mathrm{ILGL}_{\mathrm{H}}$-provability is reflected. A key reason for this is that the multiplicativity of $\blacktriangleright$ is not directly represented in $\mathrm{ILGL}_{\mathrm{H}}$.

We take an alternative route: we define a labelled tableaux system for ILGL, utilising a method first showcased on tableaux systems for BBI \cite{Lar14a} and DMBI \cite{Courtault-Galmiche15}, strongly influenced by previous work 
for BI \cite{GMP05}. By carefully restricting the labelling algebra of the proof system, we are able to uniformly transform the labels of a derivation into a layered graph model under the right conditions. This is similar in spirit to the labelled tableaux system for Separation Logic \cite{GM10}, which has a countermodel extraction procedure that outputs heap models, the intended models of the logic. It will be shown in Section~\ref{sec:decidability} that this captures a notion of provability equivalent to $\mathrm{ILGL}_{\mathrm{H}}$. 

\begin{defi}[Graph labels] \label{def:labels}
Let $\Sigma = \lbrace c_i \mid i \in \mathbb{N} \rbrace$ be a countable set of \emph{atomic labels}. 
We define the set  $\mathbb{L} = \lbrace x \in \Sigma^{\star} \mid 0 < \lvert x \rvert \leq 2 \rbrace \setminus 
\{ c_ic_i \mid c_i \in \Sigma \}$ to be the set of \emph{graph labels}. A \emph{sub-label} $y$ of a label $x$ 
is a non-empty sub-word of $x$, and we denote the set of sub-labels of $x$ by $\subLabel{x}$. \qed
\end{defi}

The graph labels are a syntactic representation of the subgraphs of a model, with labels of length $2$
representing a graph that can be decomposed into two layers. We exclude the possibility $c_ic_i$ as
layering is anti-reflexive. In much the same way we give a syntactic representation of preorder.

\begin{defi}[Constraints] \label{def:constraints} 
A \emph{constraint} is an expression of the form $x \preccurlyeq y$,
where $x$ and $y$ are graph labels. \qed
\end{defi}

Let $\mathcal{C}$ be a set of constraints. The \emph{domain} of $\mathcal{C}$, $\domain{\mathcal{C}}$, is the set of all 
sub-labels appearing in $\mathcal{C}$. In particular,  
$\domain{\mathcal{C}} = \bigcup_{x \preccurlyeq y \in \mathcal{C}} (\subLabel{x} \cup \subLabel{y})$
The \emph{alphabet} of $\mathcal{C}$ is the set of atomic labels appearing in $\mathcal{C}$. In particular,
we have $\alphabet{\mathcal{C}} = \Sigma \cap   \domain{\mathcal{C}}$. 

\begin{figure}[ht]
\hrule
\vspace{2mm}
\begin{center}
{\small 
\begin{tabular}{c} 
     \AxiomC{$x \preccurlyeq y$}
    \RLabel{$\langle \mathrm{R}_1 \rangle$}
    \UnaryInfC{$x \preccurlyeq x$} 
    \DisplayProof 
    \qquad
       \AxiomC{$x \preccurlyeq y$}
    \RLabel{$\langle \mathrm{R_2} \rangle$}
    \UnaryInfC{$y \preccurlyeq y$} 
    \DisplayProof 
    \qquad
    \AxiomC{$x \preccurlyeq yz$}
    \RLabel{$\langle \mathrm{R_3} \rangle$}
    \UnaryInfC{$y \preccurlyeq y$} 
    \DisplayProof 
  \qquad  
    \AxiomC{$x \preccurlyeq yz$}
    \RLabel{$\langle \mathrm{R_4} \rangle$}
    \UnaryInfC{$z \preccurlyeq z$} 
    \DisplayProof 
\\ \\
	    \AxiomC{$xy \preccurlyeq z$}
    \RLabel{$\langle \mathrm{R}_5 \rangle$}
    \UnaryInfC{$x \preccurlyeq x$} 
    \DisplayProof 
  \qquad
    \AxiomC{$xy \preccurlyeq z$}
    \RLabel{$\langle \mathrm{R}_6 \rangle$}
    \UnaryInfC{$y \preccurlyeq y$} 
    \DisplayProof 
    \qquad
    \AxiomC{$x \preccurlyeq y$}
    \AxiomC{$y \preccurlyeq z$}
    \RLabel{$\langle \mathrm{Tr} \rangle$}
    \BinaryInfC{$x \preccurlyeq z$} 
    \DisplayProof 
\end{tabular}}
\end{center} \vspace{-2mm}
\caption{ {\bf Rules for closure of constraints}} \vspace{2mm}
\label{fig_contraints_closure}
\hrule 
\end{figure}

\begin{defi}[Closure of constraints] \label{def:contraints_closure}
Let $\mathcal{C}$ be a set of constraints. The closure of $\mathcal{C}$, 
denoted $\closure{\mathcal{C}}$, is the least relation closed under
the rules of Figure \ref{fig_contraints_closure} such that
$\mathcal{C} \subseteq \closure{\mathcal{C}}$. \hfill \qed
\end{defi}

This closure yields a preorder on $\domain{\mathcal{C}}$, with $\langle \mathrm{R}_1 \rangle - \langle \mathrm{R}_6 \rangle$ 
generating reflexivity and $\langle \mathrm{Tr }\rangle$ yielding transitivity. Crucially, taking the closure of the 
constraint set does not cause labels to proliferate and the generation of any particular constraint from an 
arbitrary constraint set $\mathcal{C}$ is fundamentally a finite process.

\begin{prop} \label{prop:constraints}
Let $\mathcal{C}$ be a set of constraints. (1) $x \in \domain{\closure{\mathcal{C}}}$ iff $x \preccurlyeq x \in
\closure{\mathcal{C}}$.  (2) $\domain{\mathcal{C}} = \domain{\closure{\mathcal{C}}}$  and
$\alphabet{\mathcal{C}} = \alphabet{\closure{\mathcal{C}}}$. \hfill \qed 
\end{prop}
\begin{lem}[Compactness] \label{lem_compactness}
Let $\mathcal{C}$ be a (possibly countably infinite) set of constraints. If $x \preccurlyeq y \in \closure{\mathcal{C}}$, 
then there is a finite set of constraints $\mathcal{C}_f \subseteq \mathcal{C}$ such that 
$x \preccurlyeq y \in \closure{\mathcal{C}_f}$. \hfill \qed
\end{lem}

\begin{defi}[Labelled Formula / CSS] \label{def:labelled_formula_AND_CSS}
A \emph{labelled formula} is a triple $(\mathbb{S}, \varphi, x) \in \{\mathbb{T}, \mathbb{F}\} \times \mathrm{Form} \times \mathbb{L}$,  
written $\labelledFormula{S}{\varphi}{x}$. A \emph{constrained set of statements} (CSS) is a pair $\CSS{\mathcal{F}}{\mathcal{C}}$, where $\mathcal{F}$ is a set of labelled formulae and  $\mathcal{C}$ is
a set of constraints, satisfying the following properties: for all $x \in \mathbb{L}$ and distinct 
$c_i, c_j, c_k \in \Sigma$, 
\begin{enumerate}
\item $(\text{Ref })$ if $\labelledFormula{S}{\varphi}{x} \in
\mathcal{F}$, then $x \preccurlyeq x \in \closure{\mathcal{C}}$,
\item $(\text{Contra})$ if $c_ic_j \in \domain{\mathcal{C}}$, then $c_jc_i \not\in \domain{\mathcal{C}}$, and
\item  $(\text{Freshness})$ if $c_ic_j \in \domain{\mathcal{C}}$, then 
$c_ic_k, c_kc_i, c_jc_k, c_kc_j \not \in \domain{\mathcal{C}}$.  
\end{enumerate}
A CSS $\CSS{\mathcal{F}}{\mathcal{C}}$ is \emph{finite} if $\mathcal{F}$ and $\mathcal{C}$ are finite. 
The relation $\CSSinclusion$ is defined on CSSs by 
$\CSS{\mathcal{F}}{\mathcal{C}} \CSSinclusion
  \CSS{\mathcal{F}'}{\mathcal{C}'} \text{ iff } 
  \mathcal{F} \subseteq \mathcal{F}' \text{ and }  \mathcal{C}
  \subseteq \mathcal{C}'$. We denote by $\CSS{\mathcal{F}_f}{\mathcal{C}_f} \CSSfiniteInclusion
\CSS{\mathcal{F}}{\mathcal{C}}$ when $\CSS{\mathcal{F}_f}{\mathcal{C}_f} \CSSinclusion
\CSS{\mathcal{F}}{\mathcal{C}}$ holds and $\CSS{\mathcal{F}_f}{\mathcal{C}_f}$ is finite. \hfill \qed
\end{defi}

The CSS properties ensure models can be built from the labels: (Ref) ensures we have enough data for the 
closure rules to generate a preorder, (Contra) ensures the contra-commutativity of graph layering is respected, and 
(Freshness) ensures the layering structure of the models we construct is exactly that specified by the labels and 
constraints in the CSS. As with constraint closure, CSSs have a finite character.

\begin{prop} \label{prop_css_finite_css}
For any CSS $\CSS{\mathcal{F}_f}{\mathcal{C}}$ in which
$\mathcal{F}_f$ is finite, there exists $\mathcal{C}_f \subseteq
\mathcal{C}$ such that $\mathcal{C}_f$ is finite and
$\CSS{\mathcal{F}_f}{\mathcal{C}_f}$ is a CSS. \hfill \qed
\end{prop}

\begin{figure}
\hspace{-.2cm} 
{\small 
\setlength\tabcolsep{1pt}
\setlength\extrarowheight{10pt}
\begin{tabular}{cccc}
\hline
$\langle \mathbb{T} \wedge \rangle$ 
&
      \AxiomC{$\labelledFormula{\mathbb{T}}{\varphi \wedge \psi}{x}
        \in \mathcal{F}$} 
      \UnaryInfC{$\CSS{\{\labelledFormula{\mathbb{T}}{\varphi}{x},
          \labelledFormula{\mathbb{T}}{\psi}{x} \}}{\emptyset}$}  
      \DisplayProof
&
$\langle \mathbb{F} \wedge \rangle$
&
      \AxiomC{$\labelledFormula{\mathbb{F}}{\varphi \wedge \psi}{x}
        \in \mathcal{F}$} 
      \UnaryInfC{$\CSS{\{\labelledFormula{\mathbb{F}}{\varphi}{x}
          \}}{\emptyset} \ \mid \  
                  \CSS{\{ \labelledFormula{\mathbb{F}}{\psi}{x}
                    \}}{\emptyset}$}  
      \DisplayProof 
      \\
$\langle \mathbb{T}\vee \rangle$
&
      \AxiomC{$\labelledFormula{\mathbb{T}}{\varphi \vee \psi}{x}
        \in \mathcal{F}$} 
      \UnaryInfC{$\CSS{\{\labelledFormula{\mathbb{T}}{\varphi}{x}
          \}}{\emptyset} \ \mid \  
                  \CSS{\{\labelledFormula{\mathbb{T}}{\psi}{x}
                    \}}{\emptyset}$}  
      \DisplayProof
&
$\langle \mathbb{F}\vee \rangle$
&
      \AxiomC{$\labelledFormula{\mathbb{F}}{\varphi \vee \psi}{x}
        \in \mathcal{F}$} 
      \UnaryInfC{$\CSS{\{ \labelledFormula{\mathbb{F}}{\varphi}{x},
          \labelledFormula{\mathbb{F}}{\psi}{x} \}}{\emptyset}$}  
      \DisplayProof
	\\
$\langle \mathbb{T} \rightarrow \rangle$
&
      \AxiomC{$\labelledFormula{\mathbb{T}}{\varphi \rightarrow \psi}{x}
        \in \mathcal{F} \text{ and } x \preccurlyeq y \in \closure{\mathcal{C}}$} 
      \UnaryInfC{$\CSS{\{ \labelledFormula{\mathbb{F}}{\varphi}{y}
          \}}{\emptyset} \ \mid \  
                  \CSS{\{\labelledFormula{\mathbb{T}}{\psi}{y}
                    \}}{\emptyset}$}  
      \DisplayProof 
&
$\langle \mathbb{F} \rightarrow \rangle$
&
      \AxiomC{$\labelledFormula{\mathbb{F}}{\varphi \rightarrow \psi}{x}
        \in \mathcal{F}$} 
      \UnaryInfC{$\CSS{\{\labelledFormula{\mathbb{T}}{\varphi}{c_i},
          \labelledFormula{\mathbb{F}}{\psi}{c_i} \}}{\{x \preccurlyeq c_i \}}$}  
      \DisplayProof 
      \\
$\langle \mathbb{T}\blacktriangleright \rangle$
&
      \AxiomC{$\labelledFormula{\mathbb{T}}{\varphi \blacktriangleright \psi}{x}
        \in \mathcal{F}$} 
      \UnaryInfC{$\CSS{\{\labelledFormula{\mathbb{T}}{\varphi}{c_i},
          \labelledFormula{\mathbb{T}}{\psi}{c_j} \}}{\{ c_ic_j \preccurlyeq x \}}$}  
      \DisplayProof
&
$\langle \mathbb{F}\blacktriangleright \rangle$
&
      \AxiomC{$\labelledFormula{\mathbb{F}}{\varphi \blacktriangleright \psi}{x}
        \in \mathcal{F} \text{ and } yz \preccurlyeq x \in
        \closure{\mathcal{C}}$} 
      \UnaryInfC{$\CSS{\{\labelledFormula{\mathbb{F}}{\varphi}{y}
          \}}{\emptyset} \ \mid \ \CSS{\{
          \labelledFormula{\mathbb{F}}{\psi}{z} \}}{\emptyset}$}  
      \DisplayProof
	\\
$\langle \mathbb{T} \gimp \rangle$
&
      \AxiomC{$\labelledFormula{\mathbb{T}}{\varphi \gimp \psi}{x}
        \in \mathcal{F} \text{ and } x \preccurlyeq y, yz \preccurlyeq yz \! \in\!
        \closure{\mathcal{C}}$} 
      \UnaryInfC{$\CSS{\{\labelledFormula{\mathbb{F}}{\varphi}{z}
          \}}{\emptyset} \ \mid \ \CSS{\{
          \labelledFormula{\mathbb{T}}{\psi}{yz} \}}{\emptyset}$}  
      \DisplayProof
&
$\langle \mathbb{F} \gimp \rangle$
&
      \AxiomC{$\labelledFormula{\mathbb{F}}{\varphi \gimp \psi}{x}
        \in \mathcal{F}$} 
      \UnaryInfC{$\CSS{\{ \labelledFormula{\mathbb{T}}{\varphi}{c_j},
          \labelledFormula{\mathbb{F}}{\psi}{c_ic_j} \}}{\{ x \preccurlyeq c_i, c_ic_j \preccurlyeq c_ic_j \}}$}  
      \DisplayProof 
      \\
$\langle \mathbb{T} \limp \rangle$
     & 
       \AxiomC{$\labelledFormula{\mathbb{T}}{\varphi \limp \psi}{x}
        \in \mathcal{F} \text{ and } x \preccurlyeq y, zy \preccurlyeq zy \!\in \!
        \closure{\mathcal{C}}$} 
      \UnaryInfC{$\CSS{\{\labelledFormula{\mathbb{F}}{\varphi}{z}
          \}}{\emptyset} \ \mid \ \CSS{\{
          \labelledFormula{\mathbb{T}}{\psi}{zy} \}}{\emptyset}$}  
      \DisplayProof 
&
$\langle \mathbb{F} \limp \rangle$
&
      \AxiomC{$\labelledFormula{\mathbb{F}}{\varphi \limp \psi}{x}
        \in \mathcal{F}$} 
      \UnaryInfC{$\CSS{\{ \labelledFormula{\mathbb{T}}{\varphi}{c_j},
          \labelledFormula{\mathbb{F}}{\psi}{c_jc_i} \}}{\{ x \preccurlyeq c_i, c_jc_i \preccurlyeq c_jc_i \}}$}  
      \DisplayProof 
\\
\multicolumn{4}{c}{with $c_i$ and $c_j$ being fresh atomic labels.} \\
\end{tabular}}
\caption{{\bf Tableaux rules for ILGL}}
\label{fig_tableaux_ILGL_rules}
\vspace{2mm}
\hrule
\end{figure}

Figure \ref{fig_tableaux_ILGL_rules} presents the rules of the tableaux system for ILGL. 
That `$c_i$ and $c_j$ are fresh atomic labels' means $c_i \not = c_j \in \Sigma \setminus
\alphabet{\mathcal{C}}$. This means it is impossible to introduce the word $c_ic_i \not\in \mathbb{L}$ as a label. Note also that bunching is explicit in the labels, with concatenation of labels occurring in the rules for the multiplicative connectives $\blacktriangleright$, $\gimp$, $\limp$. This is analogous (and in fact equivalent) to the concatenation of contexts via the multiplicative conjunction's context former in a sequent calculus. This is in stark contrast to the Hilbert system $\mathrm{ILGL}_{\mathrm{H}}$ which outsources this structure to the metatheory.  

\begin{defi}[Tableaux] \label{def:tableaux} Let $\CSS{\mathcal{F}_0}{\mathcal{C}_0}$ be a 
finite CSS. A \emph{tableau} for this CSS is a list of CSS, called \emph{branches}, built inductively 
according to the following rules, where $\concatList$ denotes the concatenation of lists: 
\begin{enumerate}
\item The one branch list $[\CSS{\mathcal{F}_0}{\mathcal{C}_0}]$
         is a tableau for $\CSS{\mathcal{F}_0}{\mathcal{C}_0}$; 
\item If the list $\mathcal{T}_m \concatList
  [\CSS{\mathcal{F}}{\mathcal{C}}] \concatList \mathcal{T}_n$ is a
  tableau for $\CSS{\mathcal{F}_0}{\mathcal{C}_0}$ and 
\begin{center}
  \AxiomC{cond$\CSS{\mathcal{F}}{\mathcal{C}}$}
  \UnaryInfC{$\CSS{\mathcal{F}_1}{\mathcal{C}_1}$ $\mid$ \ldots $\mid$
    $\CSS{\mathcal{F}_k}{\mathcal{C}_k}$}  
  \DisplayProof 
\end{center}
is an instance of a rule of Figure \ref{fig_tableaux_ILGL_rules} for
which cond$\CSS{\mathcal{F}}{\mathcal{C}}$ is fulfilled, then the list
$\mathcal{T}_m \concatList [\CSS{\mathcal{F} \cup
    \mathcal{F}_1}{\mathcal{C} \cup \mathcal{C}_1}; \ldots; 
  \CSS{\mathcal{F} \cup \mathcal{F}_k}{\mathcal{C} \cup
    \mathcal{C}_k}] \concatList \mathcal{T}_n$  
is a tableau for $\CSS{\mathcal{F}_0}{\mathcal{C}_0}$.
\end{enumerate}
A \emph{tableau} for the formula $\varphi$ is a \emph{tableau} for 
$\CSS{\{\labelledFormula{\mathbb{F}}{\varphi}{c_0}\}}{\{ c_0 \preccurlyeq c_0 \}}$. \hfill \qed
\end{defi}

It is a simple but tedious exercise to show that the rules of Figure
\ref{fig_tableaux_ILGL_rules} preserve the CSS properties of Definition
\ref{def:labelled_formula_AND_CSS}. 

We now give the notion of proof for our labelled tableaux.

\begin{defi}[Closed tableau/proof] \label{def:closed_tableau} 
A CSS $\CSS{\mathcal{F}}{\mathcal{C}}$ is \emph{closed} if one of the
following conditions holds: (1) $\labelledFormula{\mathbb{T}}{\varphi}{x} \in \mathcal{F}$,
$\labelledFormula{\mathbb{F}}{\varphi}{y} \in \mathcal{F}$ and $x \preccurlyeq y \in \closure{\mathcal{C}}$;
(2) $\labelledFormula{\mathbb{F}}{\top}{x} \in \mathcal{F}$; and 
(3) $\labelledFormula{\mathbb{T}}{\bot}{x} \in \mathcal{F}$. 
A CSS is \emph{open} iff it is not closed. A tableau is closed iff all
its branches are closed. A \emph{proof} for a formula $\varphi$ is a closed tableau for $\varphi$. \hfill \qed
\end{defi}

Figure \ref{fig:tableaux} shows a tableau proof of $q \limp (q \blacktriangleright (p \rightarrow (p \lor q)))$ in tree form. Each branch gives the set $\mathcal{F}$ of labelled formulae of a CSS, with the rectangular boxes giving its set of constraints $\mathcal{C}$. $\surd_i$ denotes the application of a rule, with the circled side conditions showing which condition on the closure of the constraint set was used to allow it. Finally, the cross $\times$ marks the closure of a branch. 

The first rule application at $\surd_1$ is $\langle \mathbb{F}\limp \rangle$. This introduces fresh labels $c_1$ and $c_2$ with the constraints $c_0 \preccurlyeq c_1$ and $c_2c_1 \preccurlyeq c_2c_1$. Since we have the constraint $c_2c_1 \preccurlyeq c_2c_1$ we are able to apply $\langle \mathbb{F}\blacktriangleright \rangle$ at $\surd_2$. The tableau then branches. On the left hand branch we have both $\mathbb{T}q:c_2$ and $\mathbb{F}q:c_2$, thus it is closed; on the right we apply $\langle \mathbb{F}\rightarrow \rangle$ at $\surd_3$. This introduces a fresh label $c_3$ with the constraint $c_1 \preccurlyeq c_3$ and introduces $\mathbb{T}p:c_3$ to the branch. We then apply $\langle \mathbb{F}\lor \rangle$ at $\surd_4$, introducing $\mathbb{F}p:c_3$, thus closing the branch. As both branches are closed, the tableau is closed and the formula is proved.

\begin{figure}
\centering
\includegraphics[scale=0.5]{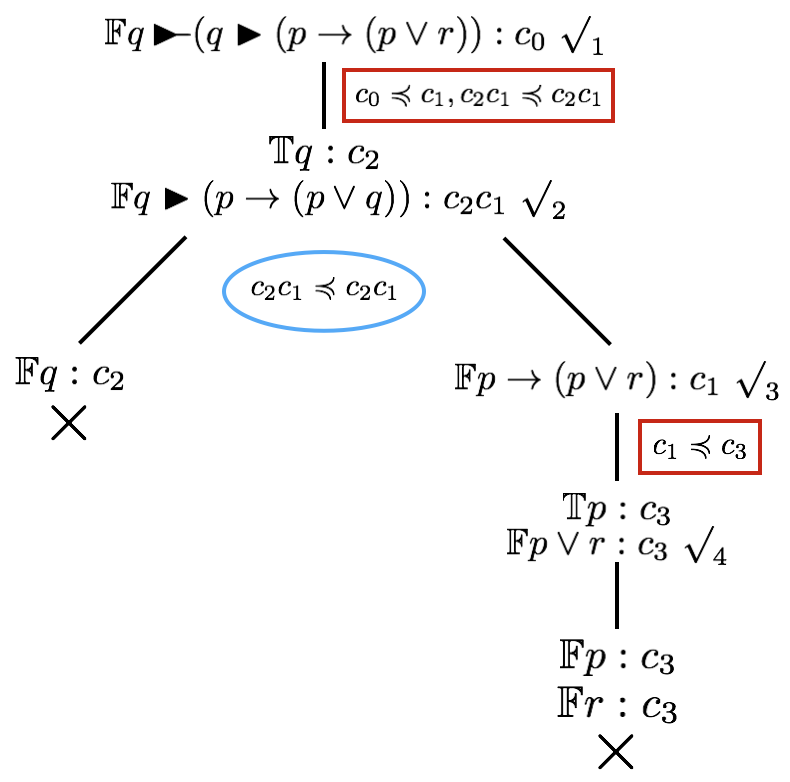}
\caption{{\bf A tableau proof of $q \limp (q \blacktriangleright (p \rightarrow (p \lor q)))$ }}
\label{fig:tableaux}
\end{figure}
%

CSSs are related back to the graph semantics via the notion of realization.

\begin{defi}[Realization] \label{def:realization}
Let $\CSS{\mathcal{F}}{\mathcal{C}}$ be a CSS. A \emph{realization} of
$\CSS{\mathcal{F}}{\mathcal{C}}$ is a triple $\mathfrak{R} = (\mathcal{X}, \valuation,
\realization{.})$ where $\model = (\mathcal{X}, \valuation)$ is a layered graph model and $\realization{.} :
\domain{\mathcal{C}} \rightarrow X$ is such that 
(1) for all $x \in \domain{\mathcal{C}}$, if $x = c_ic_j$, then $\gcompE{\realization{c_i}}{\realization{c_j}} \downarrow$ and $\realization{x} = \gcompE{\realization{c_i}}{\realization{c_j}}$), 
(2) if $x \preccurlyeq y \in \mathcal{C}$, then $\realization{x} \preccurlyeq_{\model} \realization{y}$,   
(3) if $\labelledFormula{T}{\varphi}{x} \in \mathcal{F}$, then
	$\realization{x} \satisfaction{\model} \varphi$,
(4) if $\labelledFormula{F}{\varphi}{x} \in \mathcal{F}$, then
    $\realization{x} \not \models_{\model} \varphi$. \hfill \qed    
\end{defi}

We say that a CSS is \emph{realizable} if there exists a realization of it. We say that a 
tableau is \emph{realizable} if at least one of its branches is realizable. We can also show that the
relevant clauses of the definition extend to the closure of the constraint set automatically.

\begin{prop} \label{prop_realization_relations_closure_assertions}
Let $\CSS{\mathcal{F}}{\mathcal{C}}$ be a CSS and
$\mathfrak{R} = (\mathcal{X}, \valuation, \realization{.})$
a realization of it. Then: (1) for all $x \in \domain{\closure{\mathcal{C}}}$,
$\realization{x}$ is defined; (2) if $x \preccurlyeq y \in \closure{\mathcal{C}}$, then 
$\realization{x} \preccurlyeq \realization{y}$. \hfill \qed
\end{prop}

%

\section{Soundness \& Completeness}  \label{sec:metatheory}

In this section we establish the soundness and, via countermodel extraction, the completeness 
of ILGL's tableaux system with respect to layered graph semantics. The proof of soundness is straightforward (cf. \cite{Courtault-Galmiche15,fitting1972,GMP05,Lar14a}). We begin with two key lemmas about realizability 
and closure. Their proofs proceed by simple case analysis. 

\begin{lem} \label{lem_rules_preserve_realizability}
The tableaux rules for ILGL preserve realizability. \hfill \qed
\end{lem}

\begin{lem} \label{lem_closed_tableau_not_realizable}
Closed branches are not realizable. \hfill \qed
\end{lem}

\begin{thm}[Soundness] \label{thm:soundness}
If there exists a closed tableau for the formula $\varphi$, then $\varphi$ is valid in layered graph models. 
\end{thm}

\begin{proof}
Suppose that there exists a proof for $\varphi$. Then there is a closed
tableau $\mathcal{T}_\varphi$ for the CSS $\mathfrak{C} = \CSS{\{
  \labelledFormula{\mathbb{F}}{\varphi}{c_0} \}}{\{ c_0 \preccurlyeq c_0 \}}$. Now suppose that $\varphi$ is not valid. Then
there is a countermodel $\model = (\mathcal{X},
\valuation)$ and a subgraph $H \in X$ such
that $H \not \models_{\model} \varphi$. Define $\mathfrak{R} =
(\model, \mathcal{V}, \realization{.})$ with
$\realization{c_0} = H$. Note that $\mathfrak{R}$ is a realization
of $\mathfrak{C}$, hence by Lemma \ref{lem_rules_preserve_realizability},
$\mathcal{T}_\varphi$ is realizable. By Lemma
\ref{lem_closed_tableau_not_realizable}, $\mathcal{T}_\varphi$ cannot be
closed. But, this contradicts the fact that $\mathcal{T}_\varphi$ is a proof and
therefore a closed tableau. It follows that $\varphi$ is valid.
\end{proof}

We now proceed to establish the completeness of the labelled tableaux with respect to layered 
graph semantics. We begin with the notion of a Hintikka CSS, which will facilitate the construction 
of countermodels.

\begin{defi}[Hintikka CSS] \label{def:hintikka_css}
A CSS $\CSS{\mathcal{F}}{\mathcal{C}}$ is a \emph{Hintikka CSS} iff,
for any formulae  $\varphi, \psi \in \mathrm{Form}$ and any graph labels $x,y \in \mathbb{L}$, 
we have the following: 
\[ {\small 
\begin{array}{rl}
1. &  \mbox{\rm $ \labelledFormula{T}{\varphi}{x} \not \in \mathcal{F}$ or
             $\labelledFormula{F}{\varphi}{y} \not \in \mathcal{F}$ or
             $x \preccurlyeq y \not \in \closure{\mathcal{C}}$ \quad 
        2. $\labelledFormula{F}{\top}{x} \not \in \mathcal{F}$ \quad 
        3. $\labelledFormula{T}{\bot}{x} \not \in \mathcal{F}$} \\

4. & \mbox{\rm if $\labelledFormula{T}{\varphi \wedge \psi}{x} \in \mathcal{F}$,
        then $\labelledFormula{T}{\varphi}{x} \in \mathcal{F}$ and
        $\labelledFormula{T}{\psi}{x} \in \mathcal{F}$}  \\ 
 
5. & \mbox{\rm if $\labelledFormula{F}{\varphi \wedge \psi}{x} \in \mathcal{F}$,
        then $\labelledFormula{F}{\varphi}{x} \in \mathcal{F}$ or
        $\labelledFormula{F}{\psi}{x} \in \mathcal{F}$} \\

6. & \mbox{\rm if $\labelledFormula{T}{\varphi \vee \psi}{x} \in \mathcal{F}$,
        then $\labelledFormula{T}{\varphi}{x} \in \mathcal{F}$ or
        $\labelledFormula{T}{\psi}{x} \in \mathcal{F}$}  \\

        7. & \mbox{\rm if $\labelledFormula{F}{\varphi \vee \psi}{x} \in \mathcal{F}$,
        then $\labelledFormula{F}{\varphi}{x} \in \mathcal{F}$ and
        $\labelledFormula{F}{\psi}{x} \in \mathcal{F}$}  \\
 
8. & \mbox{\rm if $\labelledFormula{T}{\varphi \rightarrow \psi}{x} \in \mathcal{F}$,
        then, for all $y \in \mathbb{L}$, if $x\preccurlyeq y \in
        \closure{\mathcal{C}}$, then  $\labelledFormula{F}{\varphi}{y} \in \mathcal{F}$ or
        $\labelledFormula{T}{\psi}{y} \in \mathcal{F}$}   \\   
 
9. & \mbox{\rm if $\labelledFormula{F}{\varphi \rightarrow \psi}{x} \in
   	 \mathcal{F}$, then there exists $y \in \mathbb{L}$ such that 
   	 $x \preccurlyeq y \in \closure{\mathcal{C}}$} \\ 
	 & \mbox{\rm and
         $\labelledFormula{T}{\varphi}{y} \in \mathcal{F}$ and
         $\labelledFormula{F}{\psi}{y} \in \mathcal{F}$} \\

10. & \mbox{\rm if $\labelledFormula{T}{\varphi \blacktriangleright \psi}{x} \in \mathcal{F}$,
        then there are $c_i, c_j \in \Sigma$ such that 
        $c_ic_j \preccurlyeq x \in  \closure{\mathcal{C}}$ 
        and} \\ 
        &  \mbox{\rm $\labelledFormula{T}{\varphi}{c_i} \in \mathcal{F}$ and
        $\labelledFormula{T}{\psi}{c_j} \in \mathcal{F}$} \\

11. & \mbox{\rm if $\labelledFormula{F}{\varphi \blacktriangleright \psi}{x} \in \mathcal{F}$,
        then, for all $c_i, c_j \in \Sigma$, 
        if $c_ic_j \preccurlyeq x  \in \closure{\mathcal{C}}$, then} \\ 
       & \mbox{\rm $\labelledFormula{F}{\varphi}{c_i} \in \mathcal{F}$ or
        $\labelledFormula{F}{\psi}{c_j} \in \mathcal{F}$}  \\

12. & \mbox{\rm if $\labelledFormula{T}{\varphi \gimp \psi}{x} \in
    \mathcal{F}$, then, 
        for all $c_i, c_j \in \Sigma$, if 
        $x \preccurlyeq c_i \in \closure{\mathcal{C}}$ and
        $c_ic_j \in \domain{\closure{\mathcal{C}}}$, then} \\
        & \mbox{\rm $\labelledFormula{F}{\varphi}{c_j} \!\in\! \mathcal{F}$ or
        $\labelledFormula{T}{\psi}{c_ic_j} \!\in\! \mathcal{F}$}  \\
            
13. & \mbox{\rm if $\labelledFormula{F}{\varphi \gimp \psi}{x} \in
    \mathcal{F}$, then  
         there are $c_i, c_j \in \Sigma$ such that 
	$x \preccurlyeq c_i \in \closure{\mathcal{C}}$ and
        $c_ic_j \in \domain{\closure{\mathcal{C}}}$ and} \\  
        & \mbox{\rm $\labelledFormula{T}{\varphi}{c_j} \in \mathcal{F}$ and
        $\labelledFormula{F}{\psi}{c_ic_j} \in \mathcal{F}$}  \\
 
14. & \mbox{\rm if $\labelledFormula{T}{\varphi \limp \psi}{x} \in
    \mathcal{F}$, then, 
        for all $c_i, c_j \in \Sigma$, if 
        $x \preccurlyeq c_i \in \closure{\mathcal{C}}$ and
        $c_jc_i \in \domain{\closure{\mathcal{C}}}$, then} \\
        & \mbox{\rm $\labelledFormula{F}{\varphi}{c_j} \!\in\! \mathcal{F}$ or \/
        $\labelledFormula{T}{\psi}{c_jc_i} \!\in\! \mathcal{F}$}  \\
     
15. &  \mbox{\rm if $\labelledFormula{F}{\varphi \gimp \psi}{x} \in
    \mathcal{F}$, then  
         there are $c_i, c_j \in \Sigma$ such that 
	$x \preccurlyeq c_i \in \closure{\mathcal{C}}$ and
        $c_jc_i \in \domain{\closure{\mathcal{C}}}$ and } \\
        & \mbox{\rm $\labelledFormula{T}{\varphi}{c_j} \in \mathcal{F}$ and \/
        $\labelledFormula{F}{\psi}{c_jc_i} \in \mathcal{F}$.}     
   \end{array}}
 \] \qed
\end{defi}

We now give the definition of a function $\Omega$ that extracts a countermodel from a Hintikka CSS.   
A Hintikka CSS can thus be seen as the \emph{labelled} tableaux counterpart of Hintikka sets, which are 
maximally consistent sets satisfying a subformula property. 

\begin{defi}[Function $\Omega$]  \label{def:function_omega}
Let $\CSS{\mathcal{F}}{\mathcal{C}}$ be a Hintikka CSS. The function
$\Omega$ associates to $\CSS{\mathcal{F}}{\mathcal{C}}$ a tuple
$\Omega(\CSS{\mathcal{F}}{\mathcal{C}}) = (\mathcal{G}, \mathcal{E}, X, \preccurlyeq,
\valuation)$, such that 
\begin{enumerate}
\item $V(\mathcal{G}) = \mathcal{A}(\mathcal{C}$),
\item $E(\mathcal{G}) = \{ (c_i, c_j) \mid c_ic_j \in \domain{\mathcal{C}} \} = \mathcal{E}$, 
$X = \lbrace x^{\omegasubgraph} \mid x \in \domain{\mathcal{C}} \}$, where 
$V(c^{\omegasubgraph}_i) = \{ c_i \}$, $E(c^{\omegasubgraph}_i) = \emptyset$, 
$V((c_ic_j)^{\omegasubgraph}) = \{ c_ic_j \}$, and $E((c_ic_j)^{\omegasubgraph}) = \{ (c_i, c_j) \}$, 
\item $x^{\omegasubgraph} \preccurlyeq y^{\omegasubgraph}$ iff  $x \preccurlyeq y \in \closure{\mathcal{C}}$, 
and
\item $x^{\omegasubgraph} \in \valuation(p)$ iff there exists $y \in \domain{\mathcal{C}} $ such that
$y \preccurlyeq x \in \closure{\mathcal{C}}$ and $\labelledFormula{T}{p}{y} \in \mathcal{F}$. \qed
\end{enumerate}
\end{defi}

The next lemma shows that there is a precise correspondence between the structure that the Hintikka CSS 
properties impose on the labels and the layered structure specified by the construction of the model.

\begin{lem} \label{lem:omega_prop} Let $\CSS{\mathcal{F}}{\mathcal{C}}$ be a Hintikka CSS and 
$\Omega(\CSS{\mathcal{F}}{\mathcal{C}}) = (\mathcal{G}, \mathcal{E}, X, \preccurlyeq,\valuation)$.
\begin{enumerate}
\item If $c_i, c_j \in \alphabet{\mathcal{C}}$, then $c_ic_j \in \domain{\mathcal{C}}$ iff $\gcompE{c^{\omegasubgraph}_i}{c^{\omegasubgraph}_j} \downarrow$.
\item If $c_ic_j \in \domain{\mathcal{C}}$, then $(c_ic_j)^{\omegasubgraph} = \gcompE{c^{\omegasubgraph}_i}{c^{\omegasubgraph}_j}$. 
\item $\gcompE{x^{\omegasubgraph}}{y^{\omegasubgraph}}\downarrow$ iff there exist $c_i, c_j \in \alphabet{\mathcal{C}}$ s.t. $x=c_i$, $y=c_j$ and $c_ic_j \in \domain{\mathcal{C}}$.
\end{enumerate}
\end{lem}

\begin{proof}
\begin{enumerate}
	\item Immediate from CSS property (Contra).
%
	\item Immediate from 1. and the definition of $\Omega$.
	\item The right-to-left direction is trivial, so assume 
		$\gcompE{x^{\omegasubgraph}}{y^{\omegasubgraph}}\downarrow$.
		There are three possible cases for $x$ and $y$ other than $x=c_i$ and $y=c_j$:
		we attend to one as the others are similar. Suppose $x=c_ic_j$ and $y=c_k$. Then
		$\gcompE{x^{\omegasubgraph}}{y^{\omegasubgraph}}\downarrow$ must hold 
		because of either $(c_i, c_k) \in \mathcal{E}$ or $(c_j, c_k) \in \mathcal{E}$. 
		That is, $c_ic_k \in \domain{\mathcal{C}}$ or $c_jc_k \in \domain{\mathcal{C}}$.
		In both cases the $CSS$ property (Freshness) is contradicted so neither can hold.
		It follows that only the case $x=c_i$ and $y=c_j$ is non-contradictory, and so by 1.
		$c_ic_j \in \domain{\mathcal{C}}$. \qedhere
\end{enumerate}
\end{proof}

\begin{lem} \label{lem_omega_extracts_model}
Let $\CSS{\mathcal{F}}{\mathcal{C}}$ be a Hintikka CSS. $\Omega(\CSS{\mathcal{F}}{\mathcal{C}})$ 
is a layered graph model. 
\end{lem}
\begin{proof}
$\mathcal{G}$ is clearly a graph and $\preccurlyeq$ being a preorder on 
$X$ can be read off of the rules for the closure of constraint sets. 
Thus the only non-trivial aspects of the proof are that $X$ is admissible 
and that $\valuation$ is persistent.

First we show that	$X$ is an admissible subgraph set.
		Let $H, K \in S\!g(\mathcal{G})$ with $\gcompE{H}{K}\downarrow$. 
		First we assume $H, K \in X$. Then $H=x^{\omegasubgraph}$ and $K=y^{\omegasubgraph}$
		for labels $x, y$. By the previous lemma it follows that $x = c_i$ and $y=c_j$
		and $c_ic_j \in \domain{\mathcal{C}}$. Thus $\gcompE{H}{K} = 
		\gcompE{c_i^{\omegasubgraph}}{c_j^{\omegasubgraph}} 
		= (c_ic_j)^{\omegasubgraph} \in X$.
		Now suppose $\gcompE{H}{K} \in X$. Then $\gcompE{H}{K} = x^{\omegasubgraph}$
		for some $x \in \domain{\mathcal{C}}$. The case $x=c_i$ is clearly impossible
		as $E(c_i^{\omegasubgraph}) = \emptyset$ so necessarily $x = c_ic_j$. Then we have
		$c_i, c_j \in \domain{\mathcal{C}}$ as sub-labels of $c_ic_j$ and
		$\gcompE{c^{\omegasubgraph}_i}{c^{\omegasubgraph}_j} \downarrow$ with 
		$\gcompE{c^{\omegasubgraph}_i}{c^{\omegasubgraph}_j}$ the only possible composition
		equal to  $(c_ic_j)^{\omegasubgraph}$. It follows that $H = c_i^{\omegasubgraph} \in X$
		and $K = c_j^{\omegasubgraph} \in X$ as required.
      
Finally we must show	$\valuation$ is a persistent valuation.
		Let $H \in \valuation(\Atom{p})$ with $H \preccurlyeq K$. 
		Then $H=x^{\omegasubgraph}$ and $K=y^{\omegasubgraph}$ for some 
		$x, y \in \domain{\mathcal{C}}$ with $x \preccurlyeq y \in \closure{\mathcal{C}}$. 
		By definition of $\valuation$ there exists $z \in \domain{\mathcal{C}}$ with 
		$z \preccurlyeq x \in \closure{\mathcal{C}}$ and 
		$\labelledFormula{\mathbb{T}}{p}{z} \in \mathcal{F}$. 
		By the closure rule $\langle Tr \rangle$ we have 
		$z \preccurlyeq y \in \closure{\mathcal{C}}$ so 
		$K = y^{\omegasubgraph} \in \valuation(p)$.      
\end{proof}

\begin{lem}\label{lem:countermodel}
Let $\CSS{\mathcal{F}}{\mathcal{C}}$ be a Hintikka CSS and
$\model = \Omega(\CSS{\mathcal{F}}{\mathcal{C}}) = (\mathcal{G}, \mathcal{E}, X, \preccurlyeq,
\valuation)$. For all formulae $\varphi \in \mathrm{Form}$, and all $x \in \domain{\mathcal{C}}$. 
we have (1) if $\labelledFormula{F}{\varphi}{x} \in \mathcal{F}$, then 
$x^{\omegasubgraph} \not \models_{\model} \varphi$, and (2) 
if $\labelledFormula{T}{\varphi}{x} \in \mathcal{F}$, then $x^{\omegasubgraph} \satisfaction{\model} \varphi$. 
Hence, if $\labelledFormula{F}{\varphi}{x} \in \mathcal{F}$, then $\varphi$ is not valid and $\Omega(\CSS{\mathcal{F}}{\mathcal{C}})$ 
is a countermodel of $\varphi$.
\end{lem}

\begin{proof}
We proceed by a simultaneous structural induction on $\varphi$, concentrating on cases of interest. 
\begin{itemize}
\item[-] {Base cases.}
  \begin{itemize}
    \item Case $\labelledFormula{F}{\Atom{p}}{x} \in \mathcal{F}$.  
         We suppose that $x^{\omegasubgraph} \satisfaction{\model} \Atom{p}$.
         Then $x^{\omegasubgraph} \in \valuation(\Atom{p})$. By the definition of $\valuation$,
         there is a label $y$ such that $y \preccurlyeq x \in
         \closure{\mathcal{C}}$ and 
         $\labelledFormula{T}{\Atom{p}}{y} \in \mathcal{F}$. Then by condition
         (1) of Definition \ref{def:hintikka_css}, 
         $\CSS{\mathcal{F}}{\mathcal{C}}$ is not a Hintikka CSS, a contradiction.
         It follows that $x^{\omegasubgraph} \not \models_{\model} \Atom{p}$.
           
    \item Case $\labelledFormula{T}{\Atom{p}}{x} \in \mathcal{F}$.  
          By property $(Ref)$, $x \preccurlyeq x \in \closure{\mathcal{C}}$.
          Thus, by definition of $\valuation$ we have $x^{\omegasubgraph} \in \valuation(\Atom{p})$.
          Thus $x^{\omegasubgraph} \satisfaction{\model} \Atom{p}$.
          
  \end{itemize}
  
\item[-] Inductive step. We now suppose that (1) and
(2) hold for formulae $\varphi$ and $\psi$ (IH). 
  \begin{itemize}
         \item Case $\labelledFormula{T}{\varphi \rightarrow \psi}{x} \in \mathcal{F}$.
         	Suppose $x^{\omegasubgraph} \preccurlyeq y^{\omegasubgraph}$. Then 
         	$x \preccurlyeq y \in \closure{C}$ and by Definition \ref{def:hintikka_css} 
         	property (8) it follows that $\labelledFormula{F}{\varphi}{y} \in \mathcal{F}$
         	or $\labelledFormula{T}{\psi}{y} \in \mathcal{F}$. By (IH) it follows that
         	if $y^{\omegasubgraph} \satisfaction{\model} \varphi$ then  
         	$y^{\omegasubgraph} \satisfaction{\model} \psi$ as required.
          
    	\item Case $\labelledFormula{T}{\varphi \blacktriangleright \psi}{x} \in \mathcal{F}$.
    		By Definition \ref{def:hintikka_css} property (10) there exist labels 
    		$c_i, c_j \in \domain{\mathcal{C}}$ such that $c_i c_j \preccurlyeq x \in 
    		\closure{\mathcal{C}}$ and $\labelledFormula{T}{\varphi}{c_i} \in \mathcal{F}$
    		and $\labelledFormula{T}{\psi}{c_j} \in \mathcal{F}$. By (IH) we have
    		$c_i^{\omegasubgraph} \satisfaction{\model} \varphi$ and
    		$c_j^{\omegasubgraph} \satisfaction{\model} \psi$. Further, by definition of
    		$\Omega$ we have that $(c_i c_j)^{\omegasubgraph} = 
    		\gcompE{c_i^{\omegasubgraph}}{c_j^{\omegasubgraph}} \preccurlyeq
    		x^{\omegasubgraph}$, so 
    		$x^{\omegasubgraph} \satisfaction{\model} \varphi \blacktriangleright \psi$. \qedhere
         
  \end{itemize}
\end{itemize}
\end{proof}

This construction of a countermodel would fail in the analogous labelled tableaux system for LGL (i.e., the layered graph logic with classical 
additives \cite{CMP14}). We would require a systematic way to construct the internal structure of each subgraph in the model, 
as the classical semantics for $\blacktriangleright$ demands strict equality between the graph under interpretation and the decomposition 
into layers. This issue is sidestepped for ILGL since each time the tableaux rules require a decomposition of a subgraph into layers we 
can move to a `fresh' layered subgraph further down the ordering. Thus we can safely turn each graph label into the simplest 
instantiation of the kind of graph it represents: either a single vertex (indecomposable) or two vertices and an edge (layered). A well-foundedness condition on CSSs may make this method adaptable to \logicfont{LGL} but we defer this investigation to another occasion.

We now show how to construct a Hintikka CSS. We first require a listing of all labelled formulae that may need to be added to the CSS 
in order to satisfy properties $4$--$15$. We require a particularly strong condition on the listing to make this procedure work: that every 
labelled formula appears infinitely often to be tested.

\begin{defi}[Fair strategy] \label{def:fair_strategy} 
A \emph{fair strategy} for a language $\mathcal{L}$ is a labelled sequence of 
formul{\ae} $(\labelledFormulaIndex{S}{\varphi}{x}{i})_{i \in \mathbb{N}}$ in 
$\{\mathbb{T}, \mathbb{F}\} \times \mathrm{Form} \times \mathbb{L}$ such that $\{i \in \mathbb{N} \mid
\labelledFormulaIndex{S}{\varphi}{x}{i} \equiv \labelledFormula{S}{\varphi}{x}\}$ is infinite for any
$\labelledFormula{S}{\varphi}{x} \in \{\mathbb{T}, \mathbb{F}\} \times \mathrm{Form} \times \mathbb{L}$. \hfill \qed
\end{defi}

\begin{prop}[cf. \cite{Courtault-Galmiche15}] \label{prop:exists_fair_strategy}
There exists a fair strategy for the language of ILGL. \hfill \qed
\end{prop}
Next we need the concept of an oracle. Here an oracle allows Hintikka sets to be constructed inductively, 
testing the required consistency properties at each stage. 

\begin{defi}[Oracle] \label{def:oracle}
Let $\mathcal{P}$ be a set of CSSs. (1) $\mathcal{P}$ is \emph{$\CSSinclusion$-closed} if
        $\CSS{\mathcal{F}}{\mathcal{C}} \in \mathcal{P}$
        holds whenever $\CSS{\mathcal{F}}{\mathcal{C}}
        \CSSinclusion \CSS{\mathcal{F}'}{\mathcal{C}'}$ and
        $\CSS{\mathcal{F}'}{\mathcal{C}'} \in \mathcal{P}$
        holds. (2)  $\mathcal{P}$ is of \emph{finite character} if
        $\CSS{\mathcal{F}}{\mathcal{C}} \in \mathcal{P}$
        holds whenever
        $\CSS{\mathcal{F}_f}{\mathcal{C}_f} \in
        \mathcal{P}$ holds for every
        $\CSS{\mathcal{F}_f}{\mathcal{C}_f}
        \CSSfiniteInclusion \CSS{\mathcal{F}}{\mathcal{C}}$. (3) 
        $\mathcal{P}$ is \emph{saturated} if, for any
        $\CSS{\mathcal{F}}{\mathcal{C}} \in \mathcal{P}$ and any instance \vspace{-4mm}
        \begin{prooftree}
          \AxiomC{$cond(\mathcal{F}, \mathcal{C})$}
          \UnaryInfC{$\CSS{\mathcal{F}_1}{\mathcal{C}_1} \ \mid \
            \ldots \ \mid \ \CSS{\mathcal{F}_k}{\mathcal{C}_k}$}    
        \end{prooftree} 
        of a rule of Figure~\ref{fig_tableaux_ILGL_rules},
        if $cond(\mathcal{F}, \mathcal{C})$ is fulfilled, then
        $\CSS{\mathcal{F} \cup \mathcal{F}_i}{\mathcal{C} \cup
        \mathcal{C}_i} \in \mathcal{P}$,  
        for at least one $i \in \{ 1, \ldots, k \}$. 
An \emph{oracle} is a set of open CSSs which is $\CSSinclusion$-closed, of finite character, and saturated. \hfill \qed 
\end{defi}

\begin{defi}[Consistency] \label{def:consistency} 
Let $\CSS{\mathcal{F}}{\mathcal{C}}$ be a CSS. We say $\CSS{\mathcal{F}}{\mathcal{C}}$ is \emph{consistent} if it is finite and has no closed tableau. We say $\CSS{\mathcal{F}}{\mathcal{C}}$ is \emph{finitely consistent} if every finite sub-CSS $\CSS{\mathcal{F}_f}{\mathcal{C}_f}$ is consistent. \hfill \qed
\end{defi}

\begin{prop}[cf. \cite{Courtault-Galmiche15}] \label{prop:consistency} 
(1) Consistency is $\CSSinclusion$-closed. (2) A finite CSS is consistent iff it is finitely consistent. \hfill \qed
\end{prop}

We denote the set of finitely consistent CSS by $\mathcal{P}_{\mathrm{fincon}}$.

\begin{lem}\label{lem:oracle}
$\mathcal{P}_{\mathrm{fincon}}$ is an oracle.
\end{lem}
\begin{proof}
For $\CSSinclusion$-closure and finite character see \cite{Courtault-Galmiche15}.  For saturation we show the case 
$\langle \mathbb{T}\gimp \rangle$: the rest are similar. 

		Let $\CSS{\mathcal{F}}{\mathcal{C}} \in  \mathcal{P}_{\mathrm{fincon}}$, $\labelledFormula{T}{\varphi \gimp \psi}{x} \in \mathcal{F}$ and
		$x \preccurlyeq y, yz \preccurlyeq yz \in \closure{\mathcal{C}}$.
		Suppose neither $\CSS{\mathcal{F} \cup \{ \labelledFormula{F}{\varphi}{z} \} }{\mathcal{C}} 
		\in  \mathcal{P}_{\mathrm{fincon}}$ nor $\CSS{\mathcal{F} \cup \{ \labelledFormula{T}{\psi}{yz} \}}{\mathcal{C}} 
		\in  \mathcal{P}_{\mathrm{fincon}}$. 
		Then there exist $\CSS{\mathcal{F}^A_f}{\mathcal{C}^A_f} 
		\CSSinclusion_{f} \CSS{\mathcal{F} \cup \{ \labelledFormula{F}{\varphi}{z} \} }{\mathcal{C}}$
		and $\CSS{\mathcal{F}^B_f}{\mathcal{C}^B_f} \CSSinclusion_{f}
		\CSS{\mathcal{F} \cup \{ \labelledFormula{T}{\psi}{yz} \}}{\mathcal{C}}$ that are inconsistent.
		By compactness (Lemma \ref{lem_compactness}), there exist finite $\mathcal{C}_0, \mathcal{C}_1 \subseteq \mathcal{C}$
		such that $z \preccurlyeq z \in \mathcal{C}_0$ and $yz \preccurlyeq yz \in \mathcal{C}_1$.
		Thus we define $\mathcal{F}^{\prime}_f = (\mathcal{F}^A_f \setminus \{
		\labelledFormula{F}{\varphi}{z} \}) \cup (\mathcal{F}^B_f \setminus \{\labelledFormula{T}{\psi}{yz} \})
		\cup \{ \labelledFormula{T}{\varphi \gimp \psi}{x} \}$ and $\mathcal{C}^{\prime}_f =
		\mathcal{C}^A_f \cup \mathcal{C}^B_f \cup \mathcal{C}_0 \cup \mathcal{C}_1$.
		Then $\CSS{\mathcal{F}^{\prime}_f}{\mathcal{C}^{\prime}_f}$ is a finite CSS and
		$[\CSS{\mathcal{F}^{\prime}_f \cup \{\labelledFormula{F}{\varphi}{z}\}}{\mathcal{C}^{\prime}_f}
		; \CSS{\mathcal{F}^{\prime}_f \cup \{ \labelledFormula{T}{\psi}{yz} \}}{\mathcal{C}^{\prime}_f}]$
		is a tableau for it. We have $\CSS{\mathcal{F}^{A}_f}{\mathcal{C}^A_f} \CSSinclusion_{f}
		\CSS{\mathcal{F}^{\prime}_f \cup \{\labelledFormula{F}{\varphi}{z}\}}{\mathcal{C}^{\prime}_f}$ and
		$\CSS{\mathcal{F}^{B}_f}{\mathcal{C}^B_f} \CSSinclusion_{f} 
		\CSS{\mathcal{F}^{\prime}_f \cup \{ \labelledFormula{T}{\psi}{yz} \}}{\mathcal{C}^{\prime}_f}$
		so by $\CSSinclusion$-closure of consistency $\CSS{\mathcal{F}^{A}_f}{\mathcal{C}^A_f}$ and
		$\CSS{\mathcal{F}^{B}_f}{\mathcal{C}^B_f}$ are inconsistent: respectively let
		 $\mathcal{T}_A$ and $\mathcal{T}_B$ be closed tableaux for them. Then 
		 $\mathcal{T}_A \concatList \mathcal{T}_B$ is a closed tableau for 
		 $\CSS{\mathcal{F}^{\prime}_f}{\mathcal{C}^{\prime}_f}$ and the CSS is inconsistent: contradicting
		 $\CSS{\mathcal{F}^{\prime}_f}{\mathcal{C}^{\prime}_f} \CSSinclusion_{f} 
		 \CSS{\mathcal{F}}{\mathcal{C}} \in \mathcal{P}_{\mathrm{fincon}}$.
\end{proof}

We can now show completeness of our tableaux system. Consider a formula $\varphi$ for which there 
exists no closed tableau. We show there is a countermodel to $\varphi$. We start with the initial tableau 
$\mathcal{T}_0$ for $\varphi$. Then, we have (1) $\mathcal{T}_0 = 
[\CSS{\{ \labelledFormula{F}{\varphi}{c_0} \}}{\{ c_0 \preccurlyeq c_0) \}}]$  
and (2) $\mathcal{T}_0$ cannot be closed. By 
Proposition \ref{prop:exists_fair_strategy}, there exists a fair strategy, which we denote by 
$\mathcal{S}$, with $\labelledFormulaIndex{S}{\varphi}{x}{i}$ the $i^\text{th}$ formula
of $\mathcal{S}$. As $\mathcal{T}_0$ cannot be closed, $\CSS{\{
  \labelledFormula{F}{\varphi}{c_0} \}}{\{ c_0 \preccurlyeq
  c_0 \}} \in  \mathcal{P}_{\mathrm{fincon}}$. We build a sequence $\CSS{\mathcal{F}_i}{\mathcal{C}_i}_{i
  \geqslant 0}$ as follows:
\begin{itemize}
\item[-] $\CSS{\mathcal{F}_0}{\mathcal{C}_0} =
      \CSS{\{ \labelledFormula{F}{\varphi}{c_0} \}}{\{ c_0 \preccurlyeq c_0 \}}$;  
\item[-] if $\CSS{\mathcal{F}_i \cup \{
    \labelledFormulaIndex{S}{\varphi}{x}{i} \} 
      }{\mathcal{C}_i} \not \in \mathcal{P}_{\mathrm{fincon}}$,  
      then we have $\CSS{\mathcal{F}_{i + 1}}{\mathcal{C}_{i + 1}} = 
      \CSS{\mathcal{F}_i}{\mathcal{C}_i}$; and  
\item[-] if $\CSS{\mathcal{F}_i \cup \{
    \labelledFormulaIndex{S}{\varphi}{x}{i} \} 
      }{\mathcal{C}_i} \in \mathcal{P}_{\mathrm{fincon}}$, then we have
      $\CSS{\mathcal{F}_{i + 1}}{\mathcal{C}_{i + 1}} =
      \CSS{\mathcal{F}_i \cup \{
        \labelledFormulaIndex{S}{\varphi}{x}{i} \} 
      \cup F_e }{\mathcal{C}_i \cup \mathcal{C}_e}$ where $F_e$
      and $\mathcal{C}_e$ are determined by

\begin{center}
{\small 
\begin{tabular}{|c|c||c|c|}
\hline
$\mathbb{S}_i$ & $\varphi_i$ & $F_e$ & $\mathcal{C}_e$  \\ 
\hline \hline 
$\mathbb{F}$ & $\varphi \rightarrow \psi$ & $\{ \labelledFormula{T}{\varphi}{c_{\mathfrak{I}+1}}, 
	\labelledFormula{F}{\psi}{c_{\mathfrak{I}+1}} \}$ & $\{ x_i \preccurlyeq c_{\mathfrak{I}+1} \}$ \\ 
\hline 
$\mathbb{T}$ & $\varphi \blacktriangleright \psi$ & $\{
\labelledFormula{T}{\varphi}{c_{\mathfrak{I}+1}},
\labelledFormula{T}{\psi}{c_{\mathfrak{I}+2}} \}$ & 
$\{ c_{\mathfrak{I}+1}c_{\mathfrak{I}+2} \preccurlyeq x_i \}$\\  
\hline 
$\mathbb{F}$ & $\varphi \gimp \psi$ & $\{
\labelledFormula{T}{\varphi}{c_{\mathfrak{I}+2}},
\labelledFormula{F}{\psi}{c_{\mathfrak{I}+1}c_{\mathfrak{I}+2}} \}$ 
& $\{ x_i \preccurlyeq c_{\mathfrak{I}+1}, 
c_{\mathfrak{I}+1}c_{\mathfrak{I}+2} \preccurlyeq c_{\mathfrak{I}+1}c_{\mathfrak{I}+2} \}$ \\  
\hline 
$\mathbb{F}$ & $\varphi \limp \psi$ & $\{
\labelledFormula{T}{\varphi}{c_{\mathfrak{I}+2}},
\labelledFormula{F}{\psi}{c_{\mathfrak{I}+2}c_{\mathfrak{I}+1}} \}$ 
& $\{ x_i \preccurlyeq c_{\mathfrak{I}+1}, 
c_{\mathfrak{I}+2}c_{\mathfrak{I}+1} \preccurlyeq c_{\mathfrak{I}+2}c_{\mathfrak{I}+1} \}$ \\  
\hline 
\multicolumn{2}{|c||}{Otherwise} & $\emptyset$ & $\emptyset$ \\ 
\hline 
\end{tabular}}
\end{center}
with $\mathfrak{I} = \text{max}\{j \mid c_j \in \alphabet{\mathcal{C}_i} \cup \mathcal{S}(x_i) \}$. 
\end{itemize}

\begin{prop} \label{prop:css_sequence} For any $i \in \mathbb{N}$, the following 
properties hold: (1) $\mathcal{F}_i \subseteq \mathcal{F}_{i + 1}$ and
$\mathcal{C}_i \subseteq \mathcal{C}_{i + 1}$; (2) $\CSS{\mathcal{F}_i}{\mathcal{C}_i} \in \mathcal{P}_{\mathrm{fincon}}$. 
\end{prop}

\begin{proof}
Only 2 is non-trivial. and we prove it by induction on $i$. The base case $i=0$ is given by our initial assumption. 
Now for the inductive hypothesis (IH) we have that $\CSS{\mathcal{F}_i}{\mathcal{C}_i} \in \mathcal{P}_{\mathrm{fincon}}$. Then 
the inductive step is an immediate consequence of Lemma \ref{lem:oracle} for the non-trivial cases.
\end{proof}

We now define the limit $\CSS{\mathcal{F}_\infty}{\mathcal{C}_\infty} = 
\CSS{\bigcup_{i \geqslant 0} \mathcal{F}_i}{\bigcup_{i \geqslant 0} \mathcal{C}_i}$ of the sequence
$\CSS{\mathcal{F}_i}{\mathcal{C}_i}_{i \geqslant 0}$.

\begin{prop} \label{prop:css_limit}
The following properties hold: (1) $\CSS{\mathcal{F}_\infty}{\mathcal{C}_\infty} \in \mathcal{P}_{\mathrm{fincon}}$;  
(2) For all labelled  formulae $\labelledFormula{S}{\varphi}{x}$, if $\CSS{\mathcal{F}_\infty \cup \{
	\labelledFormula{S}{\varphi}{x}\}}{\mathcal{C}_\infty} \in \mathcal{P}_{\mathrm{fincon}}$, then 
	$\labelledFormula{S}{\varphi}{x} \in \mathcal{F}_\infty$. 
\end{prop}

\begin{proof}
\begin{enumerate}
\item First note that $\CSS{\mathcal{F}_\infty}{\mathcal{C}_\infty}$ is a CSS
	since each stage of construction satisfies (Ref) and by our choice of constants
	throughout the construction (Contra) and (Freshness) are satisfied. Further, it is open since otherwise there would be some stage $\CSS{\mathcal{F}_k}{\mathcal{C}_k}$
	at which the offending closure condition is satisfied, which would contradict that each 
	$\CSS{\mathcal{F}_i}{\mathcal{C}_i}$ is consistent.
	Now let $\CSS{\mathcal{F}_f}{\mathcal{C}_f} \CSSinclusion_{f} 
	\CSS{\mathcal{F}_\infty}{\mathcal{C}_\infty}$. Then there exists $k \in \mathbb{N}$
	such that $\CSS{\mathcal{F}_f}{\mathcal{C}_f} \CSSinclusion_{f}
	\CSS{\mathcal{F}_k}{\mathcal{C}_k}$. By Proposition \ref{prop:css_sequence}
	$\CSS{\mathcal{F}_k}{\mathcal{C}_k} \in \mathcal{P}_{\mathrm{fincon}}$ so it follows
	$\CSS{\mathcal{F}_f}{\mathcal{C}_f} \in \mathcal{P}_{\mathrm{fincon}}$. As $\mathcal{P}_{\mathrm{fincon}}$ is of finite character,
	we thus have $\CSS{\mathcal{F}_\infty}{\mathcal{C}_\infty} \in \mathcal{P}_{\mathrm{fincon}}$.

\item   First note that $\CSS{\mathcal{F}_\infty \cup \{
          \labelledFormula{S}{\varphi}{x}\}}{\mathcal{C}_\infty}$ is a CSS so 
          (Contra) and (Freshness) are satisfied when the label $x$ is introduced. By compactness,
          there exists finite $\mathcal{C}_0 \subseteq \mathcal{C}_\infty$ such that 
          $x \preccurlyeq x \in \closure{\mathcal{C}_0}$. As it is finite, there exists $k \in \mathbb{N}$
          such that $\mathcal{C}_0 \subseteq \mathcal{C}_k$ and by fairness there exists 
          $l \geq k$ such that $\labelledFormulaIndex{S}{\varphi}{x}{l} \equiv  \labelledFormula{S}{\varphi}{x}$.
          Since (Freshness) and (Contra) are fufilled with respect to $\mathcal{F}_\infty$ they are also
          fulfilled with respect to $\mathcal{F}_l \cup \{ \labelledFormula{S}{\varphi}{x}\}$ so
          $\CSS{\mathcal{F}_{l+1}}{\mathcal{C}_{l+1}} = 
          \CSS{\mathcal{F}_l \cup \{ \labelledFormula{S}{\varphi}{x}\}}{\mathcal{C}_l} \in \mathcal{P}_{\mathrm{fincon}}$
          and $\CSS{\mathcal{F}_{l+1}}{\mathcal{C}_{l+1}} = 
          \CSS{\mathcal{F}_l \cup \{ \labelledFormula{S}{\varphi}{x}\} 
          \cup \mathcal{F}_e}{\mathcal{C}_l \cup \mathcal{C}_e}$. 
          Hence $\labelledFormula{S}{\varphi}{x} \in \mathcal{F}_\infty$. \qedhere
\end{enumerate}
\end{proof}

\begin{lem} \label{lem:limit_is_hintikka}
The limit CSS is a Hintikka CSS.
\end{lem}
 \begin{proof}
		For properties $(1) - (3)$ we have that $\CSS{\mathcal{F}_\infty}{\mathcal{C}_\infty}$ 
		is open. For the other conditions,
		the saturation property of the oracle $\mathcal{P}_{\mathrm{fincon}}$ and Proposition 
		\ref{prop:css_limit} item 2. suffice.
\end{proof}

We immediately obtain completeness.

\begin{thm}[Completeness] \label{th_completeness}
If $\varphi$ is valid, then there exists a closed tableau for $\varphi$. \qed
\end{thm}

\section{Decidability}\label{sec:decidability}

In this section, we prove the decidability of ILGL. Given the infinite behaviour of the countermodel extraction procedure, we do not do so by means of a tableaux termination argument. Instead, we define a semantics on a class of algebras that we prove equivalent to the layered graph semantics by means of a representation theorem. By proving the finite embeddability property holds for this class of algebras, we obtain the finite model property and thus decidability of validity for the layered graph semantics. We are able to extend the representation theorem to a full topological duality, and do so in Section~\ref{sec:duality}.

\subsection{Algebraic and Relational Semantics}
We begin by defining algebraic structures appropriate for interpreting ILGL. As one might be expected, these are obtained by replacing the Boolean algebra base of the \emph{layered algebras} that act as LGL's algebraic semantics with a Heyting algebra in order to soundly interpret the propositional intuitionistic logic fragment of ILGL.

\begin{defi}[Layered Heyting algebra] A \emph{layered Heyting algebra} is a structure $\mathbb{A}=(A, \wedge, \vee, \rightarrow, \bot, \top, \blacktriangleright, \gimp, \limp)$ 
        such that $(A, \wedge, \vee, \rightarrow, \bot, \top)$ is a Heyting algebra and $\blacktriangleright$, $\gimp$, and $\limp$ are binary operations on $A$ satisfying  
	$a \blacktriangleright b \leq c \text{ iff } a \leq b \gimp c \text{ iff } b \leq a \limp c.$ \qed
\end{defi}

We can derive some simple but useful properties about the residuated groupoid structure of a layered Heyting algebra.

\begin{prop}[cf. \cite{Jipsen}] \label{prop-alg-prop} 
Let $\mathbb{A}$ be a layered Heyting algebra. Then, for all $a, b, a', b' \in A$ and $X, Y \subseteq A$,  we have
\begin{enumerate}[leftmargin=*]
\item If $a \leq a'$ and $b \leq b'$ then $a \blacktriangleright b \leq a' \blacktriangleright b'$;
\item If $\bigvee X$ and $\bigvee Y$ exist then $\bigvee_{x \in X, y \in Y} x \blacktriangleright y$ exists and
	$(\bigvee X) \blacktriangleright (\bigvee Y) =  \bigvee_{x \in X, y \in Y} x \blacktriangleright y$;
\item If $a = \bot$ or $b = \bot$ then $a \blacktriangleright b = \bot$;	
\item If $\bigvee X$ exists then for any $z \in A$: $\bigwedge_{x \in X} x \gimp z$ and $\bigwedge_{x \in X} x 	\limp z$ exist with $\bigwedge_{x \in X} x \gimp z = \bigvee X \gimp z$ and $\bigwedge_{x \in X} x 			\limp z = \bigvee X \limp z$
\item If $\bigwedge X$ exists then for any $z \in A$ $\bigwedge_{x \in X} z \gimp x$ and 
	$\bigwedge_{x \in X} z \limp x$ exist with $\bigwedge_{x \in X} z \gimp x = z \gimp \bigwedge X$
	and $\bigwedge_{x \in X} z \limp x = z \limp \bigwedge X$;
\item $a \gimp \top = a \limp \top = \bot \gimp a =  \bot \limp a = \top$. \qed
\end{enumerate}
\end{prop}

We interpret \logicfont{ILGL} on layered Heyting algebras as follows. A valuation
$\valuation: \mathrm{Prop} \rightarrow \mathbb{A}$ on a layered Heyting algebra $\mathbb{A}$, is uniquely extended to an interpretation $\llbracket - \rrbracket$ as follows: $\llbracket \mathrm{p} \rrbracket = \valuation(\mathrm{p})$, $\llbracket \top \rrbracket = \top$, $\llbracket \bot \rrbracket = \bot$, $\llbracket \varphi \circ \psi \rrbracket = \llbracket \varphi \rrbracket \circ \llbracket \psi \rrbracket$ for 
$\circ \in \{ \land, \lor, \rightarrow, \blacktriangleright, \gimp, \limp \}$.
An algebra $\mathbb{A}$ and an interpretation $\llbracket - \rrbracket$ \emph{satisfies} $\varphi$ if 
$\llbracket \varphi \rrbracket = \top$. Correspondingly, $\varphi$ is \emph{valid} if it is satisfied by every algebra $\mathbb{A}$ and interpretation $\llbracket - \rrbracket$.

We obtain soundness and completeness for this semantics with respect to the Hilbert-type system $\mathrm{ILGL}_\mathrm{H}$ of Fig \ref{fig:hilbert_rules_ILGL} by forming a Lindenbaum-Tarski algebra from it, utilising precisely the same argument as that given for LGL.

\begin{thm}[Algebraic Soundness \& Completeness (cf. \cite{CMP14})] \label{theorem:algsoundcomplete}
For all formulae $\varphi$ and $\psi$ of \logicfont{ILGL}, $\varphi \vdash \psi$ is provable in $\mathrm{ILGL}_\mathrm{H}$ iff, for all algebraic interpretations $\llbracket - \rrbracket$, 
$\llbracket \varphi \rrbracket \leq \llbracket \psi \rrbracket$.\qed
\end{thm}

We now define relational structures appropriate to interpret \logicfont{ILGL}. These are simple extensions of the frames that interpret propositional intuitionistic logic.

\begin{defi}[Intuitionistic Layered Frame]
An \emph{intuitionistic layered frame} is a triple $\mathcal{F} = (X, \preccurlyeq, R)$ where $(X, \preccurlyeq)$ is a preorder and $R$ a ternary relation over $X$. \qed
\end{defi}

This is an obvious generalization of ordered scaffold. By defining $R(G, H, K)$ iff $\gcompE{G}{H} \downarrow$ and $\gcompE{G}{H} = K$ we obtain an intuitionistic layered frame from an ordered scaffold. An alternative frame semantics can be derived from recent independent work by Galatos \& Jipsen \cite{GJ2017}. There, algebras called \emph{generalized bunched implication algebras} are investigated. These algebras can be seen as layered Heyting algebras for which $\blacktriangleright$ is associative and has an associated unit. Such algebras, together with sequent calculus presentations of substructural logics, motivate the definition of \emph{distributive residuated frames} which provide an alternative semantics to BI-like logics. Similar structures can be defined for ILGL, though we do not pursue that here.

As with the semantics on scaffolds, we require a valuation on an intuitionistic layered frame to be persistent with respect to the preorder $\preccurlyeq$. An intuitionistic layered frame $\mathcal{F}$ together with a persistent valuation $\mathcal{V}$ is a \emph{relational model} $\mathcal{M}$ of \logicfont{ILGL}, with a satisfaction relation $\vDash_{\model}$ defined analogously to that of layered graph models in Figure \ref{fig:relsat}.

\begin{figure}
\centering
\hrule
\vspace{1mm}
\setlength\tabcolsep{1.5pt}
\setlength\extrarowheight{2pt}
\begin{tabular}{c c c c l r c c c c r r c}
$x$ & $\vDash_{\model}$ & $\mathrm{p}$ & iff & \myalign{l}{$x \in \mathcal{V}(\mathrm{p})$} & & & &
$x$ & $\vDash_{\model}$ & $\top$ & 
 $x  \not\vDash_{\model} \bot$  \\
$x$ & $\vDash_{\model}$ & $\varphi \land \psi$ & iff & \myalign{l}{$x \vDash_{\model} \varphi$ and $x \vDash_{\model} \psi$} \\
$x$ & $\vDash_{\model}$ & $\varphi \lor \psi$ & iff & $x \vDash_{\model} \varphi$ or  $x \vDash_{\model} \psi$
\\
$x$ & $\vDash_{\model}$ & $\varphi \rightarrow \psi$ & iff & \myalign{l}{for all $y \succcurlyeq x$, $y \vDash_{\model}\varphi$ implies $y$  $\vDash_{\model}$ $\psi$} \\
$x$ & $\vDash_{\model}$ & $\varphi \blacktriangleright \psi$ & iff & \multicolumn{8}{l}{there exists $x', y, z$ s.t. $x \succcurlyeq x'$,  $Ryzx'$, $y \vDash_{\model} \varphi$ and  $z \vDash_{\model} \psi$} \\
$x$ & $\vDash_{\model}$ & $\varphi \gimp \psi$ & iff & \multicolumn{8}{l}{for all $x', y, z$: $x'\succcurlyeq x, Rx'yz$ and  $y \vDash_{\model} \varphi$ implies $z \vDash_{\model} \psi$} \\
$x$ & $\vDash_{\model}$ & $\varphi \limp \psi$ & iff & \multicolumn{8}{l}{for all $x', y, z$: $x' \succcurlyeq x, Ryx'z$ and $y \vDash_{\model} \varphi$ implies $z \vDash_{\model} \psi$}
\end{tabular}
\caption{Satisfaction for relational models of ILGL.}
\vspace{1mm}
\hrule
\label{fig:relsat}
\end{figure}

We can define a new notion of realization for our labelled tableaux system to give soundness and completeness for the relational semantics.

\begin{defi}[Relational Realization]\label{def:relrealization}
Let $\CSS{\mathcal{F}}{\mathcal{C}}$ be a CSS. A \emph{relational realization} of
$\CSS{\mathcal{F}}{\mathcal{C}}$ is a triple $\mathfrak{R} = (\mathcal{F}, \valuation,
\realization{.})$ where $\model = (\mathcal{F}, \valuation)$ is a relational model and $\realization{.} :
\domain{\mathcal{C}} \rightarrow X$ is such that 
	(1) For all $x \in \domain{\mathcal{C}}$, if $x = c_ic_j$, then $R\realization{c_i}\realization{c_j}\realization{c_ic_j}$;
	(2) If $x \preccurlyeq y \in \mathcal{C}$, then $\realization{x} \preccurlyeq_{\model} \realization{y}$;
	(3) If $\labelledFormula{T}{\varphi}{x} \in \mathcal{F}$, then
	$\realization{x} \satisfaction{\model} \varphi$; and 
	(4) If $\labelledFormula{F}{\varphi}{x} \in \mathcal{F}$, then
    $\realization{x} \not \models_{\model} \varphi$. \hfill \qed    
\end{defi}

Precisely the same argument as that given in Section~\ref{sec:metatheory} gives that the labelled tableaux system is sound for the relational semantics. Noting that we can transform a layered graph model into an intuitionistic layered frame with an equivalent satisfaction relation, we also have completeness by transforming the layered graph countermodels. As the labelled tableaux is sound and complete for both the layered graph semantics and the relational semantics, we obtain their equivalence.

\begin{thm}[Relational Soundness \& Completeness] {\ }
\begin{enumerate}
\item For all formulae $\varphi$ of \logicfont{ILGL}, $\varphi$ is valid in the relational semantics iff there exists a closed tableau for $\varphi$.
\item For all formulae $\varphi$ of \logicfont{ILGL}, $\varphi$ is valid in the relational semantics iff $\varphi$ is valid in the layered graph semantics. \qed
\end{enumerate}
\end{thm}

We also define a notion of morphism for intuitionistic layered frames, yielding a category $\mathrm{IntLayFr}$. This will be important for Section~\ref{sec:duality}'s duality theorem and the extension to predicate ILGL in Section \ref{sec:predicate}.

\begin{defi}[Intuitionistic Layered P-Morphism] \label{def:pmorphism}
An \emph{intuitionistic layered p-morphism} is a map $f: X \rightarrow X'$ between intuitionistic layered frames $\mathcal{F}= (X, \preccurlyeq, R)$ and $\mathcal{F}' = (X', \preccurlyeq', R')$ satisfying the following conditions:
\begin{enumerate}
\item $f$ is order preserving;
\item If $f(x) \preccurlyeq y'$ then there exists $y$ such that $x \preccurlyeq y$ and $f(y) = y'$;
\item If $Rxyz$ then $R'f(x)f(y)f(z)$;
\item If $w' \preccurlyeq' f(x)$ and $R'y'z'w'$ then there exists $w, y, z$ such that $w \preccurlyeq x$ and $Ryzw$ with $y' \preccurlyeq' f(y)$ and $z' \preccurlyeq' f(z)$;
\item If $f(x) \preccurlyeq' w'$ and $R'w'y'z'$ then there exists $w, y, z$ such that $x \preccurlyeq w$ and $Rwyz$ with $y' \preccurlyeq' f(y)$ and $f(z) \preccurlyeq' z'$;
\item If $f(x) \preccurlyeq' w'$ and $R'y'w'z'$ then there exists $w, y, z$ such that $x \preccurlyeq w$ and $Rywz$ with $y' \preccurlyeq f(y)$ and $f(z) \preccurlyeq z'$. \qed
\end{enumerate}
\end{defi}

\subsection{Equivalence of Semantics}
We now show that the algebraic and relational semantics are equivalent. This means we can prove properties of the layered graph semantics algebraically and yields the equivalence of the Hilbert and labelled tableaux systems. We begin by lifting the semantic clauses of the relational semantics to operations on the upward-closed subsets of the frame.

\begin{defi}[Complex Algebra] \label{def:layeredcomplexalgebra}
Given an intuitionistic layered frame $\mathcal{X}$, the \emph{complex algebra} of $\mathcal{X}$ is given by
$Com(\mathcal{X}) = (\mathcal{P}_{\preccurlyeq}(X), \cap, \cup, \Rightarrow, X, \emptyset, \blacktriangleright_R, \gimp_R, \limp_R)$, 
where 
\begin{align*} \begin{array}[b]{cl}
	\mathcal{P}_{\preccurlyeq}(X) &= \{ A \subseteq X \mid \text{ if } a \in A \text{ and } a \preccurlyeq b \text{ then } b \in A \} \\
	A \Rightarrow B &= \{ x \mid \text{if } x \preccurlyeq x' \text{ and } x' \in A \text{ then } x' \in B \} \\
	A \blacktriangleright_R B &= \{ x \mid \text{there exists } w, y, z \text{ s.t } w \preccurlyeq x, Ryzw, y \in A \text{ and } z \in B \} \\
	 A \gimp_R B &= \{ x \mid \text{for all } w, y, z, \text{ if } x \preccurlyeq w, Rwyz \text{ and } y \in A \text{ then } z \in B \}  \\
	 A \limp_R B &= \{ x \mid \text{for all } w, y, z, \text{ if } x \preccurlyeq w, Rywz \text{ and } y \in A \text{ then } z \in B \}.  
\end{array} \tag*{\qed} \end{align*}
\end{defi}

It is standard that the $(\mathcal{P}_{\preccurlyeq}(X), \cap, \cup, \Rightarrow, X, \emptyset)$-fragment of this structure is a Heyting algebra. That $\blacktriangleright_R$, $\gimp_R$ and $\limp_R$ satisfy the residuation property can also easily be verified. We therefore have that complex algebras are layered Heyting algebras.

\begin{lem} \label{lem:comislha}
Given an intuitionistic layered frame $\mathcal{X}$ the complex algebra $Com(\mathcal{X})$ is a layered Heyting algebra. \qed
\end{lem}

A persistent valuation $\valuation: \mathrm{Prop} \rightarrow \mathcal{P}(X)$ on an intuitionistic layered frame $\mathcal{X}$ automatically generates an interpretation $\llbracket - \rrbracket_{\valuation}$ on its complex algebra $Com(\mathcal{X})$. An inductive argument shows that satisfaction coincides in both structures.

\begin{lem} \label{lem:algsatisfiability}
Given a relational model $\mathcal{M} = (\mathcal{X}, \valuation)$, for all $x \in X$ and all formulae $\varphi$ of \logicfont{ILGL}, $x \in \llbracket \varphi \rrbracket_{\valuation}$ iff $x \vDash_{\mathcal{M}} \varphi$. \qed
\end{lem}

Conversely, we can generate a frame from any layered Heyting algebra. We first recall some definitions.\

For a bounded distributive lattice $\mathbb{A}$, a \emph{filter} is a subset $F \subseteq \mathbb{A}$ satisfying, for all $a, b \in \mathbb{A}$: i) if $a \in F$ and $a \leq b$ then $b \in F$; and ii) if $a, b \in F$ then $a \land b \in F$. It is \emph{proper} if $F \neq \mathbb{A}$. By upwards-closure this condition is equivalent to $\bot \not\in F$. A \emph{prime filter} is a proper filter additionally satisfying, for all $a, b \in \mathbb{A}$, if $a \lor b \in F$ then $a \in F$ or $b \in F$. For a set $X \subseteq \mathbb{A}$ we define $[ X ) = \{ a \mid \text{there exists } b_1, \ldots, b_n \in X \text{ such that } b_1 \land \cdots \land b_n \leq a \}$. It can be shown that this defines the least filter containing $X$. For $X = \{ a \}$ we simply write $[ a )$, and for $X = Y \cup \{ a \}$ we write $[Y, a)$.

An \emph{ideal} is the dual notion of a filter. Analogous definitions for proper and prime ideals can be obtained by switching $\land$ and $\lor$ and reversing the order in the above definitions. We also obtain the dual notation $(X]$ for least ideal containing $X$.


Importantly, given a prime ideal $P$, the complement
$\overline{P} = \{ a \in \mathbb{A} \mid a \not\in P \}$ is a prime filter (and vice versa). A (prime) filter/ideal of a layered Heyting algebra is simply a (prime) filter/ideal of the underlying bounded distributive lattice.

\begin{defi}[Prime Filter Frame] \label{def:ilglprimefilter}
Given a layered Heyting algebra $\mathbb{A}$, the \emph{prime filter frame} $Pr(\mathbb{A})$ is defined
$Pr(\mathbb{A}) = (Pr(A), \subseteq, R_{Pr(\mathbb{A})})$ where $Pr(A)$ is the set of prime filters of $\mathbb{A}$ and $R_{Pr(\mathbb{A})}$ is given by
$R_{Pr(\mathbb{A})}F_0F_1F_2 \text{ iff for all } a \in F_0 \text{ and all } b \in F_1, a \blacktriangleright b \in F_2.$ \qed
\end{defi}

We can prove the satisfaction on a layered Heyting algebra and its prime filter frame coincides. To show this we first prove a representation theorem for layered Heyting algebras. This is an analogue of the Stone representation theorem for Boolean algebras \cite{Stone} that shows that every Boolean algebra is isomorphic to a field of sets. Specifically, it extends the representation theorem for Heyting algebras \cite{DH01}.

\begin{thm}[Representation Theorem for Layered Heyting Algebras] \label{thm:representation}
Every layered Heyting algebra is isomorphic to a subalgebra of a complex algebra. Specifically, the map 
$h_{\mathbb{A}}(a): \mathbb{A} \rightarrow Com(Pr(\mathbb{A}))$ defined 
$h_{\mathbb{A}}(a) = \{ F \in Pr(A) \mid a \in F \}$ is an embedding.
\end{thm}

\begin{proof}
That $h_{\mathbb{A}}$ is injective and respects the Heyting algebra operations is simply the representation theorem for Heyting algebras. We must therefore verify $h_{\mathbb{A}}$ also respects the residuated structure. We attend to $\gimp$, leaving the similar $\blacktriangleright$ 
and $\limp$ cases to the reader.

First suppose $a = \bot$. Then by Proposition \ref{prop-alg-prop}, $a \gimp b = \top$. We have that \sloppy{ $h_{\mathbb{A}}(a) \gimp_{R_{Pr(\mathbb{A})}} h_{\mathbb{A}}(b) = Pr(A)$} since the defining condition of $\gimp_{R_{Pr(\mathbb{A})}}$ is vacuously true ($h_{\mathbb{A}}(a) = \emptyset$), so every prime filter is trivially a member of $h_{\mathbb{A}}(a) \gimp_{R_{Pr(\mathbb{A})}} h(b)$. Thus $h_{\mathbb{A}}(a \gimp b) = h_{\mathbb{A}}(a) \gimp_{R_{Pr(\mathbb{A})}} h_{\mathbb{A}}(b)$ as required. The case $b = \top$ is similar.


We divide the case where both $a \neq \bot$ and $b \neq \top$ into two subcases: $a \gimp b \neq \bot$ and $a\gimp b = \bot$. First $a \gimp b \neq \bot$. 
The left-to-right direction is easy: suppose $a \gimp b \in F$. Then for any $F_0 \supseteq F$ 
we have $a \gimp b \in F_0$. Suppose for such a $F_0$ we have $R_{pr}F_0F_1F_2$ and $F_1 \in h_{\mathbb{A}}(a)$. Then $(a \gimp b) \blacktriangleright a \in F_2$. We have $(a \gimp b) \blacktriangleright a \leq b$, 
so by filterhood $b \in F_2$. Thus $F \in h_{\mathbb{A}}(a) \gimp_{R_{Pr(\mathbb{A})}} h_{\mathbb{A}}(b)$. 

The other direction is more involved. Suppose we have $F_0 \in h_{\mathbb{A}}(a) \gimp_{R_{Pr(\mathbb{A})}} h_{\mathbb{A}}(b)$ and for contradiction
$a\gimp b \not\in F_0$. We claim this entails the existence of $F_1, F_2$ such that $R_{pr}F_0F_1F_2$ and $F_1 \in h_{\mathbb{A}}(a)$ but $F_2 \not\in h_{\mathbb{A}}(b)$. Since $a \neq \bot$ and $b \neq \top$ we have a proper filter $[a)$ and a proper ideal $(b]$ such that, for all $x \in F_0$ and all $y \in [a)$ we have $x \blacktriangleright y \not\in (b]$: if there existed $x$ and $y$ contradicting this statement then $x \blacktriangleright a \leq x \blacktriangleright y \leq b$, so by 
residuation $x \leq a \gimp b$ and $a \gimp b \in F_0$, contradicting our assumption. Hence we define $E = \{ \langle F, I \rangle \mid F \text{ is a proper filter}, I \text{ is a proper ideal}, 
		a \in F, b \in I \text{ and for all } x \in F_0, y \in F: x \blacktriangleright y \not\in I \}$. 
$E$ is non-empty. Further, when ordered by component-wise inclusion, every chain in $E$ has a maximum given by taking the union of each component of the chain. Hence, by Zorn's lemma, $E$ has a maximum $\langle F^{max}, I^{max} \rangle$. We claim $F^{max}$ is a prime filter and $I^{max}$ is a prime ideal.


We prove that $I^{max}$ is a prime ideal; the case for $F^{max}$ is similar. Suppose $x \land y \in I^{max}$ but $x, y \not\in I^{max}$. Then we have
proper ideals $(I^{max}, x], (I^{max}, y] \supset I^{max}$. $\langle F^{max}, (I^{max}, x] \rangle$ and $\langle F^{max}, (I^{max}, y] \rangle$ can only fail to be maximal in $E$ if there exist $a_0, a'_0 \in F_0$ and $a_1, a'_1 \in F^{max}$ such that
$a_0 \blacktriangleright a_1 \in (I^{max}, x]$ and $a'_0 \blacktriangleright a'_1 \in (I^{max}, y]$. There then exist $c, c' \in I^{max}$ such that 
$a_0 \blacktriangleright a_1 \leq x \lor c$ and $a'_0 \blacktriangleright a'_1 \leq y \lor c'$. Set
$a'' = a_1 \land a'_1$ and $c'' = c \lor c'$. Then $a'' \in F^{max}$ and $c'' \in I^{max}$. By monotonicity
of $\blacktriangleright$ we have $a_0 \blacktriangleright a'' \leq c'' \lor x$ and 
$a'_0 \blacktriangleright a'' \leq c'' \lor y$. Hence, by residuation and filterhood of $F_0$, we have
$(a'' \gimp (c'' \lor x)) \land (a'' \gimp (c'' \lor y)) \in F_0$. Now $\gimp$ distributes over $\land$ in its right argument so we have $a'' \gimp (c'' \lor (x \land y)) \in F_0$. Hence
$(a'' \gimp (c'' \lor (x \land y))) \blacktriangleright a'' \not\in I^{max}$. Since $I^{max}$ is downwards-closed and
$(a'' \gimp (c'' \lor (x \land y))) \blacktriangleright a'' \leq (c'' \lor (x \land y))$ we have that 
$c'' \lor (x \land y) \not\in I^{max}$. However, by assumption, $c'', x\land y \in I^{max}$, hence $c'' \lor (x \land y) \in I^{max}$ : contradiction. It follows that $x \in I^{max}$ or $y \in I^{max}$ and $I^{max}$ is prime. 

Take $F_1 = F^{max}$ and $F_2 = \overline{I^{max}}$. Then $F_0 \subseteq F_0$,
$R_{pr}F_0F_1F_2$, $F_1 \in h_{\mathbb{A}}(a)$ and $F_2 \not\in h_{\mathbb{A}}(b)$, contradicting that $F_0 \in h_{\mathbb{A}}(a) \gimp_{R_{Pr(\mathbb{A})}} h_{\mathbb{A}}(b)$. Hence $a \gimp b \in F$ as required.

Finally, suppose $a \gimp b = \bot$. Then $h_{\mathbb{A}}(a \gimp b) = \emptyset$.
Let $F$ be arbitrary. We have, for all $x \in F$ and all $y \in [a)$, that $x \blacktriangleright y \not\in (b]$ since otherwise
$x \blacktriangleright a \leq b$ and then $x \leq a \gimp b = \bot$: hence $x = \bot \not\in F$, a contradiction. Hence we can run the argument of the previous subcase to obtain $F_1$ and
$F_2$ such that $F_1 \in h(a)$, $R_{pr}FF_1F_2$, but $F_2 \not\in h_{\mathbb{A}}(b)$.
\end{proof}

Now given an interpretation $\llbracket - \rrbracket$ on a layered Heyting algebra $\mathbb{A}$ we immediately obtain a persistent valuation on the prime filter frame by setting $\valuation(\mathrm{p}) = h_{\mathbb{A}}(\llbracket \mathrm{p} \rrbracket)$. This makes $\mathrm{Pr} = (Pr(\mathbb{A}), \valuation)$ a relational model of \logicfont{ILGL}. The fact that $h_{\mathbb{A}}$ is a homomorphism, together with Lemma \ref{lem:algsatisfiability}, immediately yields the following lemma.

\begin{lem} \label{lem:framesatisfiability}
Given a layered Heyting algebra $\mathbb{A}$ and an interpretation $\llbracket - \rrbracket$, for all 
$F \in Pr(A)$ and all formulae $\varphi$ of ILGL, $\llbracket \varphi \rrbracket \in F$ iff 
$F \vDash_{\mathrm{Pr}} \varphi$. \qed
\end{lem}

Taking Lemmas \ref{lem:algsatisfiability} and \ref{lem:framesatisfiability} together yields the equivalence of the algebraic and relational semantics.

\begin{thm}[Equivalence of algebraic and relational semantics] {\ }
For all formulae $\varphi \in \mathrm{Form}$: $\varphi$ is satisfiable (valid) on layered Heyting algebras iff $\varphi$ is satisfiable (valid) on layered frames. \qed
\end{thm}

We also obtain the equivalence of the Hilbert and labelled tableaux systems, and with it soundness and completeness with respect to $\mathrm{ILGL}_\mathrm{H}$ for the layered graph semantics. With these equivalences proven, we may refer to \logicfont{ILGL} without specifying the particular semantics and proof system of those stated thus far.

\begin{cor}[Equivalence of the Hilbert and tableaux systems] \label{cor:Hilbert-tableaux} {\ }
\begin{enumerate}
\item For all formulae $\varphi$ and $\psi$ of \logicfont{ILGL}, $\varphi \vdash \psi$ is provable in $\mathrm{ILGL}_\mathrm{H}$ iff\/ there is a closed tableau for $\varphi \rightarrow \psi$.
\item For all formulae $\varphi$ and $\psi$ of \logicfont{ILGL}, $\varphi \vdash \psi$ is provable in $\mathrm{ILGL}_\mathrm{H}$ iff $\varphi \rightarrow \psi$ is valid on preordered scaffolds. \qed
\end{enumerate}
\end{cor}

A sequent calculus, a natural deduction system, and a display calculus (satisfying cut elimination) can be given for 
LGL. They define equivalent notions of provability to its Hilbert system \cite{CMP14}. It is straightforward to weaken the classical propositional logic component of each to obtain such systems for ILGL. Each of these systems is also proven equivalent to our labelled tableaux system by this corollary.

\subsection{Finite Model Property \& Decidability}
We now prove the finite embeddability property (FEP) for the class of layered Heyting algebras, and with it, decidability of \logicfont{ILGL}. We first recall the definition of the FEP for classes of ordered algebras.

\begin{defi}[Finite Embeddability Property] A class of 
algebras $\mathcal{K}$ has the \emph{finite embeddability property} (FEP) if, for any algebra $\mathbb{A} \in \mathcal{K}$ and any finite subset $B \subseteq A$ there exists a finite algebra $\mathbb{C} \in \mathcal{K}$ and an injective map $g: B \rightarrow C$ such that for all algebraic operations $f$, if $b_1, \ldots, b_n \in B$ and $f_{\mathbb{A}}(b_1, \ldots, b_n) \in B$ then $g(f_{\mathbb{A}}(b_1, \ldots, b_n)) = f_{\mathbb{C}}(g(b_1), \ldots, g(b_n))$. \qed 
\end{defi}

Intuitively, the FEP states that every finite partial subalgebra can be completed as a finite algebra. If the class of layered Heyting algebras has the FEP then ILGL has a finite model property and is thus decidable by the following argument. Suppose the algebra $\mathbb{A}$ with interpretation $\llbracket - \rrbracket$ witnesses that $\varphi \vdash \psi$ does not hold: that is, $\llbracket \varphi \rrbracket \not\leq_{\mathbb{A}} \llbracket \psi \rrbracket$. Set $B = \{ \llbracket \chi \rrbracket \mid \chi \text{ a subformula of } \varphi \rightarrow \psi \} \cup \{ \top_{\mathbb{A}}, \bot_{\mathbb{A}} \}$. By the FEP we obtain a finite algebra $\mathbb{C}$ and injective map $g: B \rightarrow \mathbb{C}$, yielding an interpretation $\widetilde{\llbracket - \rrbracket}$ generated by setting $\widetilde{\llbracket \mathrm{p} \rrbracket} = g(\llbracket \mathrm{p} \rrbracket)$ for all propositional atoms $\mathrm{p}$ occurring in $\varphi \rightarrow \psi$. As $g$ is injective and preserves existing algebraic operations in $B$ this gives that $\widetilde{\llbracket \varphi \rightarrow \psi \rrbracket} <_{\mathbb{C}} \top_{\mathbb{C}}$. Since in layered Heyting algebras $a \leq b$ iff $a \rightarrow b = \top$, we have a finite algebra $\mathbb{C}$ and interpretation $\widetilde{\llbracket - \rrbracket}$ such that $\widetilde{\llbracket \varphi \rrbracket} \not\leq \widetilde{\llbracket \psi \rrbracket}$, witnessing that $\varphi \vdash \psi$ does not hold. As layered Heyting algebras are finitely axiomatised this yields decidability of the logic.

We adapt an argument given by Hanikov\'a \& Hor\u{c}\'{i}k \cite{FEP} to prove the class of bounded residuated distributive-lattice ordered groupoids has the FEP. If such an algebra supports Heyting implication it is a layered Heyting algebra. Thus we simply have to additionally account for Heyting implication to make the desired proof go through.

\begin{thm}[cf. \cite{FEP}]  \label{theorem:fep}
The class of layered Heyting algebras has the FEP.
\end{thm}

\begin{proof} Let $\mathbb{A}$ be a layered Heyting algebra and $B \subseteq \mathbb{A}$ a finite
subset that, \emph{wlog}, contains $\top_{\mathbb{A}}$ and $\bot_{\mathbb{A}}$. Denote by 
$(C, \land_{\mathbb{C}}, \lor_{\mathbb{C}}, \top_{\mathbb{C}}, \bot_{\mathbb{C}})$ the distributive sublattice of the distributive lattice reduct of $\mathbb{A}$ generated by $B$. As $B$ was finite, so too is $C$. Define $a \rightarrow_{\mathbb{C}} b = \bigvee_{\mathbb{C}} \{ c \in C \mid a \land_{\mathbb{C}} c \leq_{\mathbb{C}} b \}$. Since each join is finite this is well defined, and this makes
$(C, \land_{\mathbb{C}}, \lor_{\mathbb{C}}, \rightarrow_{\mathbb{C}}, \top_{\mathbb{C}}, \bot_{\mathbb{C}})$ a Heyting algebra. It can be shown that if $b, b', b \rightarrow_{\mathbb{A}} b' \in B$ then 
$b \rightarrow_{\mathbb{A}} b' = b \rightarrow_{\mathbb{C}} b'$.

The rest of the proof now proceeds as in \cite{FEP}. Define operations $\lambda, \sigma: A \rightarrow A$
by $\lambda(a) = \bigwedge_{\mathbb{C}}\{ c \in C \mid a \leq_{\mathbb{A}} c \}$ and
$\sigma(a) = \bigvee_{\mathbb{C}}\{ c \in C \mid c \leq_{\mathbb{A}} a \}$. These are both well-defined because $C$ is finite. It follows that for
$c \in C$, $\lambda(c) = c = \sigma(c)$. We then define 
$\blacktriangleright_{\mathbb{C}}, \gimp_{\mathbb{C}}, \limp_{\mathbb{C}}$ on $C$ by
$c \blacktriangleright_{\mathbb{C}} c' = \lambda(c \blacktriangleright_{\mathbb{A}} c')$,
$c \gimp_{\mathbb{C}} c' = \sigma(c \gimp_{\mathbb{A}} c')$ and
$c \limp_{\mathbb{C}} c' = \sigma(c \limp_{\mathbb{A}} c')$. The fact that $\lambda$ is a closure operator and $\sigma$ an interior operator can be used to show that the required residuation properties hold for these operations. We thus have a finite layered Heyting algebra
$\mathbb{C} = (C, \land_{\mathbb{C}}, \lor_{\mathbb{C}}, \rightarrow_{\mathbb{C}}, \top_{\mathbb{C}}, \bot_{\mathbb{C}}, \blacktriangleright_{\mathbb{C}}, \gimp_{\mathbb{C}}, \limp_{\mathbb{C}})$, with the inclusion map of $B$ into $C$ satisfying the defining property of the FEP.
\end{proof}

\begin{cor}[Decidability of ILGL] \label{cor:decidability}
The consequence relation $\vdash$ for ILGL is decidable. \qed
\end{cor}

The finite countermodel for an invalid formula is bounded in size: for a formula $\varphi$ with $n$ subformulae, the cardinality of the finite algebra is bounded by $2^{2^n}$. In \cite{FEP} this is improved upon by showing such an algebra can be represented by a poset of join-irreducibles of cardinality $m \leq 2^n - 2$ and a ternary relation $R$ of cardinality $m^3$, where join-irreducibles are those elements that are not equal to $\bot$ and cannot be represented as the join of two distinct, non-$\bot$ elements. We defer to another occasion an in-depth investigation of the computational complexity of the decision procedure.

\section{Duality} \label{sec:duality}

We now extend Theorem \ref{thm:representation} to give a Stone-type \cite{PTJ82} duality result for \logicfont{ILGL}. Such duality theorems elegantly abstract the relationship between syntax and semantics in logic, and support systematic approaches to soundness and completeness theorems. As previously noted, ILGL is the weakest bunched logic with intuitionistic additives. Not only can the other intuitionistic bunched logics be constructed by adding structure to ILGL, but also duality theorems for each intuitionistic bunched logic, including BI and Separation Logic. We defer the presentation of these results to another occasion, although we demonstrate the extension to predicate logic in Section~\ref{sec:predicate}.
We also note that a duality theorem for BI that can be relativized to ILGL has been independently given by Jipsen \& Litak~\cite{JL2018}.

Our duality theorem extends Esakia's duality for Heyting algebras \cite{Esakia}, which we now recall. A \emph{Priestley Space} $\mathcal{X}$ is a structure $\mathcal{X} = (X, \mathcal{O}, \preccurlyeq)$ such that $(X, \mathcal{O})$ is a compact topological space, $\preccurlyeq$ is a preorder and the \emph{Priestley separation axiom} holds: if $x \not\preccurlyeq y$ then there exists a clopen, upward-closed set $C$ such that $x \in C$ and $y \not\in C$. The upward-closed clopen sets $\mathcal{C}_{\preccurlyeq}(\mathcal{X}) = \{ C \in \mathcal{O} \mid C \text{ clopen and for all } x \in C, \text{ if } x \preccurlyeq y \text{ then } y \in C \}$ of a Priestley space carry the structure of a bounded distributive lattice. An \emph{Esakia space} is a Priestley space satisfying the additional property that for each clopen set $C$, the downwards closure
$\{ x \in X \mid \exists y \in C: x \preccurlyeq y \}$ is clopen. This extra property is sufficient for the upward-closed clopen sets to additionally support Heyting implication.

For a Heyting algebra $\mathbb{A}$ define $\neg h_{\mathbb{A}}(a) = \{ F \in Pr(A) \mid a \not\in F \}$. Then $S = \{ h_{\mathbb{A}}(a) \mid a \in A \} \cup \{ \neg h_{\mathbb{A}}(a) \mid a \in A \}$ (where $h_{\mathbb{A}}$ is defined as in Theorem \ref{thm:representation}) forms a subbase that generates a topology $\mathcal{O}_{Pr(A)}$ on the set of prime filters of $\mathbb{A}$, $Pr(A)$. This topology makes $(Pr(A), \mathcal{O}_{Pr(A)}, \subseteq)$ an Esakia space. Conversely, given an Esakia space $\mathcal{X}$ the set of upward-closed clopen sets carries the structure of a Heyting algebra. Now define maps $f_{\mathcal{X}}: \mathcal{X} \rightarrow Pr(\mathcal{C}_{\preccurlyeq}(\mathcal{X}))$ by $f_{\mathcal{X}}(x) = \{ C \in \mathcal{C}_{\preccurlyeq}(\mathcal{X}) \mid x \in C \}$. Then $f_{(-)}$ and $h_{(-)}$ define natural isomorphisms underpinning a dual equivalence of the category of Heyting algebras with the category of Esakia spaces.

We now define the topological dual of layered Heyting algebras, \emph{intuitionistic layered spaces}, which combine the structure of an intuitionistic layered frame with that of an Esakia space, with compatibility conditions on $R$ providing the glue. Intuitonistic layered spaces are closely related to gaggle spaces, the topological dual of distributive gaggles \cite{Dunn}.

\begin{defi}[Intuitionistic Layered Space]
An \emph{intutionistic layered space} is a structure $ \mathcal{X} = (X, \mathcal{O}, \preccurlyeq, R)$ such that 
\begin{enumerate}
\item $(X, \mathcal{O}, \preccurlyeq)$ is an Esakia space;
\item $R$ is a ternary relation over $X$;
\item $\mathcal{C}_{\preccurlyeq}(\mathcal{X})$ is closed under $\blacktriangleright_R, \gimp_R$ and  $\limp_R$;
\item If $Rxyz$ does not hold then there exists $C_1$, $C_2 \in \mathcal{C}_{\preccurlyeq}(\mathcal{X})$ such that $x \in C_1$, $y \in C_2$ and $z \not\in C_1 \blacktriangleright_R C_2$. \qed
\end{enumerate} 
\end{defi}

A morphism of intuitonistic layered spaces is simply a \emph{continuous} intuitionistic layered p-morphism (see Definition \ref{def:pmorphism}). This yields a category $\mathrm{IntLaySp}$. Note that continuity, together with conditions (1) and (2) of Definition \ref{def:pmorphism} define morphisms of Esakia spaces. $\mathrm{IntLaySp}$ is thus a subcategory of the category of Esakia spaces.

We define contravariant functors $\mathrm{F}: \mathrm{LayHeyAlg}^{op} \rightarrow \mathrm{IntLaySp}$ and
$\mathrm{G}: \mathrm{IntLaySp}^{op} \rightarrow \mathrm{LayHeyAlg}$ as follows: 
{\small \[ \begin{cases}  \mathrm{F}(\mathbb{A}) = (Pr(A), \mathcal{O}_{Pr(\mathbb{A})}, \subseteq, R_{Pr(\mathbb{A})}) \\
\mathrm{F}f: \mathrm{F}(\mathbb{A}') \rightarrow \mathrm{F}(\mathbb{A}), X' \mapsto f^{-1}[X']
\end{cases}
\begin{cases} \mathrm{G}(\mathcal{X}) = (\mathcal{C}_{\preccurlyeq}(\mathcal{X}), \cap, \cup, \Rightarrow, X, \emptyset, \blacktriangleright_R, \gimp_R, \limp_R) \\
\mathrm{G}g: \mathrm{G}(\mathcal{X}') \rightarrow \mathrm{G}(\mathcal{X}), C' \mapsto g^{-1}[C']
\end{cases} \] }
\noindent where $\mathcal{O}_{Pr(\mathbb{A})}$ and $\mathcal{C}_{\preccurlyeq}(\mathcal{X})$ are as given in the preceeding discussion on Esakia duality and $\blacktriangleright_R, \gimp_R, \limp_R$ are as defined in Definition \ref{def:layeredcomplexalgebra}. We show the output of these functors gives structures and morphisms of the correct type.

\begin{lem} \label{lem:fwelldefined}
The functor $\mathrm{F}: \mathrm{LayHeyAlg}^{op} \rightarrow \mathrm{IntLaySp}$ is well-defined.
\end{lem}

\begin{proof}
We first verify objects. Let $\mathbb{A}$ be a layered Heyting algebra. By Esakia duality, $(Pr(A), \mathcal{O}_{\mathbb{A}}, \subseteq)$ is an Esakia space. Closure under
$\blacktriangleright_{R_{Pr(\mathbb{A})}}, \gimp_{R_{Pr(\mathbb{A})}}$ and $\limp_{R_{Pr(\mathbb{A})}}$ is given by Theorem \ref{thm:representation} and the fact that, by Esakia duality, every upward-closed clopen set is of the form $h_{\mathbb{A}}(a)$ for $a \in A$. Finally, suppose $R_{Pr(\mathbb{A})}F_0F_1F_2$ does not hold. Then there exist $a \in F_0$ and $b \in F_1$ such that $a \blacktriangleright b \not\in F_2$. The upward-closed clopen sets $h_{\mathbb{A}}(a)$ and $h_{\mathbb{A}}(b)$ then satisfy condition 4.

Next morphisms. Let $f: \mathbb{A} \rightarrow \mathbb{A}'$ be a homomorphism. That $\mathrm{F}(f) = f^{-1}$ maps prime filters to prime filters, is continuous and satisfies properties (1) and (2) of Definition \ref{def:pmorphism} is precisely Esakia duality. It remains to verify properties 3--6. We show 3 and the characteristic example of 5, leaving the other similar cases as an exercise to the reader. First, suppose $R'_{Pr(\mathbb{A})}F'_0F'_1F'_2$ and let $a \in \mathrm{F}(f)(F'_0)$ and $b \in \mathrm{F}(f)(F'_1)$. Then 
$f(a) \in F'_0$ and $f(b) \in F'_1$ so $f(a) \blacktriangleright f(b) = f(a \blacktriangleright b) \in F'_2$. It thus follows that $a \blacktriangleright b \in \mathrm{F}(f)(F'_2)$ so $R_{Pr(\mathbb{A})}\mathrm{F}(f)(F'_0)\mathrm{F}(f)(F'_1)\mathrm{F}(f)(F'_2)$ as required. 

For 5, suppose $R_{Pr(\mathbb{A})}F_3F_1F_2$ and $\mathrm{F}(f)(F'_0) \subseteq F_3$. Consider
$\gamma = (f(\overline{F_3})]$, $\alpha = [f(F_1))$ and $\beta = (f(\overline{F_2})]$. We will use these ideals and filter to prove the existence of the requisite prime filters by showing that the partial order of tuples of proper filters and ideals defined by the satisfaction of the properties required for 5 is non-empty. We can then apply Zorn's lemma and obtain the required prime filters as in Theorem \ref{thm:representation}. First we note an important alternative characterization of $R_{Pr(\mathbb{A})}$:
\[ R_{Pr(\mathbb{A})}F_3F_1F_2 \text{ iff for all } a, b: a \in F_1 \text{ and } b \not\in F_2 \text{ implies } a \gimp b \not\in F_3. \]

We leave it to the reader to verify this holds. With this we show that each of $\gamma$, $\alpha$, $\beta$ is proper. Suppose $\top \in \gamma$. Then there exists $c \not\in F_2$ such that $f(c) = \top$. Hence $c \not\in \mathrm{F}(f)(F'_0)$ so $\top = f(c) \not\in F'_0$, contradicting filterhood of $F'_0$. Suppose $\bot \in \alpha$. Then there exists $a \in F_1$ such that $f(a) = \bot$. Let $b \not\in F_2$ be arbitrary (this is possible because $F_2$ is proper). Then we have $a \gimp b \not\in F_3$, so it follows that $a \gimp b \not\in \mathrm{F}(f)(F'_0)$. But this entails that $f(a \gimp b) = f(a) \gimp f(b) = \top \not\in F'_0$, contradicting filterhood of $F'_0$. Finally, suppose $\top \in \beta$. There thus exists $b \not\in F_2$ such that $f(b) = \top$. Let $a \in F_1$ be arbitrary. Then, similarly to the previous case, $f(a \gimp b) = f(a) \gimp f(b) = \top \not\in F'_0$, contradicting filterhood.

Next we show the defining condition of $R'_{Pr(\mathbb{A})}$ is satisfied by $\overline{\gamma}$, $\alpha$ and $\overline{\beta}$.
Let $a \in \overline{\gamma}$ and $b \in \alpha$. It follows that if $a \leq f(d)$ then $d \in F_3$ and that there exists $b' \in F_1$ such that $f(b') \leq b$. We now show $a \blacktriangleright b \in \overline{\beta}$. Suppose instead that $a \blacktriangleright b \in \beta$. Then there exists $c \not\in F_2$ such that $a \blacktriangleright b \leq f(c)$. We have $a \blacktriangleright f(b') \leq a \blacktriangleright b \leq f(c)$ so $a \leq f(b' \gimp c)$. Thus $b' \gimp c \in F_3$ and so necessarily $c \in F_2$, but this is a contradiction. Hence $a \blacktriangleright b \in \overline{\beta}$, as required.

Finally, we show that the required inclusions hold. That $F_1 \subseteq \mathrm{F}(f)(\alpha)$ and $\mathrm{F}(f)(\overline{\beta}) \subseteq F_2$ is obvious.  For $F'_0 \subseteq \overline{\gamma}$,  suppose $a \in F'_0$ and assume for contradiction that $a \in \gamma$. 
There thus exists $c \not\in F_3$ such that $a \leq f(c)$. By assumption,  
$c \not\in \mathcal{F}(f)(F'_0) \subseteq F_3$ but $a \leq f(c)$ entails $c \in \mathrm{F}(f)(F'_0)$, a contradiction. Hence $a \in \overline{\gamma}$.

By a similar argument to that of Theorem \ref{thm:representation}, there thus exist prime ideals $I^{max}_2$, $I^{max}_3$ and a prime filter $F'_1$ such that $F'_3 = \overline{I^{max}_3}$, $F'_2 = \overline{I^{max}_2}$ and $F'_1$ are prime filters with the property that $R'_{Pr(\mathbb{A'})}F'_3F'_1F'_2$ with $F'_0 \subseteq F'_3$, 
$F_1 \subseteq \mathrm{F}(f)(F'_1)$ and $\mathrm{F}(f)(F'_2) \subseteq F_2$, as required.
\end{proof}

\begin{lem} \label{lem:gwelldefined}
The functor $\mathrm{G}: \mathrm{IntLaySp} \rightarrow \mathrm{LayHeyAlg}$ is well-defined.
\end{lem}

\begin{proof}
Verifying objects: the fact that, for an intuitionistic layered space $\mathcal{X}$, $\mathrm{G}(\mathcal{X})$ 
is a layered Heyting algebra essentially follows from Lemma \ref{lem:comislha} and the fact the upward-closed clopen sets are closed under $\blacktriangleright_{R}, \gimp_{R}, \limp_{R}$.

Verifying morphisms: let $g: \mathcal{X} \rightarrow \mathcal{X}'$ be an intuitionistic layered morphism. That $\mathrm{G}(g)$ maps upward-closed clopen sets to upward-closed clopen sets and respects the Heyting algebra operations is simply Esakia duality. We show that it also respects the operation $\gimp$, leaving the similar $\blacktriangleright$ and $\limp$ cases to the reader. 

First, assume $x \in \mathrm{G}(g)(C_1 \gimp_{R'} C_2)$ and suppose $x \preccurlyeq w$ with $Rwyz$ and $y \in \mathrm{G}(g)(C_1)$. Then $g(x) \preccurlyeq g(w)$ with $R'g(w)g(y)g(z)$ and $g(y) \in C_1$. By assumption, this entails $g(z) \in C_2$ so $z \in \mathrm{G}(g)(C_2)$, as required.

Conversely, assume $x \in \mathrm{G}(g)(C_1) \gimp_R \mathrm{G}(g)(C_2)$ and let 
$g(x) \preccurlyeq' w'$ with $R'w'y'z'$ and $y' \in C_1$. Since $g$ is an intuitionistic layered p-morphism, there exist $w, y, z$ such that $Rwyz$, $x \preccurlyeq w$, $y' \preccurlyeq' g(y)$ and $g(z) \preccurlyeq' z'$. By upwards-closure, $g(y) \in C_1$, so $y \in \mathrm{G}(g)(C_1)$ and by assumption this entails $z \in \mathrm{G}(g)(C_2)$. By upwards-closure again, $z' \in C_2$. Hence $g(x) \in C_1 \gimp_{R'} C_2$, as required.
\end{proof}

\begin{thm}[Duality Theorem for Layered Heyting Algebras] \label{thm:ilglduality}
The categories $\mathrm{LayHeyAlg}$ and $\mathrm{IntLaySp}$ are dually equivalent.
\end{thm}

\begin{proof}
The natural isomorphisms $h_{(-)}$ and $f_{(-)}$ defined in the discussion on Esakia duality together with the functors $\mathrm{F}$ and $\mathrm{G}$
suffice to give the dual equivalence of categories. By Theorem \ref{thm:representation} each $h_{\mathbb{A}}$ is a homomorphism (in fact an isomorphism). To see each $f_{\mathcal{X}}$ is an isomorphism it suffices to show: 
$Rxyz$ iff $R_{Pr(\mathrm{G}(\mathcal{X}))}f_{\mathcal{X}}(x)f_{\mathcal{X}}(y)f_{\mathcal{X}}(z)$. The left-to-right direction is a simple unpacking of the definitions, whilst the right-to-left definition is proved by contraposition, using property 4 in the definition of intuitionistic layered space. Naturality is then immediate from Esakia duality.
\end{proof}

\section{Extension To Predicate ILGL} \label{sec:predicate}
In this final section we extend ILGL to predicate ILGL and prove soundness and completeness with respect to a truth-functional semantics. Our motivation is to provide metatheory for an expressive extension in which real systems modelling can be done. 

As an example, consider the following two predicates: an unary predicate $Contains(-)$ and a binary predicate $- \mapsto -$. Informally, $Contains(r)$ designates that a subgraph contains a resource $r$ and $r \mapsto r'$ designates that a subgraph contains a path from a resource $r$ to a resource $r'$. This signature is styled after Pointer Logic \cite{IO01} (the formulation of Separation Logic \cite{Rey02} as a model of BI) and as such is joined by the standard intuitionistic quantifiers. Note that we do not consider \emph{multiplicative} quantification (see, for example, \cite{Pym09a}): our 
extension is instead akin to the definition of predicate BI given by Biering et al \cite{Biering2005}.

We specify a model of this extended language as follows. Let $\mathcal{X}$ be a bigraph scaffold of the sort defined in Example \ref{ex:bigraph}. Denote by $G$ the union of the place graphs of the system. A \emph{resource assignment} $s$ is a finite partial function $s: Res \rightarrow \mathcal{P}_{\preccurlyeq}(\mathcal{V}(G))$. As the order is determined by the spatial containment encoded by the place graphs, that a resource assignment maps resources to upwards closed sets just means that if a location $x$ contains a resource $r$, and  $x$ is contained in $y$ then $y$ also contains $r$.

Given a resource assignment $s$, together with $r \in Res$, we define $s[r \mapsto A]$ to be the resource assignment that is equal to $s$ everywhere except $r$, where it assigns $r$ to the upward-closed set of vertices $A$. We define a semantics on pairs $(s, H)$ for this extended language in Figure \ref{fig:bigraphsat}.
\begin{figure}
\centering
\hrule
\vspace{1mm}
\setlength\tabcolsep{1.5pt}
\setlength\extrarowheight{2pt}
\begin{tabular}{c c c c l r c c c c r r c}
$s, H$ & $\vDash$ & $\top$ & & always\\
$s, H$ & $\vDash$ & $\bot$ & & never  \\
$s, H$ & $\vDash$ & $\mathrm{Contains}(r)$ & iff & $H$ contains a $s(r)$-vertex \\ 
$s, H$ & $\vDash$ & $r \mapsto r'$ & iff & exists a non-empty path from a $s(r)$-vertex to a $s(r')$-vertex in H \\
$s, H$ & $\vDash$ & $\varphi \land \psi$ & iff & \myalign{l}{$s, H \vDash \varphi$ and $s, H \vDash \psi$} \\
$s, H$ & $\vDash$ & $\varphi \lor \psi$ & iff & $s, H \vDash \varphi$ or  $s, H \vDash \psi$
\\
$s, H$ & $\vDash$ & $\varphi \rightarrow \psi$ & iff & \myalign{l}{for all $K \succcurlyeq H$, $s, K \vDash\varphi$ implies $s, K$  $\vDash$ $\psi$} \\
$s, H $ & $\vDash$ & $\varphi \blacktriangleright \psi$ & iff & \multicolumn{8}{l}{there exists $\gcompE{K_0}{K_1}\downarrow$ s.t. $H \succcurlyeq \gcompE{K_0}{K_1}$, $s, K_0 \vDash \varphi$ and  $s, K_1 \vDash \psi$} \\
$s, H$ & $\vDash$ & $\varphi \gimp \psi$ & iff & \multicolumn{8}{l}{for all $K$ and $L \succcurlyeq H$ s.t. $\gcompE{L}{K} \downarrow$:  $s, L \vDash \varphi$ implies $s, \gcompE{L}{K} \vDash \psi$} \\
$s, H$ & $\vDash$ & $\varphi \limp \psi$ & iff & \multicolumn{8}{l}{for all $K$ and $L \succcurlyeq H$ s.t. $\gcompE{K}{L} \downarrow$: $s, L \vDash \varphi$ implies $s, \gcompE{K}{L} \vDash \psi$} \\
$s, H$ & $\vDash$ & $\exists r\varphi$ & iff & \multicolumn{8}{l}{there exists $A$ s.t. $s[r \mapsto A], H \vDash \varphi$} \\
$s, H$ & $\vDash$ & $\forall r\varphi$ & iff & \multicolumn{8}{l}{for all $A$, $H \preccurlyeq K$ implies $s[r \mapsto A], K \vDash \varphi$}
\end{tabular}
\caption{Satisfaction for bigraph models of predicate ILGL.}
\vspace{1mm}
\hrule
\label{fig:bigraphsat}
\end{figure}

There are many design choices possible here. For example, we could additionally assign weights or permissions to edges, with satisfaction of the $\mapsto$ predicate mediated by conditions on the path. Coupled with a notion of dynamics that evolves assignments and the underlying graph theoretic structure we would have a rich environment within which to model a variety of distributed systems. We defer to another occasion an in-depth investigation of such a framework.

\subsection{Soundness \& completeness for predicate ILGL}
We now prove soundness and completeness of many-sorted predicate ILGL (henceforth predicate ILGL) with respect to a truth-functional semantics and a Hilbert system. The model described in the preceeding discussion is a particular instantiation of this semantics.  We consider a many-sorted language to enable the construction of models where different types of entity are assigned to the graph and because the category theoretic structures we employ easily support it. 

A \emph{signature} for predicate ILGL is comprised of $n$-ary function symbols $f$ of type $X_1 \times \cdots \times X_n \rightarrow X$ and $m$-ary predicate symbols $P$ of type $Y_1, \ldots, Y_m$. Terms of predicate ILGL are given in much the same way as classical predicate logic: each variable $v:X$ is a term of type $X$, and for terms $t_1, \ldots, t_n$ of type $X_1, \ldots, X_n$ respectively and an n-ary function symbol $f$ of type $X_1 \times \cdots \times X_n \rightarrow X$, $f(t_1, \ldots, t_n)$ is a term of type $X$. The syntax of predicate ILGL is then given as follows, where $t_i$, $t$, $t'$ range over the terms of the right type for the predicate symbols:
\[
\varphi \! ::= \! Pt_1\ldots t_m \mid t \!=_X\! t' \mid \top \mid \bot \mid \varphi \wedge \varphi \mid \varphi \vee \varphi \mid 
		\varphi \!\rightarrow\! \varphi \mid 
		\varphi \!\blacktriangleright\! \varphi \mid \varphi \!\gimp\! \varphi \mid \varphi \!\limp\! \varphi \mid \exists v\!\!:\!\!X \varphi \mid \forall v\!\!:\!\!X \varphi
\]

Figure \ref{fig:quantifier_rules_ILGL} gives rules for intuitionistic quantifiers that extend the Hilbert system $\mathrm{ILGL}_{\mathrm{H}}$ of Figure \ref{fig:hilbert_rules_ILGL} to predicate ILGL. These rules have auxillary conditions: in rules $\exists_1$ and $\forall_1$, $v$ must not be free in $\varphi$; in rules $\exists_2$ and $\forall _2$, $t$ must be a term of type $X$ and it must be that no occurence of a variable in $t$ becomes bound when substituted into $\varphi$ . 

\begin{figure}
\centering
\hspace{-.2cm}
{\small 
\begin{tabular}{c}
\hline \\
     \AxiomC{$\varphi \vdash \psi$} \hspace{-2mm}
\RightLabel{$(\exists_1)$}
\UnaryInfC{$\exists v:X \varphi \vdash \psi$}
      \DisplayProof 
     \,  
    \AxiomC{$\varphi \vdash \psi$}
\RightLabel{$(\forall_1)$}
\UnaryInfC{$\varphi \vdash \forall v:X \psi$}  
      \DisplayProof 
       
	\\ \\

         \AxiomC{} \hspace{-2mm}
\RightLabel{$(\exists_2)$}
\UnaryInfC{$\varphi(t) \vdash \exists v :X\varphi(v)$}
      \DisplayProof 
     \,
\AxiomC{}
\RightLabel{$(\forall_2)$}
\UnaryInfC{$\forall v:X \varphi(v) \vdash \varphi(t)$} 
      \DisplayProof 
      
      \\ \\
\hline 
\end{tabular}}
\caption{Extension of $\mathrm{ILGL}_\mathrm{H}$ to predicate ILGL}
\label{fig:quantifier_rules_ILGL}
\end{figure}

Our proof mirrors the structure of the results given in Section~\ref{sec:decidability}. First, we define algebraic structures on which a completeness argument for the Hilbert system can easily be constructed. We then introduce relational structures, of which our intended semantics is a particular instance, and show we can generate an algebraic structure from a relational structure and vice versa, with satisfiability coinciding in each case.

\begin{defi}[ILGL Hyperdoctrine]
An \emph{ILGL hyperdoctrine} is a tuple $(\mathbb{P}: \mathrm{C}^{op} \rightarrow \mathrm{Poset}, (=_X)_{X \in Ob(\mathrm{C})}, (\exists X_{\Gamma}, \forall X_{\Gamma})_{\Gamma, X \in Ob(\mathrm{C})})$
such that:
	\begin{enumerate}
	\item $\mathrm{C}$ is a category with finite products;
	\item $\mathbb{P}: \mathrm{C}^{op} \rightarrow \mathrm{Poset}$ is a contravariant functor;
	\item For each object $X$ and each morphism $f$ in $\mathrm{C}$, $\mathbb{P}(X)$ is 					additionally a layered Heyting algebra and $\mathbb{P}(f)$ a layered Heyting algebra homomorphism
	\item For each object $X$ in $\mathrm{C}$ and each diagonal morphism 
		$\Delta_X: X \rightarrow X \times X$ in $\mathrm{C}$,  
		$=_{X} \in \mathbb{P}(X \times X)$ is such that, for all $a \in \mathbb{P}(X \times X)$, 
		$\top \leq \mathbb{P}(\Delta_X)(a) \text{ iff } =_{X} \leq a$; 
	\item For each pair of objects $\Gamma, X$ in $\mathrm{C}$ and each projection 
		$\pi_{\Gamma, X}: \Gamma \times X \rightarrow \Gamma$ in
		$\mathrm{C}$, $\exists X_{\Gamma}$ and $\forall X_{\Gamma}$ are left and right adjoints to 			$\mathbb{P}(\pi_{\Gamma, X})$ respectively: that is, monotone maps
		$\exists X_{\Gamma}: \mathbb{P}(\Gamma \times X) \rightarrow \mathbb{P}(\Gamma)$ and
		$\forall X_{\Gamma}: \mathbb{P}(\Gamma \times X) \rightarrow \mathbb{P}(\Gamma)$
		such that, for all $a, b \in \mathbb{P}(\Gamma)$, 
	$\exists X_{\Gamma}(a) \leq b \text{ iff } a \leq \mathbb{P}(\pi_{\Gamma, X})(b)$ and 
	$\mathbb{P}(\pi_{\Gamma, X})(b) \leq a \text{ iff } b \leq \forall X_{\Gamma}(a)$. 
		This assignment of maps is additionally natural in $\Gamma$: given a morphism
	$s: \Gamma \rightarrow \Gamma'$, the following diagrams commute: \vspace{1mm}

	{\small \begin{tikzcd}
	\mathbb{P}(\Gamma' \times X) \arrow{r}{\mathbb{P}(s \times id_X)} \arrow{d}[swap]{\exists X_{\Gamma'}} &
	\mathbb{P}(\Gamma \times X) \arrow{d}{\exists X_{\Gamma}} \\
	\mathbb{P}(\Gamma') \arrow{r}[swap]{\mathbb{P}(s)} & \mathbb{P}(\Gamma)
	\end{tikzcd}{\ } {\ } {\ } {\ } \begin{tikzcd}
	\mathbb{P}(\Gamma' \times X) \arrow{r}{\mathbb{P}(s \times id_X)} \arrow{d}[swap]{\forall X_{\Gamma'}} &
	\mathbb{P}(\Gamma \times X) \arrow{d}{\forall X_{\Gamma}} \\
	\mathbb{P}(\Gamma') \arrow{r}[swap]{\mathbb{P}(s)} & \mathbb{P}(\Gamma)
	\end{tikzcd}}
	\end{enumerate} \qed
\end{defi}


For a given signature of predicate ILGL, an interpretation $\llbracket - \rrbracket$ on an ILGL hyperdoctrine $\mathbb{P}: \mathrm{C}^{op} \rightarrow \mathrm{Poset}$ is defined as follows below. For each type $X$, we assign an object $\llbracket X \rrbracket$ in $\mathrm{C}$ and for each n-ary function symbol $f$ of type $X_1 \times \cdots \times X_n \rightarrow X$ in the signature, we assign a morphism $\llbracket f \rrbracket: \llbracket X_1 \rrbracket \times \cdots \times \llbracket X_n \rrbracket \rightarrow \llbracket X \rrbracket$. A context $\Gamma = \{ v_1: X_1, \ldots, v_n: X_n \}$ is then assigned $\llbracket \Gamma \rrbracket = \llbracket X_1 \rrbracket \times \cdots \times \llbracket X_n \rrbracket$ and each term $t$ of type $X$ in context $\Gamma$ is inductively assigned a morphism $\llbracket t \rrbracket: \llbracket \Gamma \rrbracket \rightarrow \llbracket X \rrbracket$: for variables $v_i$, $\llbracket v_i \rrbracket = \pi_i$, the i-th projection from $\llbracket \Gamma \rrbracket$; for terms $t_1, \ldots, t_n$ and n-ary function symbol $f$, 
$\llbracket f(t_1, \ldots, t_n) \rrbracket = \llbracket f \rrbracket \circ \langle \llbracket t_1 \rrbracket, \ldots, \llbracket t_n \rrbracket \rangle.$

We then extend $\llbracket - \rrbracket$ to formulae of predicate ILGL. For each $\varphi$ with free variables in context $\Gamma$ we define an element $\llbracket \varphi \rrbracket \in \mathbb{P}(\llbracket \Gamma \rrbracket)$ inductively. For atomic formulae $Pt_1\ldots t_m$ and $t =_X t'$, $\llbracket Pt_1\ldots t_m \rrbracket$ and $\llbracket t =_X t' \rrbracket$ are given by $\mathbb{P}(\langle \llbracket t_1 \rrbracket, \ldots, \llbracket t_m \rrbracket \rangle)(\llbracket P \rrbracket)$ and 
$\mathbb{P}(\langle \llbracket t \rrbracket, \llbracket t' \rrbracket \rangle)(=_X)$ respectively. We then proceed by structural induction, using the algebraic structure of $\mathbb{P}(\llbracket \Gamma \rrbracket)$: $\llbracket \top \rrbracket = \top$, $\llbracket \bot \rrbracket = \bot$, 
$\llbracket \varphi \circ \psi \rrbracket = \llbracket \varphi \rrbracket \circ \llbracket \psi \rrbracket$ for all $\circ \in \{ \land, \lor, \rightarrow, \blacktriangleright, \gimp, \limp \}$. Assuming $\varphi$ has free variables in $\Gamma \cup \{ x: X \}$, 
$\llbracket Q x:X \varphi \rrbracket = 
Q \llbracket X \rrbracket_{\llbracket \Gamma \rrbracket}( \llbracket \varphi \rrbracket)$ for $Q \in \{ \exists, \forall \}$. Finally, term substitution is given by 
$\llbracket \varphi(t/x) \rrbracket = \mathbb{P}(\llbracket t \rrbracket)(\llbracket \varphi \rrbracket)$.

A formula $\varphi$ with free variables in $\Gamma$ is satisfied under interpretation $\llbracket - \rrbracket$ if $\llbracket \varphi \rrbracket$ is $\top$ in $\mathbb{P}(\llbracket \Gamma \rrbracket)$. It is valid if it is satisfied under all interpretations. The soundness and completeness of predicate ILGL for this notion of validity is an immediate consequence of the analogous result for BI and BI hyperdoctrines.

\begin{thm}[cf. \cite{Biering2005}]
For all formulae $\varphi$ of predicate ILGL, $\varphi$ is provable iff $\varphi$ is valid on ILGL hyperdoctrines. \qed
\end{thm}

We now define the relational structures on which we give a truth-functional semantics, \emph{indexed layered frames}.

\begin{defi}[Indexed Layered Frame] \label{def:indlayeredframe}
An \emph{indexed layered frame} is a functor 
$\mathcal{R}: \mathrm{C} \rightarrow \mathrm{IntLayFr}$ such that
\begin{enumerate}
\item $\mathrm{C}$ is a category with finite products;
\item (Psuedo Epi) For all objects $\Gamma, \Gamma'$ and $X$ in $\mathrm{C}$, all morphisms
		$s: \Gamma \rightarrow \Gamma'$ and all projections
		$\pi_{\Gamma, X}: \Gamma \times X \rightarrow \Gamma$, 
		if $\mathcal{R}(\pi_{\Gamma', X})(y) \preccurlyeq \mathcal{R}(s)(x)$ then there exists $z$ such that
		$\mathcal{R}(\pi_{\Gamma, X})(z) \preccurlyeq x$ and
		$y \preccurlyeq \mathcal{R}(s \times id_X)(z)$.
\end{enumerate} \qed
Given an  indexed layered frame $\mathcal{R}: \mathrm{C} \rightarrow \mathrm{IntLayFr}$
and an object $X$ in $\mathrm{C}$ we denote the intuitionistic layered frame at $X$ by 
$\mathcal{R}(X) = (\mathcal{R}(X), \preccurlyeq_{\mathcal{R}(X)}, R_{\mathcal{R}(X)})$.
\end{defi}


An interpretation for an indexed layered frame is given in precisely the same way as  for an ILGL hyperdoctrine, but for one key difference: each m-ary predicate symbol $P$ of type $X_1, \ldots, X_m$ is interpreted as an upward-closed set $\llbracket P \rrbracket \in \mathcal{P}_{\preccurlyeq}(\mathcal{R}(\llbracket X_1 \rrbracket, \ldots \llbracket X_m \rrbracket))$. For $\varphi$ with free variables in $\Gamma$, an interpretation $\llbracket - \rrbracket$, and $x \in \mathcal{R}(\llbracket \Gamma \rrbracket)$, the satisfaction relation $\vDash^{\Gamma}$ is defined in Figure \ref{fig:ilfsat}.

\begin{figure}
\centering
\hrule
\vspace{1mm}
\setlength\tabcolsep{1.3pt}
\setlength\extrarowheight{2pt}
\begin{tabular}{c c c c l r c c c c r r c}
$x, \llbracket - \rrbracket$ & $\vDash^{\Gamma}$ & $\top$ & & always\\
$x, \llbracket - \rrbracket$ & $\vDash^{\Gamma}$ & $\bot$ & & never  \\
$x, \llbracket - \rrbracket$ & $\vDash^{\Gamma}$ & $Pt_1\ldots t_m$ & iff & $\mathcal{R}(\langle \llbracket t_1 \rrbracket, \ldots, \llbracket t_m \rrbracket \rangle)(x) \in \llbracket P \rrbracket$ \\ 
$x, \llbracket - \rrbracket$ & $\vDash^{\Gamma}$ & $t =_X t'$ & iff & $\mathcal{R}(\langle \llbracket t \rrbracket, \llbracket t' \rrbracket \rangle)(x) \in Ran(\mathcal{R}(\Delta_{\llbracket X \rrbracket}))$ \\
$x, \llbracket - \rrbracket$ & $\vDash^{\Gamma}$ & $\varphi \land \psi$ & iff & \myalign{l}{$x, \llbracket - \rrbracket \vDash^{\Gamma} \varphi$ and $x, \llbracket - \rrbracket \vDash^{\Gamma} \psi$} \\
$x, \llbracket - \rrbracket$ & $\vDash^{\Gamma}$ & $\varphi \lor \psi$ & iff & $x, \llbracket - \rrbracket \vDash^{\Gamma} \varphi$ or  $x, \llbracket - \rrbracket \vDash^{\Gamma} \psi$
\\
$x, \llbracket - \rrbracket$ & $\vDash^{\Gamma}$ & $\varphi \rightarrow \psi$ & iff & \myalign{l}{for all $y \succcurlyeq x$, $s, y \vDash^{\Gamma}\varphi$ implies $s, y$  $\vDash^{\Gamma}$ $\psi$} \\
$x, \llbracket - \rrbracket $ & $\vDash^{\Gamma}$ & $\varphi \blacktriangleright \psi$ & iff & \multicolumn{8}{l}{there exists $x'\!, \!y,\!z$ s.t. $x' \preccurlyeq x$, $R_{\mathcal{R}(\llbracket \Gamma \rrbracket)}yzx'$,
$y, \!\llbracket - \rrbracket \! \vDash^{\Gamma} \!\varphi$ and $z, \!\llbracket - \rrbracket \! \vDash^{\Gamma} \!\psi$} \\
$x, \llbracket - \rrbracket$ & $\vDash^{\Gamma}$ & $\varphi \gimp \psi$ & iff & \multicolumn{8}{l}{for all $x'\!, \!y, \!z$: 
$x \preccurlyeq x'$, $R_{\mathcal{R}(\llbracket \Gamma \rrbracket)}x'yz$ and $y, \!\llbracket - \rrbracket \!\vDash^{\Gamma} \!\varphi$
implies $z,\! \llbracket - \rrbracket  \!\vDash^{\Gamma}\! \psi$} \\
$x, \llbracket - \rrbracket$ & $\vDash^{\Gamma}$ & $\varphi \limp \psi$ & iff & \multicolumn{8}{l}{for all $x'\!, \!y, \!z$: 
$x \preccurlyeq x'$, $R_{\mathcal{R}(\llbracket \Gamma \rrbracket)}yx'z$ and $y, \!\llbracket - \rrbracket \!\vDash^{\Gamma} \!\varphi$
implies $z,\! \llbracket - \rrbracket  \!\vDash^{\Gamma}\! \psi$} \\
$x, \llbracket - \rrbracket$ & $\vDash^{\Gamma}$ & $\exists r\varphi$ & iff & \multicolumn{8}{l}{there exists $x'$
s.t. $\mathcal{R}(\pi_{\llbracket \Gamma \rrbracket, \llbracket X \rrbracket})(x') = x$
and $x', \llbracket - \rrbracket  \vDash^{\Gamma \cup \{ v : X \}} \varphi$ } \\
$x, \llbracket - \rrbracket$ & $\vDash^{\Gamma}$ & $\forall r\varphi$ & iff & \multicolumn{8}{l}{for all $x' $:
if $x \preccurlyeq \mathcal{R}(\pi_{\llbracket \Gamma \rrbracket, \llbracket X \rrbracket})(x')$
then $x', \llbracket - \rrbracket  \vDash^{\Gamma \cup \{ v:X\}} \varphi$}
\end{tabular}
\caption{Satisfaction for indexed layered frame models of predicate ILGL.}
\vspace{1mm}
\hrule
\label{fig:ilfsat}
\end{figure}
The bigraph model of predicate ILGL defined in the previous subsection is an example of an indexed layered frame over the category $\mathrm{Set}$.

\begin{exa}[The bigraph model as an indexed layered frame]
First, let $(X, R, \preccurlyeq)$ be the layered frame corresponding to the bigraph scaffold $\mathcal{X}$: $RHKL$ iff $\gcompE{H}{K}\downarrow$ and $\gcompE{H}{K} = L$. Set $\mathcal{R}: \mathrm{Set} \rightarrow \mathrm{IntLayFr}$ by $\mathcal{R}(A) = (A \times X, R_A, \preccurlyeq_A)$ where $R_A(x, H)(y, K)(z, L)$ iff $x = y = z$ and
$RHKL$, and $(x, H) \preccurlyeq_A (y, K)$ iff $x = y$ and $H \preccurlyeq K$. For functions $f: A \rightarrow B$ set $\mathcal{R}(f)(a, H) = (f(a), H)$. This defines a functor $\mathcal{R}$ which trivially satisfies (Psuedo Epi). Hence $\mathcal{R}$ is an indexed layered frame.

We define the following interpretation. The single sort is interpreted as $\mathcal{P}_{\preccurlyeq}(V(G))$ where $G$ is the graph union of the place graph vertices of the system of bigraphs. The predicate symbols are interpreted as $\llbracket Contains(-) \rrbracket = \{ (A, H) \mid \exists x \in A: x \in V(H) \}$ and
\[ \llbracket \mapsto \rrbracket = \{ ((A_1, A_2), H) \mid \exists x_1 \in A_1 \text{ and } x_2 \in A_2: H \text{ contains non-empty path } x_1 \text{ to } x_2 \}. \] 
Let $r_i$ be an enumeration of resources in $Res$ and let $\varphi$ be a formula with free variables amongst $r_1, \ldots, r_n$. Then
\begin{align*}\{ (r_1, A_1), \ldots, (r_n, A_n) \}, G \vDash \varphi \text{ iff } 
((A_1, \ldots, A_n), G), \llbracket - \rrbracket \vDash^{\{r_1, \ldots, r_n\}} \varphi.\tag*{\qed}\end{align*}
\end{exa}

It remains to show satisfibility coincides for ILGL hyperdoctrines and indexed layered frames. We first obtain an ILGL hyperdoctrine from an indexed layered frame.

\begin{defi}[Complex Hyperdoctrine]
Given an indexed layered frame $\mathcal{R}: \mathrm{C} \rightarrow \mathrm{IntLayFr}$, the \emph{complex hyperdoctrine} of $\mathcal{R}$, $Com(\mathcal{R}(-)): \mathrm{C}^{op} \rightarrow \mathrm{Poset}$, is defined by extending $Com(\mathcal{R}(-))$ (Definition \ref{def:layeredcomplexalgebra}) to morphisms by
$Com(\mathcal{R}(f)) = (\mathcal{R}(f))^{-1}$ and setting $Ran(\mathcal{R}(\Delta_X))$ as $=_X$, $\mathcal{R}(\pi_{\Gamma, X})^*$ as $\exists X_{\Gamma}$, and $\mathcal{R}(\pi_{\Gamma, X})_*$ as $\forall X_{\Gamma}$, where
	$\mathcal{R}(\pi_{\Gamma, X})^*(A) = \{ x \mid \text{there exists } y \in A: \mathcal{R}(\pi_{\Gamma, X})(y) \preccurlyeq x \}$ and
	$\mathcal{R}(\pi_{\Gamma, X})_*(A) = \{ x \mid \text{for all } y, \text{ if } x \preccurlyeq \mathcal{R}(\pi_{\Gamma, X})(y) \text{ then } y \in A \}$ \qed
\end{defi}
Given that the complex operations are designed to mimic the semantic clauses, one might expect the direct image $\mathcal{R}(\pi_{\Gamma, X})$ to be given as $\exists X_{\Gamma}$ rather than $\mathcal{R}(\pi_{\Gamma, X})^*$. We note that these define precisely the same functions because $\mathcal{R}(\pi_{\Gamma, X})$ is an intuitionistic layered p-morphism: property (2) of Definition \ref{def:pmorphism} easily shows this to be the case. The formulation of $\mathcal{R}(\pi_{\Gamma, X})^*$ simplifies many proofs in what follows, however.

\begin{lem} \label{lemma:comlexisilf}
Given an indexed layered frame $\mathcal{R}: \mathrm{C} \rightarrow \mathrm{IntLayFr}$, the 
complex hyperdoctrine $Com(\mathcal{R}(-))$ is an ILGL hyperdoctrine.
\end{lem}

\begin{proof}
By Lemma \ref{lem:comislha} $Com(\mathcal{R}(X))$ is a layered Heyting algebra, and by a similar argument to ILGL duality each $Com(\mathcal{R}(f))$ is a layered Heyting algebra homomorphism. As $Ran(\mathcal{R}(\Delta_X))$ is upward-closed by property (2) of Definition \ref{def:pmorphism}, it is an element of $Com(\mathcal{R}(X \times X))$, as required. We also have that $\mathcal{R}(\pi_{\Gamma, X})^*$ and $\mathcal{R}(\pi_{\Gamma, X})_*$ map upward-closed sets to upward-closed sets and are monotone with respect to $\subseteq$. The adjointness properties of $=_X$, $\exists X_{\Gamma}$ and $\forall X_{\Gamma}$ are then all easily verified.

We thus concentrate on the remaining property, naturality. We restrict ourselves to the case for
$\exists X_{\Gamma}$. Let $s: \Gamma \rightarrow \Gamma'$ be a morphism in $\mathrm{C}$. and suppose $A \in Com(\mathcal{R}(\Gamma' \times X))$. We must show
$\mathcal{R}(\pi_{\Gamma, X})^*(\mathcal{R}(s \times id_X)^{-1}(A)) = 
\mathcal{R}(s)^{-1}(\mathcal{R}(\pi_{\Gamma', X})*(A))$. Suppose 
$x \in \mathcal{R}^*(\pi_{\Gamma, X})(\mathcal{R}(s \times id_X)^{-1}(A))$: then there exists $y$ such that $\mathcal{R}(\pi_{\Gamma, X})(y) \preccurlyeq x$ and $\mathcal{R}(s \times id_X)(y) \in A$. We have $\mathcal{R}(\pi_{\Gamma', X})(\mathcal{R}(s \times id_X)(y)) =
\mathcal{R}(s)(\mathcal{R}(\pi_{\Gamma, X})(y)) \preccurlyeq \mathcal{R}(s)(x)$. Hence
$x \in \mathcal{R}(s)^{-1}(\mathcal{R}(\pi_{\Gamma', X})^*(A))$, as required.


Conversely, assume $x \in \mathcal{R}(s)^{-1}(\mathcal{R}(\pi_{\Gamma', X})^*(A))$.  Then there exists $y \in A$ such that $\mathcal{R}(\pi_{\Gamma', X})(y) \preccurlyeq \mathcal{R}(s)(x)$. Then by (Psuedo Epi), there exists $z$ such that $\mathcal{R}(\pi_{\Gamma, X})(z) \preccurlyeq x$ and $y \preccurlyeq \mathcal{R}(s \times id_X)(z)$. By upwards-closure of $A$, $\mathcal{R}(s \times id_X)(z) \in A$. Hence
$x \in \mathcal{R}(\pi_{\Gamma, X})^*(\mathcal{R}(s \times id_X)^{-1}(A))$, as required. 
\end{proof}

As in the cases for the propositional logics, an interpretation $\llbracket - \rrbracket$ on an indexed resource frame immediately yields an interpretation on its complex hyperdoctrine, as each predicate symbol is assigned an upward-closed subset. We overload the notation $\llbracket - \rrbracket$ to also refer to the interpretation on the complex hyperdoctrine because of this. By unpacking the definition of the semantic clauses and noting the way they coincide with the analogous structure of the complex hyperdoctrine, we can show satisfiability coincides in this case.

\begin{lem}
Given an indexed layered frame $\mathcal{R}$ and an interpretation $\llbracket - \rrbracket$, for all formulae $\varphi$ of predicate ILGL in context $\Gamma$, and all $x \in \mathcal{R}(\llbracket \Gamma \rrbracket)$ we have $x \in \llbracket \varphi \rrbracket$ iff $x, \llbracket - \rrbracket \vDash^{\Gamma} \varphi$. \qed
\end{lem}

\begin{defi}[Indexed Prime Filter Frame] 
Given an ILGL hyperdoctrine $\mathbb{P}: \mathrm{C}^{op} \rightarrow \mathrm{Poset}$, the \emph{indexed prime filter frame} $Pr(\mathbb{P}(-)): \mathrm{C} \rightarrow \mathrm{IntResFr}$ is given by extending $Pr(\mathbb{P}(-))$ (Definition \ref{def:ilglprimefilter}) to morphisms with $Pr(\mathbb{P}(f)) = (\mathbb{P}(f))^{-1}$. \qed
\end{defi}

\begin{lem} \label{lem:indexedprimefilter}
Given an ILGL hyperdoctrine $\mathbb{P}: \mathrm{C}^{op} \rightarrow \mathrm{Poset}$, the 
indexed prime filter frame $Pr(\mathbb{P}(-))$ is an indexed layered frame.
\end{lem}

\begin{proof}
We show (Psuedo Epi) holds. Assume we have objects $\Gamma, \Gamma'$ and $X$ in $\mathrm{C}$ and a morphism
$s: \Gamma \rightarrow \Gamma'$. Let prime filters $F_0$ and $F_1$ be such that
$\mathbb{P}(\pi_{\Gamma', X})^{-1}(F_1) \subseteq \mathbb{P}(s)^{-1}(F_0)$. As in Theorem \ref{thm:representation} we can prove the existence of the required prime filter by demonstrating the partial order of proper filters satisfying the property is non-empty and invoking Zorn's Lemma to yield a maximum that can be shown to be prime. 

Consider the filter $\alpha = [ \mathbb{P}(s \times id_X)(F_1) )$. Suppose for contradiction that $\alpha$ is not proper. Then there exists $a \in F_1$ such that $\mathbb{P}(s \times id_X)(a) = \bot$.  By adjointness $\exists X_{\Gamma}(\bot) = \bot$ so 
$\mathbb{P}(s)(\exists X_{\Gamma'}(a)) = \exists X_{\Gamma}(\mathbb{P}(s \times id_X)(a)) = \bot$.
This entails $\exists X_{\Gamma'}(a) \not\in \mathbb{P}(s)^{-1}(F_0)$ so 
$\exists X_{\Gamma'}(a) \not\in \mathbb{P}(\pi_{\Gamma', X})^{-1}(F_1)$ by assumption. However by adjointness and filterhood, $\mathbb{P}(\pi_{\Gamma', X})(\exists X_{\Gamma'}(a)) \in F_1$, a contradiction.


Clearly $\mathbb{P}(s \times id_X)^{-1}(\alpha) \subseteq F_1$. To see, the other inclusion holds, let
$a \in \mathbb{P}(\pi_{\Gamma, X})^{-1}(\alpha)$. Then there exists $b \in F_1$ such that
$\mathbb{P}(s \times id_X)(b) \leq \mathbb{P}(\pi_{\Gamma, X})(a)$. By adjointness
$\exists X_{\Gamma}(\mathbb{P}(s \times id_X)(b)) \leq a$ and so by naturality
$\mathbb{P}(s)(\exists X_{\Gamma'}(b)) \leq a$. Since $\mathbb{P}(\pi_{\Gamma', X})(\exists X_{\Gamma'}(b)) \in F_1$, we have $\exists X_{\Gamma'}(b) \in \mathbb{P}(\pi_{\Gamma', X})^{-1}(F_1) \subseteq \mathbb{P}(s)^{-1}(F_0)$. Thus by filterhood, $a \in F_0$.
\end{proof}

Given an interpretation $\llbracket - \rrbracket$ on an ILGL hyperdoctrine $\mathbb{P}$ we can obtain 
an interpretation $\widetilde{\llbracket - \rrbracket}$ on its indexed prime filter frame $Pr(\mathbb{P}(-))$: $\widetilde{\llbracket - \rrbracket}$ is defined identically to $\llbracket - \rrbracket$ everywhere except on predicate symbols, where for an m-ary predicate symbol $P$ of type $X_1, \ldots, X_m$, 
$\widetilde{\llbracket P \rrbracket} =
h_{\mathbb{P}(\llbracket X_1 \rrbracket \times \cdots \times \llbracket X_m \rrbracket)}(\llbracket P \rrbracket)$.

\begin{lem} \label{lem:ipfsatisfiability}
Given an ILGL hyperdoctrine $\mathbb{P}: \mathrm{C}^{op} \rightarrow \mathrm{Poset}$ and interpretation $\llbracket - \rrbracket$, for all formulae $\varphi$ of predicate ILGL in context $\Gamma$ and $F \in Pr(\mathbb{P}(\llbracket \Gamma \rrbracket))$, 
$\llbracket \varphi \rrbracket \in F \text{ iff }
F, \widetilde{\llbracket - \rrbracket} \vDash^{\Gamma} \varphi.$ 
\end{lem}

\begin{proof}
Most of the cases follow immediately from the representation theorem for ILGL. Of the new cases, we attend to universal quantification. The others are shown by proving the existence of prime filter witnesses in much the same way as Theorem \ref{thm:representation}. 

For the case we consider, it suffices to prove $\forall \llbracket X \rrbracket_{\llbracket \Gamma \rrbracket}(\llbracket \varphi \rrbracket) \in F$ iff for all prime filters $G$, if $F \subseteq \mathbb{P}(\pi_{\llbracket \Gamma \rrbracket, \llbracket X \rrbracket})^{-1}(G)$ then $\llbracket \varphi \rrbracket \in G$. Assume $\forall \llbracket X \rrbracket_{\llbracket \Gamma \rrbracket}(\llbracket \varphi \rrbracket) \in F$ and suppose $F \subseteq \mathbb{P}(\pi_{\llbracket \Gamma \rrbracket, \llbracket X \rrbracket})^{-1}(G)$. 
Then $\mathbb{P}(\pi_{\llbracket \Gamma \rrbracket, \llbracket X \rrbracket})(\forall \llbracket X \rrbracket_{\llbracket \Gamma \rrbracket}(\llbracket \varphi \rrbracket)) \in G$. By adjointness and upwards-closure, $\llbracket \varphi \rrbracket \in G$. In the other direction, suppose 
$\forall \llbracket X \rrbracket_{\llbracket \Gamma \rrbracket}(\llbracket \varphi \rrbracket) \not\in F$.
Then we can prove the existence of a prime filter $G$ satisfying $F \subseteq \mathbb{P}(\pi_{\llbracket \Gamma \rrbracket, \llbracket X \rrbracket})^{-1}(G)$ and $\llbracket \varphi \rrbracket \not\in G$. 
Consider $\alpha = ( \llbracket \varphi \rrbracket ]$. We can assume this is proper because otherwise
$\llbracket \varphi \rrbracket = \top$ and $\forall \llbracket X \rrbracket_{\llbracket \Gamma \rrbracket}(\llbracket \varphi \rrbracket) = \top \not\in F$, a contradiction. We have that for any $a \in F$, 
by adjointness $\mathbb{P}(\pi_{\Gamma, X})(b) \not\in \alpha$, since otherwise
$b \leq \forall \llbracket X \rrbracket_{\llbracket \Gamma \rrbracket}(\llbracket \varphi \rrbracket) \in F$.
By an argument similar to that of Theorem \ref{thm:representation}, there therefore exists a prime ideal $P$ such that $F \subseteq \mathbb{P}(\pi_{\Gamma, X})^{-1}(\overline{P})$ and 
$\llbracket \varphi \rrbracket \not\in \overline{P}$.
\end{proof}
We thus obtain the equivalence of the hyperdoctrinal and truth-functional semantics. This additionally yields completeness of the indexed resource frame semantics with respect to the predicate ILGL Hilbert system.

\begin{thm}
For all formulae $\varphi$ of predicate ILGL: $\varphi$ is satisfiable (valid) on ILGL hyperdoctrines iff $\varphi$ is satisfiable (valid) on indexed layered frames. \qed
\end{thm}

\begin{cor}
For all formuale $\varphi$ of predicate ILGL, $\varphi$ is provable iff $\varphi$ is valid on indexed layered frames. \qed
\end{cor}

\subsection{Duality} \label{subsec:duality}
To finish we extend this correspondence to a dual equivalence of categories, in the same manner as the propositional case. We denote by $U: \mathrm{IntLaySp} \rightarrow \mathrm{IntLayFr}$ the forgetful functor that drops topology.

\begin{defi}[Indexed Layered Space]
An \emph{indexed layered space} is a functor $\mathcal{R}: \mathrm{C} \rightarrow \mathrm{IntLaySp}$ such that the following properties hold:
	\begin{enumerate}
	\item The composition with the forgetful functor, $U \circ \mathcal{R}: \mathrm{C} \rightarrow 				\mathrm{IntLayFr}$, is an indexed layered frame;
	\item For each object $X$ in $\mathrm{C}$, $Ran(\mathcal{R}(\Delta_X))$ is clopen;
	\item For each pair of objects $\Gamma$ and $X$ in $\mathrm{C}$, 
		$\mathcal{R}(\pi_{\Gamma, X})^*$ and 
		$\mathcal{R}(\pi_{\Gamma, X})_*$ map upward-closed clopen sets to upward-closed 				clopen sets.
	\end{enumerate}
Given an  indexed layered space $\mathcal{R}: \mathrm{C} \rightarrow \mathrm{IntResFr}$
and an object $X$ in $\mathrm{C}$ we denote the intuitionistic layered space at $X$ by 
$\mathcal{R}(X) = (\mathcal{R}(X), \mathcal{O}_{\mathcal{R}(X)}, \preccurlyeq_{\mathcal{R}(X)}, R_{\mathcal{R}(X)})$. \qed
\end{defi}

We now define morphisms to obtain categories of ILGL hyperdoctrines and indexed layered spaces. Our definition of hyperdoctrine morphism adapts that given for \emph{coherent} hyperdoctrines in \cite{COUMANS20121940}.

\begin{defi}[ILGL Hyperdoctrine Morphism]
Given ILGL hyperdoctrines $\mathbb{P}: \mathrm{C}^{op} \rightarrow \mathrm{Poset}$ and 
$\mathbb{P}': \mathrm{D}^{op} \rightarrow \mathrm{Poset}$, an \emph{ILGL hyperdoctrine morphism} 
$(\mathrm{K}, \tau): \mathbb{P} \rightarrow \mathbb{P}'$ is a pair satisfying the following properties:
	\begin{enumerate}
	\item $\mathrm{K}: \mathrm{C} \rightarrow \mathrm{D}$ is a finite product preserving functor;
	\item $\tau: \mathbb{P} \rightarrow \mathbb{P}' \circ K$ is a natural transformation;
	\item For all objects $X$ in $\mathrm{C}$: $\tau_{X \times X}(=_X) ={\ } ='_{K(X)}$;
	\item For all objects $\Gamma$ and $X$ in $\mathrm{C}$, the following squares commute:

		{\small \begin{tikzcd}	
		\mathbb{P}(\Gamma \times X) \arrow{r}{\tau_{\Gamma \times X}}
		\arrow{d}[swap]{\exists X_{\Gamma}} & \mathbb{P'}(K(\Gamma) \times K(X))
		\arrow{d}{\exists' K(X)_{K(\Gamma)}} \\
		\mathbb{P}(\Gamma) \arrow{r}{\tau_{\Gamma}} & \mathbb{P'}(K(\Gamma))
		\end{tikzcd} 
		\begin{tikzcd}
		\mathbb{P}(\Gamma \times X) \arrow{r}{\tau_{\Gamma \times X}}
		\arrow{d}[swap]{\forall X_{\Gamma}} & \mathbb{P'}(K(\Gamma) \times K(X))
		\arrow{d}{\forall' K(X)_{K(\Gamma)}} \\
		\mathbb{P}(\Gamma) \arrow{r}{\tau_{\Gamma}} & \mathbb{P'}(K(\Gamma))
		\end{tikzcd}}
	\end{enumerate}
The composition of ILGL hyperdoctrine morphisms $(K, \tau): \mathbb{P} \rightarrow \mathbb{P}'$ and $(K', \tau'): \mathbb{P}' \rightarrow \mathbb{P}''$ is given by $(K' \circ K, \tau'_{K(-)} \circ \tau)$. This yields a category $\mathrm{ILGLHyp}$. \qed
\end{defi}

\begin{defi}[Indexed Layered Space Morphism]
Given indexed layered spaces $\mathcal{R}: \mathrm{C} \rightarrow \mathrm{IntLaySp}$ and 
$\mathcal{R}': \mathrm{D} \rightarrow \mathrm{IntLaySp}$, an \emph{indexed layered space morphism}
$(L, \lambda): \mathcal{R} \rightarrow \mathcal{R}'$ is a pair $(L, \lambda)$ such that
	\begin{enumerate}
	\item $L: D \rightarrow C$ is a finite product preserving functor;
	\item $\lambda: \mathcal{R} \circ L \rightarrow \mathcal{R}'$ is a natural transformation;
	\item (Lift Property) If there exists $x$ and $y$ 
		such that $\mathcal{R}'(\Delta_X)(y) \preccurlyeq \lambda_{X \times X}(x)$ then there exists
		$y'$ such that $\mathcal{R}(\Delta_{L(X)})(y') 
		\preccurlyeq x$;
	\item (Morphism Pseudo Epi)  If there exists $x$ and 
		$y$ with 
		$\mathcal{R}'(\pi_{\Gamma, X})(x) \preccurlyeq \lambda_{\Gamma}(y) $ then there exists
		$z$ such that
		$x \preccurlyeq\lambda_{\Gamma \times X}(z)$ and
		$\mathcal{R}(\pi_{L(\Gamma), L(X)})(z) \preccurlyeq y$.
	\end{enumerate}
The composition of indexed layered space morphisms $(L', \lambda'): \mathcal{R}' \rightarrow \mathcal{R}''$ and $(L, \lambda): \mathcal{R} \rightarrow \mathcal{R}'$ is given by $(L \circ L', \lambda' \circ \lambda_{L'(-)})$. This yields a category $\mathrm{IndLaySp}$. \qed
\end{defi}

We now define functors $\mathfrak{F}: \mathrm{ILGLHyp} \rightarrow \mathrm{IndLaySp}$ and 
$\mathfrak{G}: \mathrm{IndLaySp} \rightarrow \mathrm{ILGLHyp}$. The functors 
$\mathrm{F}: \mathrm{LayHeyAlg} \rightarrow \mathrm{IntLaySp}$ and $\mathrm{G}: \mathrm{IntLaySp} \rightarrow \mathrm{LayHeyAlg}$ are as in ILGL duality in Section \ref{sec:duality} and $X$ and $\Gamma$ range over objects of $\mathrm{C}$.
%
%

{\small \[ \begin{cases}  \mathfrak{F}(\mathbb{P}) = \mathrm{F} \circ \mathbb{P} \\
\mathfrak{F}f: \mathfrak{F}(\mathbb{P}') \rightarrow \mathfrak{F}(\mathbb{P}), (K, \tau) \mapsto (K, \tau^{-1})
\end{cases}
 \begin{cases}\mathfrak{G}(\mathcal{R}) = (\mathrm{G} \circ \mathcal{R}, 
		Ran(\mathcal{R}(\Delta_X)), 
		\mathcal{R}(\pi_{\Gamma, X})^*, \mathcal{R}(\pi_{\Gamma, X})_*) \\
\mathfrak{G}g: \mathfrak{G}(\mathcal{R}') \rightarrow \mathrm{G}(\mathcal{R}), (L, \lambda) \mapsto (L, \lambda^{-1})
\end{cases}\] }


\begin{lem}
The functor $\mathfrak{F}: \mathrm{ILGLHyp} \rightarrow \mathrm{IndLaySp}$ is well defined.
\end{lem}

\begin{proof}
We first show $\mathfrak{F}$ is well defined on objects. Let $\mathbb{P}$ be a ILGL hyperdoctrine. We show $\mathfrak{F}(\mathbb{P})$ is an indexed layered space. As $U \circ F \circ \mathbb{P} = Pr(\mathbb{P}(-))$, Lemma \ref{lem:indexedprimefilter} suffices to give 1). For 2), it can easily be shown that adjointness of $=_X$ yields
$Ran(\mathbb{P}(\Delta_X)^{-1}) = h_{\mathbb{P}(X \times X)}(=_X)$, a clopen set. For 3), we note that by Esakia duality the upward-closed clopen sets in the domain of $(\mathcal{P}(\pi_{\Gamma, X})^{-1})^*$ and $(\mathcal{P}(\pi_{\Gamma, X})^{-1})_*$ are all of the form 
$h_{\mathbb{P}(\Gamma \times X)}(a)$. A simple application of the adjointness properties shows that
$(\mathcal{P}(\pi_{\Gamma, X})^{-1})^*(h_{\mathbb{P}(\Gamma \times X)}(a)) = h_{\mathbb{P}(\Gamma)}(\exists X_{\Gamma}(a))$ and $(\mathcal{P}(\pi_{\Gamma, X})^{-1})_*(h_{\mathbb{P}(\Gamma \times X)}(a)) = h_{\mathbb{P}(\Gamma)}(\forall X_{\Gamma}(a))$.


Now let $(K, \tau)$ be an ILGL hyperdoctrine morphism. To see $\mathfrak{F}(K, \tau)$ is an indexed resource space morphism, note that $K$ is finite product preserving by definition and naturality of $\tau^{-1}$ is inherited. By ILGL duality,  each component is an intuitionistic layered space morphism. Finally, (Lift Property) and (Morphism Pseudo Epi) are verified in much the same way as (Psuedo Epi) in Lemma \ref{lem:indexedprimefilter}, using property (3) of the hyperdoctrine morphism definition for the former and (4) for the latter.
\end{proof}

\begin{lem}
The functor $\mathfrak{G}: \mathrm{IndLaySp} \rightarrow \mathrm{ILGLHyp}$ is well defined.
\end{lem}

\begin{proof} 
First we verify objects. Given an indexed layered space $\mathcal{R}$,  $\mathfrak{G}(\mathcal{R})$ is clearly a functor of the right sort, with ILGL duality giving that each $\mathfrak{G}(\mathcal{R})(X)$ is a layered Heyting algebra and each $\mathfrak{G}(\mathfrak{R})(f)$ is a layered Heyting algebra homomorphism. By assumption, for each object $X$, $Ran(\mathcal{R}(\Delta_X))$ is clopen and by Lemma \ref{lemma:comlexisilf} it is also upwards closed. Hence 
$Ran(\mathcal{R}(\Delta_X)) \in \mathfrak{G}(\mathcal{R})(X \times X)$, as required. We also have that
$\mathcal{R}(\pi_{\Gamma, X})^*$ and $\mathcal{R}(\pi_{\Gamma, X})_*$ map upward-closed clopen sets to upward-closed clopen sets by definition. Hence they are monotone maps
$\mathfrak{G}(\mathcal{R})(\Gamma \times X) \rightarrow \mathfrak{G}(\mathcal{R})(\Gamma)$, as required. The adjointness and naturality properties then follow from Lemma \ref{lemma:comlexisilf}.

We now check morphisms. Let $(L, \lambda)$ be an indexed resource space morphism. Then as in the previous lemma, $L$ is a finite product preserving functor by assumption and $\lambda^{-1}$ inherits naturality. By ILGL duality, since each component of $\lambda$ is an intuitionistic layered frame morphism, each component of $\lambda^{-1}$ is layered Heyting algebra homomorphism. Finally, preservation of $=_X$ follows from (Lift Property) and commutativity of the adjoints follows from (Morphism Psuedo Epi): the proofs are similar to that given in Lemma \ref{lemma:comlexisilf}.
\end{proof}

\begin{thm}
The categories $\mathrm{ILGLHyp}$ and $\mathrm{IndLaySp}$ are dually equivalent.
\end{thm}

\begin{proof}
\begin{sloppypar}We obtain the required natural transformations from ILGL duality: 
the maps $(Id_{\mathrm{ILGLHyp}}, h_{\mathbb{P}(-)}): \mathbb{P} \rightarrow \mathfrak{GF}(\mathbb{P})$ and $(Id_{\mathrm{IndLaySp}}, f_{\mathcal{R}(-)}): \mathcal{R} \rightarrow \mathfrak{FG}(\mathcal{R})$ define the components at $\mathbb{P}$ and $\mathcal{R}$, respectively. It remains to verify that each is a morphism in the respective category, and that they indeed form natural isomorphisms.
\end{sloppypar}

First, that $(Id_{\mathrm{ILGLHyp}}, h_{\mathbb{P}(-)})$ is an ILGL hyperdoctrine morphism. Clearly, it is a functor and a natural transformation of the right kind. Further, we know 
$h_{\mathbb{P}(X \times X)}(=_X) = Ran(\mathbb{P}(\Delta_X)^{-1})$, satisfying $=_X$ preservation. That $h_{\mathbb{P}(-)}$ commutes with the adjoints is shown by a similar argument to that given in Lemma \ref{lem:ipfsatisfiability}.

For the case for $(Id_{\mathrm{IndLaySp}}, f_{\mathcal{R}(-)})$, we once again have that the functor and natural transformation are of the right sort. To see that (Lift Property) holds, let $x \in \mathcal{R}(X \times X)$ and $F \in \mathfrak{FG}(X)$ with $\mathfrak{FG}(\Delta_X)(F) \subseteq f_{\mathcal{R}(X \times X)}(x)$. Then for any upward-closed clopen set $C$, if $\mathcal{R}(\Delta_X)^{-1}(C) \in F$ then $x \in C$. By Esakia duality, there exists $z$ such that $f_{\mathcal{R}(X)}(z) = F$. Suppose
$\mathcal{R}(\Delta_X)(z) \not\preccurlyeq x$. Then by the Priestley separation axiom there exists an upward-closed clopen set $C$ such that $\mathcal{R}(\Delta_X)(z) \in 
C$ and $x \not\in C$: but then $z \in \mathcal{R}(\Delta_X)^{-1}(C)$ so $\mathcal{R}(\Delta_X)^{-1}(C) \in F$ and $x \in C$. It must thus be the case that $\mathcal{R}(\Delta_X)(z) \preccurlyeq x$. (Morphism Psuedo Epi) is shown in much the same way.

Finally, that these components form natural isomorphisms follows immediately from the fact that $h_{(-)}$ and $f_{(-)}$ are natural isomorphisms by ILGL duality.
\end{proof}

\section{Conclusions \& Further Work} \label{sec:conclusions}

In this paper we have given a comprehensive metatheoretic treatment of an intuitionistic substructural logic for reasoning about layered graphs, ILGL. Most notably, we give a semantics on layered graphs that is a sound and complete with respect to a labelled tableaux system with countermodel extraction. An equivalent algebraic and relational semantics is also given, and can be used to bridge the gap between the labelled tableaux system and sequential systems -- including a display calculus with cut elimination -- that directly present an algebraic semantics. This is done via a representation theorem that can be extended to a full topological duality. This metatheory is also provided for the natural predicate logic extension of ILGL. Finally, decidability of propositional ILGL is proved algebraically by establishing the finite embeddability property for the logic's algebraic semantics. 

A first direction for further work is real systems modelling and verification with ILGL. A good first step would be the class of bigraph models, given ILGL appears to be suited to reasoning about them. A thorough investigation of a suitable signature of predicate ILGL for use as an assertion language for bigraph models would be an appropriate starting point, with an extension to a Separation Logic-style language incorporating a dynamics corresponding to Bigraphical Reactive Systems as an end goal. This in itself would represent a rich program of research. Reformulating the proof theoretic properties that makes Separation Logic scalable -- for example, the frame rule and bi-abduction \cite{bi-abduction} -- in this setting would also make a strong contribution to the theory of application of bigraphs.

Further logical work can be done as well. As previously discussed, layered graph logics are the bunched logics with the weakest multiplicatives. It is thus possible to reconstruct the other systems in the literature by adding appropriate structure to the metatheory of the layered graph logics. We aim to use these logics as a foundation for a systematic treatment of the entire hierarchy of bunched logics, including a uniform construction of semantics and proof theory. In this direction we have already proved representation and duality theorems for the breadth of the bunched logics in the literature, obtaining systematic soundness and completeness theorems for each logic as a corollary. Such a treatment can also guide the design of Henessey-Millner-van Bentham style logics of state for location-resource-processes \cite{CMP14,CMP15, DP16} that incorporate layering.




\begin{thebibliography}{10}

\bibitem{AP16}
G. Anderson and D. Pym.
\newblock A calculus and logic of  bunched resources and processes.
\newblock {\em Theor. Comp. Sci.} 614:63-96, 2016.

\bibitem{BdeJ2006}
N. Bezhanishvili and D. de~Jongh.
\newblock Intuitionistic logic.
\newblock Technical Report PP-2006-25, Institute for Logic, Language and
  Computation, Universiteit van Amsterdam, 2006.

\bibitem{Biering2005}
B.Biering, L. Birkedal, and N. Torp-Smith.
\newblock BI hyperdoctrines and higher-order separation logic.
\newblock In {\em Proc. 14th ESOP}, 233--247, Springer-Verlag, 2005.

\bibitem{Dunn}
 K. Bimb\'o and J. M. Dunn.
\newblock {\em Generalized Galois Logics: Relational Semantics of Non-classical Calculi}.
\newblock CSLI Publications, 2008.

\bibitem{brodka}
P.~Br\'{o}dka, K.~Skibicki, P.~Kazienko, and K.~Musia\l{}.
\newblock A degree centrality in multi-layered social network.
\newblock In {\em Proc. CASoN '11}, 237--242, 2011.

\bibitem{bi-abduction}
C.~Calcagno, D.~Distefano, P. O'Hearn, and H.~Yang.
\newblock Compositional shape analysis by means of bi-abduction.
\newblock {\em J. ACM}, 58(6), 2011.

\bibitem{CGG02}
L. Cardelli, P. Gardner, G. Ghelli.
\newblock A spatial logic for querying graphs.
\newblock In \emph{Proc ICALP '02}, LNCS 2380, 597--610, 2002.

\bibitem{clark}
D. D. Clark.
\newblock The design philosophy of the DARPA internet protocols.
\newblock In {\em Proc. SIGCOMM '88}, {Computer Communication Review}, 18(4): 106--114, 1988.

\bibitem{CMP14}
M. Collinson, K. McDonald, and D. Pym.
\newblock A substructural logic for layered graphs.
\newblock {\em J. Log. Comp.}, 24(4):953--988, 2014.

\bibitem{CMP15}
M. Collinson, K. McDonald, and D. Pym.
\newblock Layered graph logic as an assertion language for access control
  policy models.
\newblock {\em J. Log. Comp.}, 27(1):41--80 2017. 
 

\bibitem{CMP2012}
M.~Collinson, B.~Monahan, and D.~Pym.
\newblock {\em A Discipline of Mathematical Systems Modelling}.
\newblock College Publications, 2012.

\bibitem{Pym09a}
M.~Collinson and D.~Pym.
\newblock Algebra and logic for resource-based systems modelling.
\newblock {\em Math. Struc. Comp. Sci.}, 19(5):959--1027,
  2009.

\bibitem{bilog}
G. Conforti, D. Macedonio, and V. Sassone. 
\newblock Spatial logics for bigraphs. 
\newblock In {\em Proc. ICAP '05}, LNCS 3580, 766--778, 2005.

\bibitem{COUMANS20121940}
D. Coumans.
\newblock Generalising canonical extension to the categorical setting.
\newblock {\em Ann. Pure. Appl. Log.}, 163(12):1940--1961, 2012.

\bibitem{Courtault-Galmiche15}
J.-R. Courtault and D. Galmiche. 
\newblock A modal separation logic for resource dynamics. 
\emph{Journal of Logic and Computation}, doi:10.1093/logcom/exv031, 2015.   


\bibitem{DP16}
S. Docherty and D. Pym.
\newblock Intuitionistic layered graph logic.
\newblock \!\emph{Proc.\! IJCAR 2016}.\! LNAI 9706:\,469--486, 2016.

\bibitem{DH01}
J. M. Dunn and G. Hardegree.
\newblock {\em Algebraic Methods In Philosophical Logic}
\newblock OUP, 2001.

\bibitem{Esakia}
L. Esakia.
\newblock Topological Kripke models.
\newblock {\em Soviet Math. Dokl.} 15, 147--15, 1974.

\bibitem{fitting1972}
M. Fitting.
\newblock Tableau methods of proof for modal logics.
\newblock {\em Notre Dame J. Fom. Log.}, 13(2):237--247, 1972.

\bibitem{GJ2017}
N. Galatos and P. Jipsen.
\newblock Distributive residuated frames and generalized bunched implication algebras.
\newblock {\em Algebra Univers.}, 78(3):303--336, 2017.

\bibitem{GM10}
D. Galmiche and D. M{\'e}ry.
\newblock Tableaux and resource graphs for separation logic.
\newblock {\em J. Log. Comp.}, 20(1): 189-231, 2010.

\bibitem{GMP05}
D. Galmiche, D. M{\'e}ry, and D. Pym.
\newblock {The semantics of BI and resource tableaux}.
\newblock {\em Math. Str. Comp. Sci.}, 15(06):1033--1088,
  2005.

\bibitem{Girard}
J.-Y. Girard.
\newblock Linear logic.
\newblock {\em Theor. Comp. Sci.}, 50(1): 1--101, 1987.

\bibitem{Gouveia}
L. Gouveia, L. Simonetti, and E. Uchoa.
\newblock Modeling hop-constrained and diameter-constrained minimum spanning tree problems as Steiner tree problems over layered graphs.
\newblock \emph{Math. Prog.}, 128(1): 123--148, 2011.

\bibitem{GM07}
D. Grohmann and M. Miculan.
\newblock Directed bigraphs.
\newblock In {\em Proc. MFPS XXIII}, ENTCS 173, 121--137, 2007. 

\bibitem{FEP}
Z. Hanikov\'a and R. Hor\u{c}\'{i}k.
\newblock The finite embeddability property for residuated
groupoids
\newblock {\em Algebra Univers.}, 72(1):1--13, 2014.

\bibitem{HP2003}
J. Harland and D. Pym.
\newblock Resource-distribution via Boolean constraints. 
\newblock \emph{ACM ToCL}, 4(1):56--90, 2003.  

\bibitem{IO01}
S. S. Ishtiaq and P. O'Hearn.
\newblock {BI} as an assertion language for mutable data structures.
\newblock In {\em Proc. POPL '01}, ACM Sigplan Notices 36(3):14--26, 2001.

\bibitem{JL2018}
P. Jipsen and T. Litak.
\newblock An algebraic glimpse at bunched implications and separation
logic.
\newblock In {\em Outstanding Contributions: Hiroakira Ono on
Residuated Lattices and Substructural Logics}, arXiv:1709.07063v2, to
appear.

\bibitem{Jipsen}
P. Jipsen and C. Tsinakis.
\newblock A survey of residuated lattices.
\newblock In \emph{Ordered Algebraic Structures}, Developments in Mathematics 7:19-56, 2002.


\bibitem{PTJ82}
P. T. Johnstone
\newblock {\em Stone Spaces}.
\newblock Cambridge Studies In Advanced Mathematics 3, CUP, 1982.

\bibitem{KAB14}
M. Kivel\"{a},  A. Arenas,  M. Barthelemy, J. P. Gleeson, Y. Moreno and M. A. Porter.
\newblock Multilayer networks.
\newblock {\em J. Comp. Net.}, 2(3): 203--271, 2014.

\bibitem{Kripke}
S. A. Kripke.
\newblock A semantical analysis of intuitionistic logic {I}.
\newblock In {\em Formal Systems and Recursive Functions}, Studies In Logic and the Foundations of Mathematics 40:92--130, 1965.

\bibitem{PhysRevLett.96.138701}
M. Kurant and P. Thiran.
\newblock Layered complex networks.
\newblock {\em Phys. Rev. Lett.}, 96:138701, 2006.

\bibitem{Lambek1961}
J.~Lambek.
\newblock On the calculus of syntactic types.
\newblock In {\em Studies of Language and its Mathematical Aspects}, 166--178, 1961.

\bibitem{Lambek1993}
J.~Lambek.
\newblock From categorical grammar to bilinear logic.
\newblock In P.~Schroeder-Heister and K.~Do\v{s}en, editors, {\em Substructural
  Logics}, 207--237. OUP, 1993.

\bibitem{Lar14a}
D.~Larchey-Wendling.
\newblock The formal proof of the strong completeness of partial
  monoidal {B}oolean {BI}.
\newblock {\em J. Log. Comp.}, 26(2):605--640, 2016.

\bibitem{maus}
C. Maus, S. Rybacki, and A. M. Uhrmacher.
\newblock Rule-based multi-level modeling of cell biological systems
\newblock {\em BMC Sys. Bio.}, 5(166), doi:10.1186/1752-0509-5-166, 2011.

\bibitem{milner}
R. Milner.
\newblock {\em The Space and Motion of Communicating Agents}.
\newblock CUP, 2009.

\bibitem{OP99}
P. O'Hearn and D. Pym.
\newblock The logic of bunched implications.
\newblock {\em Bull. Symb. Log.}, 5(2):215--244, 1999.



\bibitem{POY}
D. Pym, P. O'Hearn, and H.~Yang.
\newblock Possible worlds and resources: The semantics of {BI}.
\newblock {\em Theor. Comp. Sci.}, 315(1):257--305, 2004.

\bibitem{Rey02}
J.C. Reynolds.
\newblock Separation logic: a logic for shared mutable data structures.
\newblock In {\em Proc LICS '02}, IEEE Comp. Soc. Press,  55--74 2002.



\bibitem{Stone}
M.H. Stone.
\newblock The theory of representations of Boolean algebras.
\newblock {\em Trans. AMS} 40: 37--111, 1936.

\bibitem{Wang}
H. Wang, J. Wang and P. De Wilde.
\newblock Topological analysis of a two coupled evolving networks model for business systems.
\newblock {\em Expert Syst. Appl.}, 36(5):9548--9556, 2009.

\end{thebibliography}
\end{document}